\documentclass[reqno]{amsart}
\usepackage{fullpage,amsfonts,amssymb,amsthm,mathrsfs,graphicx,color,framed,esint,stackrel,amsmath}
\usepackage{cancel,bm}
\usepackage{tikz}
\usetikzlibrary{decorations.markings}
\usepackage{booktabs,longtable}
\usepackage[initials]{amsrefs}

\DeclareMathAlphabet\mathbfcal{OMS}{cmsy}{b}{n}

\usepackage[colorlinks=true, pdfstartview=FitV, linkcolor=blue, citecolor=blue, urlcolor=blue]{hyperref}

\renewcommand{\d}{{\mathrm d}}

\newcommand{\im}{\mathrm{i}}
\newcommand{\e}{\mathrm{e}}

\def\res{\mathop{\mathrm{res}}\limits}
\def\tr{\mathop{\mathrm{tr}}\limits}

\newtheorem{theo}{Theorem}[section]

\newtheorem{lem}[theo]{Lemma}
\newtheorem{rem}[theo]{Remark}
\newtheorem{problem}[theo]{Riemann-Hilbert Problem}

\newtheorem{prop}[theo]{Proposition} 
\newtheorem{cor}[theo]{Corollary}

\newtheorem{assum}[theo]{Assumption}

\usepackage{mdframed}

\begin{document}
\title[Painlev\'e-II with parameter]{Universality for random matrices with an edge spectrum singularity}

\author{Thomas Bothner}
\address{School of Mathematics, University of Bristol, Fry Building, Woodland Road, Bristol, BS8 1UG, United Kingdom}
\email[Corresponding author]{thomas.bothner@bristol.ac.uk}

\author{Toby Shepherd}
\address{School of Mathematics, University of Bristol, Fry Building, Woodland Road, Bristol, BS8 1UG, United Kingdom}
\email{toby.shepherd@bristol.ac.uk}

\date{\today}


%\dedicatory{For Arno B.J. Kuijlaars on the occasion of his 60th birthday}


%\thanks{T.B. is grateful to Peter Forrester for drawing his attention to the problem at hand. This work is dedicated, with admiration, to Arno Kuijlaars on the occasion of his $60$th birthday. We are deeply indebted to Arno for his trail-blazing contributions to the modern theory of Riemann-Hilbert problems and their asymptotic analysis over the past $25$ years.}

\begin{abstract}
We study invariant random matrix ensembles
\begin{equation*}
	\mathbb{P}_n(\d M)=Z_n^{-1}\exp(-n\,\textnormal{tr}\,V(M))\,\d M
\end{equation*}
defined on complex Hermitian matrices $M$ of size $n\times n$, where $V$ is real analytic such that the underlying density of states is one-cut regular. Considering the average 
\begin{equation*}
	E_n[\phi;\lambda,\alpha,\beta]:=\mathbb{E}_n\bigg(\prod_{\ell=1}^n\big(1-\phi(\lambda_{\ell}(M))\big)\omega_{\alpha\beta}(\lambda_{\ell}(M)-\lambda)\bigg),\ \ \ \ \ \omega_{\alpha\beta}(x):=|x|^{\alpha}\begin{cases}1,&x<0\\ \beta,&x\geq 0\end{cases},
\end{equation*}
taken with respect to the above law and where $\phi$ is a suitable test function, we evaluate its large-$n$ asymptotic assuming that $\lambda$ lies within the soft edge boundary layer, and $(\alpha,\beta)\in\mathbb{R}\times\mathbb{C}$ satisfy $\alpha>-1,\beta\notin(-\infty,0)$. Our results are obtained by using Riemann-Hilbert problems for orthogonal polynomials and integrable operators and they extend previous results of Forrester and Witte \cite{FW} that were obtained by an application of Okamoto's $\tau$-function theory. A key role throughout is played by distinguished solutions to the Painlev\'e-XXXIV equation.

\end{abstract}

\maketitle
%%%%%%%%%%%%%%%%%%%%%%%%%%%%%%%%%%%%%%%%%%%%%%%%%%%%%%%%%%%%%%
\section{Introduction and statement of results} 
\subsection{The initial setting} For $n\in\mathbb{N}$, we consider the invariant ensemble
\begin{equation}\label{e:1}
	\mathbb{P}_n(\d M)=Z_n^{-1}\e^{-n\,\textnormal{tr}\,V(M)}\d M,\ \ \ \ \d M:=\prod_{j=1}^n\d M_{jj}\prod_{1\leq j<k\leq n}\d\Re M_{jk}\,\d\Im M_{jk},
\end{equation}
on the space of $n\times n$ complex Hermitian matrices $M=[M_{jk}]_{j,k=1}^n$, where $V:\mathbb{R}\rightarrow\mathbb{R}_{\geq 0}$ is a real analytic function such that
\begin{equation*}
	V(x)\geq(2+\epsilon)\ln(1+|x|),\ \ \ \ \ x\in\mathbb{R},
\end{equation*}
for some $\epsilon>0$ to guarantee the integrability of the measure \eqref{e:1} with normalization $Z_n>0$. It is well-known, cf. \cite{PS}, that the eigenvalues 
\begin{equation*}
	-\infty<\lambda_1(M)<\ldots<\lambda_n(M)<\infty,
\end{equation*}
 of $M$ are distributed according to a determinantal point process and as such key statistical quantities of the model are computable in terms of orthogonal polynomial data. For instance, the \textit{generating functional} of \eqref{e:1}, i.e. the quantity 
\begin{equation}\label{e:2}
	E_n[\phi]:=\mathbb{E}_n\bigg(\prod_{\ell=1}^n\Big(1-\phi\big(\lambda_{\ell}(M)\big)\Big)\bigg),
\end{equation}
defined for a suitable test function $\phi:\mathbb{R}\rightarrow\mathbb{C}$, admits a Fredholm determinant representation in terms of the reproducing kernel $K_n(x,y)$ of the family of polynomials that are orthonormal on $\mathbb{R}$ with respect to the weight $\e^{-nV(x)}$. Remarkably, scaling limits of \eqref{e:2} are determined to a large extent by the limiting mean density $\rho_V$ of eigenvalues,
\begin{equation}\label{e:3}
	\rho_V(x):=\lim_{n\rightarrow\infty}\frac{1}{n}K_n(x,x).
\end{equation}
Indeed, see \cite{DKMVZ}, near endpoints $x_0$ of the limiting spectrum where $\rho_V$ vanishes like a square-root, the appropriately centered and rescaled reproducing kernel $K_n$ converges pointwise to the ubiquitous Airy kernel $K_{\textnormal{Ai}}$. In turn, the soft-edge scaled functional \eqref{e:2} admits the limiting form
\begin{equation*}
	\lim_{n\rightarrow\infty}E_n[\phi_n]=1+\sum_{\ell=1}^{\infty}\frac{(-1)^{\ell}}{\ell!}\int_{J^{\ell}}\det\big[K_{\textnormal{Ai}}(x_j,x_k)\big]_{j,k=1}^{\ell}\prod_{m=1}^{\ell}\phi(x_m)\d x_m,\ \ \ \phi_n(x):=\phi\big((x-x_0)cn^{\frac{2}{3}}\big),
\end{equation*}
with some $c>0$ and for any bounded $\phi$ with $J=\textnormal{supp}(\phi)\subset (a,b)$, $-\infty<a<b\leq\infty$. Henceforth the limiting distribution function of the maximum eigenvalue $\lambda_n(M)$ in the model \eqref{e:1} becomes famously computable in terms of Painlev\'e functions. For instance, we have the Tracy-Widom formula \cite{TW},
\begin{equation*}
	F_2(s):=\lim_{n\rightarrow\infty}\mathbb{P}_n\bigg(\lambda_n(M)\leq x_0+\frac{s}{cn^{\frac{2}{3}}}\bigg)=\exp\left(-\int_s^{\infty}\sigma_0(x)\d x\right),\ \ \ s\in\mathbb{R},
\end{equation*}
where $\sigma_0(x)$ satisfies the Jimbo-Miwa-Okamoto $\sigma$-form of the Painlev\'e-II equation,
\begin{equation}\label{e:4}
	(\sigma_{\alpha}'')^2+4\sigma_{\alpha}'\Big((\sigma_{\alpha}')^2-t\sigma_{\alpha}'+\sigma_{\alpha}\Big)=\alpha^2;\hspace{1.5cm} \sigma_{\alpha}=\sigma_{\alpha}(t),%\ \ \ \sigma_{\alpha}'=\frac{\d\sigma_{\alpha}}{\d t}(t),\ \ \ \sigma_{\alpha}''=\frac{\d^2\sigma_{\alpha}}{\d t^2}(t),
\end{equation}
with parameter $\alpha=0$ and boundary constraint $\sigma_0(t)\sim(\textnormal{Ai}'(t))^2-t(\textnormal{Ai}(t))^2$ as $t\rightarrow+\infty$. Another representation of $F_2(s)$, the Forrester-Witte formula \cite{FW}, is
\begin{equation*}
	F_2(s)=\exp\left(\int_s^{\infty}(x-s)q_0(x)\d x\right),\ \ \ s\in\mathbb{R},
\end{equation*}
in terms of $q_0(x)$ that solves the Painlev\'e-XXXIV equation,
\begin{equation}\label{e:5}
	q_{\alpha}''=\frac{(q_{\alpha}')^2}{2q_{\alpha}}+2tq_{\alpha}-4q_{\alpha}^2-\frac{\alpha^2}{2q_{\alpha}};\hspace{1.5cm} q_{\alpha}=q_{\alpha}(t),%\ \ \ q_{\alpha}'=\frac{\d q_{\alpha}}{\d t}(t),\ \ \ q_{\alpha}''=\frac{\d^2 q_{\alpha}}{\d t^2}(t),
\end{equation}
with parameter $\alpha=0$ and boundary condition $q_0(t)\sim-(\textnormal{Ai}(t))^2$ as $t\rightarrow+\infty$. Lastly, and likely the most well-known representation of $F_2(s)$, found also by Tracy and Widom \cite{TW},
\begin{equation*}
	F_2(s)=\exp\left(-\int_s^{\infty}(x-s)u_0^2(x)\d x\right),\ \ \ s\in\mathbb{R},
\end{equation*}
which is expressed in terms of $u_0(x)$ that solves the Painlev\'e-II equation,
\begin{equation}\label{e:6}
	u_{\alpha}''=tu_{\alpha}+2u_{\alpha}^3-\alpha;\hspace{1.5cm} u_{\alpha}=u_{\alpha}(t),%\ \ \ u_{\alpha}'=\frac{\d u_{\alpha}}{\d t}(t),\ \ \ u_{\alpha}''=\frac{\d^2 u_{\alpha}}{\d t^2}(t),
\end{equation}
with boundary behavior $u_0(t)\sim\textnormal{Ai}(t)$ as $t\rightarrow+\infty$. Throughout, $\textnormal{Ai}(x)$ denotes the Airy function \cite[$9.2.2$]{NIST} and $(')$ differentiation with respect to the independent variable. 
\begin{rem}
The above Painlev\'e representations for $F_2(s)$ are consistent, for $\sigma_{\alpha}'(x)=q_{\alpha}(x)$ and because of the curious identities
\begin{equation*}
	q_{\alpha}(x)=-2^{-\frac{1}{3}}\Big(\frac{y}{2}+u_{\alpha+\frac{1}{2}}^2(y)+u_{\alpha+\frac{1}{2}}'(y)\Big)\bigg|_{y=-2^{\frac{1}{3}}x},\ \ \ \ \ \ u_0^2(-2^{-\frac{1}{3}}x)=2^{-\frac{1}{3}}\Big(\frac{x}{2}+u_{\frac{1}{2}}^2(x)+u_{\frac{1}{2}}'(x)\Big),
\end{equation*}
that are known from the works of Okamoto \cite{Ok} and Gambier \cite{G}. 
\end{rem}
In this paper we introduce a singularity in the right hand side of \eqref{e:2} and focus on the more general quantity
\begin{equation}\label{e:7}
	E_n[\phi;\lambda,\alpha,\beta;V(x)]:=\mathbb{E}_n\bigg(\prod_{j=1}^n\Big(1-\phi\big(\lambda_j(M)\big)\Big)\omega_{\alpha\beta}\big(\lambda_j(M)-\lambda\big)\bigg),\ \ \ \ \ \textnormal{supp}(\phi)=J\subset\mathbb{R},\ \ \ \lambda\in\mathbb{R},
\end{equation}
which depends on $(\alpha,\beta)\in\mathbb{R}\times\mathbb{C}$, with $\alpha>-1,\beta\notin(-\infty,0)$, through the Fisher-Hartwig factor
\begin{equation*}
	\omega_{\alpha\beta}(x):=|x|^{\alpha}\begin{cases}1,&x<0\\ \beta,&x\geq 0\end{cases}.
\end{equation*}
The average in \eqref{e:7} is once more taken with respect to \eqref{e:1} and it is our goal to study \eqref{e:7} in the limit for large $n$ when $\lambda$ lies in the boundary layer of the limiting mean density $\rho_V$ of eigenvalues, for a suitable class of potentials $V$ and test functions $\phi$. Earlier studies by Forrester and Witte \cite{FW} considered a similar problem in the Gaussian setting $V(x)=x^2$ for $\beta=1$, using Okamoto's $\tau$-function theory techniques \cite{Ok} for Painlev\'e-IV and Painlev\'e-II. Here we proceed in a purely analytic fashion, utilizing Riemann-Hilbert problems (RHPs).
\subsection{Precise assumptions} Orthogonal polynomial techniques allow us to write \eqref{e:7} in the factorized form\footnote{We refer the interested reader to Appendix \ref{facderiv} for a short derivation of the same formula.}
\begin{equation}\label{e:8}
	E_n\big[\phi;\lambda,\alpha,\beta;V(x)\big]=\frac{D_n(\lambda,\alpha,\beta;V(x))}{D_n(\lambda,0,1;V(x))}F_n\big[\phi;\lambda,\alpha,\beta;V(x)\big],
\end{equation}
where $F_n$ is the Fredholm determinant on $L^2(J)$ of identity minus $K_{n,\alpha,\beta}^{\lambda}M_{\phi}$, i.e.
\begin{equation*}
	F_n\big[\phi;\lambda,\alpha,\beta;V(x)\big]:=1+\sum_{\ell=1}^n\frac{(-1)^{\ell}}{\ell!}\int_{J^{\ell}}\det\Big[K_{n,\alpha,\beta}^{\lambda}(x_j,x_k)\Big]_{j,k=1}^{\ell}\prod_{m=1}^{\ell}\phi(x_m)\d x_m,
\end{equation*}
and $D_n$ the Hankel determinant
\begin{equation*}
	D_n\big(\lambda,\alpha,\beta;V(x)\big):=\frac{1}{n!}\int_{\mathbb{R}^n}\prod_{1\leq j<k\leq n}|x_k-x_j|^2\prod_{\ell=1}^n\omega_{\alpha\beta}(x_{\ell}-\lambda)\e^{-nV(x_{\ell})}\d x_{\ell}.
\end{equation*}
Note that $M_{\phi}$ multiplies by $\phi$ and $K_{n,\alpha,\beta}^{\lambda}$ has kernel, with the principal branch for the square-root,
\begin{equation}\label{e:9}
	K_{n,\alpha,\beta}^{\lambda}(x,y):=\omega_{\alpha\beta}^{\frac{1}{2}}(x-\lambda)\e^{-\frac{n}{2}V(x)}\omega_{\alpha\beta}^{\frac{1}{2}}(y-\lambda)\e^{-\frac{n}{2}V(y)}\sum_{\ell=0}^{n-1}\pi_{\ell,n}(x)\pi_{\ell,n}(y),
\end{equation}
defined in terms of the orthonormal polynomials $\{\pi_{j,n}\}_{j=0}^{\infty}\subset\mathbb{C}[x]$ that obey
\begin{equation*}
	\int_{-\infty}^{\infty}\pi_{j,n}(x)\pi_{k,n}(x)\omega_{\alpha\beta}(x-\lambda)\e^{-nV(x)}\d x=\begin{cases}1,&j=k\\ 0,&j\neq k\end{cases},\ \ \ \ \ \ \ \ \ \textnormal{deg}(\pi_{j,n})=j.
\end{equation*}
\begin{rem} While $\pi_{j,n}$ might not exist for all $j\in\mathbb{Z}_{\geq 0}$ since $\omega_{\alpha\beta}$ is in general complex-valued, $\pi_{n,n}$ and $\pi_{n-1,n}$ do exist for large $n$ and $\lambda=\lambda_n$ as in \eqref{e:12} below. In turn, identity \eqref{e:8} at $\lambda=\lambda_n$ is bona fide for sufficiently large $n$ and we will be able to calculate edge asymptotics for $E_n$ from it.
\end{rem}
Although $K_{n,\alpha,\beta}^{\lambda}(\lambda,\lambda)=0$ if $\alpha>0$ and $K_{n,\alpha,\beta}^{\lambda}(\lambda,\lambda)=+\infty$ if $\alpha<0$, the quantity $\frac{1}{n}K_{n,\alpha,\beta}^{\lambda}(x,x)$ does not feel the presence of $(\lambda,\alpha,\beta)$ away from $x=\lambda$ for large $n$. In fact, we have the pointwise limit
\begin{equation*}
	\rho_V(x)\stackrel{\eqref{e:3}}{=}\lim_{n\rightarrow\infty}\frac{1}{n}K_n(x,x)=\lim_{n\rightarrow\infty}\frac{1}{n}K_{n,\alpha,\beta}^{\lambda}(x,x),\ \ \ \ x\in\mathbb{R}\setminus\{\lambda\}.
\end{equation*}
However, the presence of $(\lambda,\alpha,\beta)$ is certainly felt on the level of the local scaling limits of $K_{n,\alpha,\beta}^{\lambda}$, in particular when $\lambda$ is chosen near endpoints $x_0$ of the support of $\rho_V$, as done in a few moments. To make precise that particular boundary approach, we shall first recall a few standard facts, see \cite{ST}, about $\rho_V$ and then put in place our working assumptions on $V$. Namely, $\rho_V$ is the density of the unique minimizer $\mu_V$ of the energy functional
\begin{equation*}
	I_V[\mu]:=\iint\ln\frac{1}{|x-y|}\d\mu(x)\d\mu(y)+\int V(x)\d\mu(x),
\end{equation*}
when optimized over all Borel probability measures $\mu$ on $\mathbb{R}$. For real-analytic $V$, the same minimizer is supported on a finite union of disjoint intervals, it has density $\rho_V$ and the minimizer is completely determined by the following Euler-Lagrange conditions: there exists $\ell_V\in\mathbb{R}$ such that
\begin{align}
	2\int\ln|x-y|\rho_V(y)\d y-V(x)+\ell_V=&\,\,0,\ \ \ x\in E;\nonumber\\
	2\int\ln|x-y|\rho_V(y)\d y-V(x)+\ell_V\leq&\,\,0,\ \ \ x\in\mathbb{R}\setminus E;\label{e:10}
\end{align}
with $E:=\textnormal{supp}(\mu_V)$. We will work with real analytic potentials whose equilibrium measure shares properties of the equilibrium measure of the quadratic $V(x)=x^2$. Specifically, we put in place the following:
\begin{assum}\label{1-cut} Let $V:\mathbb{R}\rightarrow\mathbb{R}_{\geq 0}$ be real analytic with $V(x)\geq (2+\epsilon)\ln(1+|x|)$ on $\mathbb{R}$ for some $\epsilon>0$. We assume that $V$ is one-cut regular, i.e.
\begin{enumerate}
	\item[(i)] that the support of its equilibrium measure $\mu_V$ is $E=[-\sqrt{2},\sqrt{2}\,]\subset\mathbb{R}$.
	\item[(ii)] that the density of $\mu_V$ is strictly positive in the interior of $E$ and vanishes like a square root at $\pm\sqrt{2}$.
	\item[(iii)] that the inequality \eqref{e:10} is strict.
\end{enumerate}
\end{assum}
Behaviours (i),(ii) and (iii), modulo the deliberate choice of the support interval endpoints, are generic for an equilibrium measure $\mu_V$ of a real analytic $V$, cf. \cite{KM}, and they allow us to express the density $\rho_V$ of $\mu_V$, for such a one-cut regular potential $V$, as
\begin{equation}\label{e:11}
	\rho_V(x)=h_V(x)\sqrt{(2-x^2)_+},\ \ \ \ \ x_+:=\max\{x,0\},
\end{equation}
with $h_V$ strictly positive on $E$, and real analytic on $\mathbb{R}$, compare \cite{DKMVZ}. 
\begin{rem} The quadratic $V(x)=x^2$ is one-cut regular with density
\begin{equation*}
	\rho_V(x)=\frac{1}{\pi}\sqrt{(2-x^2)_+},
\end{equation*}
given by the Wigner semicircle density and with $\ell_V=1+\ln 2$. See for instance, \cite[$(7.83),(7.90)$]{D}.
\end{rem}
With Assumption \ref{1-cut} in place we choose $\lambda$ in \eqref{e:7} in the right edge boundary layer by making it $n$-dependent through the rule
\begin{equation}\label{e:12}
	\lambda\rightarrow\lambda_n:=\sqrt{2}+\frac{s}{(n\tau)^{\frac{2}{3}}},\ \  s\in\mathbb{R};\ \ \ \ \ \ \ \ \ \tau:=\pi h_V(\sqrt{2})2^{\frac{3}{4}}>0.
\end{equation}
What results for \eqref{e:8} is summarized in the following subsection.
\subsection{List of main results} We begin with a limit theorem for the edge-scaled Fredholm determinant $F_n$ in \eqref{e:8} and its connection to \eqref{e:4}, resp. \eqref{e:5}.
\begin{theo}\label{BS:1} Let $\phi:\mathbb{R}\rightarrow\mathbb{C}$ be bounded with support $J=\textnormal{supp}(\phi)\subset (a,b)$ where $-\infty<a<b\leq\infty$. Set
\begin{equation}\label{e:13}
	\phi_n(x):=\phi\Big((x-\sqrt{2})(n\tau)^{\frac{2}{3}}\Big),\ \ \ n\in\mathbb{N}, 
\end{equation}
with $\tau$ as in \eqref{e:12} and recall $F_n,\lambda_n$ from \eqref{e:8},\eqref{e:12}. Then, for arbitrary $V$ as in Assumption \ref{1-cut}, we obtain the limiting Fredholm determinant formula
\begin{equation}\label{e:14}
	\lim_{n\rightarrow\infty}F_n\big[\phi_n;\lambda_n,\alpha,\beta;V(x)\big]=1+\sum_{\ell=1}^{\infty}\frac{(-1)^{\ell}}{\ell!}\int_{J^{\ell}}\det\Big[A_s^{\alpha\beta}(x_j-s,x_k-s)\Big]_{j,k=1}^{\ell}\prod_{m=1}^{\ell}\phi(x_m)\d x_m,
\end{equation}
uniformly in $s\in\mathbb{R},\alpha>-1,\beta\in\mathbb{C}\setminus(-\infty,0)$ chosen on compact sets, with kernel
\begin{equation*}
	A_s^{\alpha\beta}(x,y):=\frac{\psi_2(x;s,\alpha,\beta)\psi_1(y;s,\alpha,\beta)-\psi_1(x;s,\alpha,\beta)\psi_2(y;s,\alpha,\beta)}{2\pi\im(x-y)},\ \ \ x,y\in\mathbb{R}\setminus\{0\},
\end{equation*}
defined in terms of the unique solution $Q(\zeta;x,\alpha,\beta)$ of the Painlev\'e-XXXIV RHP \ref{YattRHP} via the formula
\begin{equation}\label{e:15}
	\begin{bmatrix}\psi_1(z;x,\alpha,\beta)\smallskip\\ \psi_2(z;x,\alpha,\beta)\end{bmatrix}:=\begin{cases}Q_+(z;x,\alpha,\beta)\begin{bmatrix}1\\ 0\end{bmatrix},&z>0\\
	Q_+(z;x,\alpha,\beta)\e^{-\im\frac{\pi}{2}\alpha\sigma_3}\begin{bmatrix}1\\ 1\end{bmatrix},&z<0\end{cases}.
\end{equation}
Moreover, choosing $\phi(x)=1$ for $x>s$ and $\phi(x)=0$ for $x<s$, and denoting the right hand side of \eqref{e:14}, with that particular choice of $\phi$, as $F(s;\alpha,\beta)$, we have
\begin{equation}\label{e:16}
	F(s;\alpha,\beta)=\exp\bigg(-\int_s^{\infty}\big(\sigma_{\alpha}(x,\beta-1)-\sigma_{\alpha}(x,\beta)\big)\d x\bigg),\ \ \ \ s\in\mathbb{R},\ \ \alpha>-1,\ \ \beta\in\mathbb{C}\setminus(-\infty,1),
\end{equation}
where $\sigma_{\alpha}=\sigma_{\alpha}(t,\gamma)$ solves \eqref{e:4} with parameter $\alpha>-1$ and boundary constraint at $t=+\infty$ equal to
\begin{equation}\label{e:17}
	\sigma_{\alpha}(t,\gamma)=-\alpha\sqrt{t}\bigg[1+\frac{\alpha}{4t^{\frac{3}{2}}}+\mathcal{O}\big(t^{-3}\big)\bigg]+\big(\e^{\im\pi\alpha}-\gamma\big)\frac{\Gamma(1+\alpha)}{2^{3(1+\alpha)}\pi}t^{-1-\frac{3}{2}\alpha}\e^{-\frac{4}{3}t^{\frac{3}{2}}}\Big(1+\mathcal{O}\big(t^{-\frac{1}{4}}\big)\Big),\ \ \gamma\notin(-\infty,0).
\end{equation}
\end{theo}
Noticing the large argument asymptotics of the difference in the integrand of \eqref{e:16}, an application of Fubini's theorem together with Lemma \ref{ArnoLax} yields an equivalent formula for $F(s;\alpha,\beta)$. Namely,
\begin{equation}\label{e:18}
	F(s;\alpha,\beta)=\exp\bigg(\int_s^{\infty}(x-s)\big(q_{\alpha}(x,\beta-1)-q_{\alpha}(x,\beta)\big)\d x\bigg),\ \ s\in\mathbb{R},\ \ \beta\in\mathbb{C}\setminus(-\infty,1),
\end{equation}
in terms of $q_{\alpha}=q_{\alpha}(t,\gamma)$ that solves Painlev\'e-XXXIV \eqref{e:5} with parameter $\alpha>-1$ and constraint at $t=+\infty$
\begin{equation}\label{e:19}
	q_{\alpha}(t,\gamma)=-\frac{\alpha}{2\sqrt{t}}\bigg[1-\frac{\alpha}{2t^{\frac{3}{2}}}+\mathcal{O}\big(t^{-3}\big)\bigg]-\big(\e^{\im\pi\alpha}-\gamma\big)\frac{\Gamma(1+\alpha)}{2^{2+3\alpha}\pi}t^{-\frac{1}{2}(1+3\alpha)}\e^{-\frac{4}{3}t^{\frac{3}{2}}}\Big(1+\mathcal{O}\big(t^{-\frac{3}{2}}\big)\Big),\ \ \gamma\notin(-\infty,0).
\end{equation}
In general, the appearance of the difference of two \textit{different} solutions to the \textit{same} nonlinear Painlev\'e equation in the right hand sides of \eqref{e:16} and \eqref{e:18} strikes us as peculiar. We are not aware of any other Fredholm determinants in integrable systems theory that display a similar feature. Still, since $\alpha$ is in general non-zero, it is precisely the difference of two Painlev\'e transcendents that makes the integrands in \eqref{e:16} and \eqref{e:18} integrable at $x=+\infty$, compare \eqref{e:17} and \eqref{e:19}. 
\begin{rem} When $(\alpha,\beta)=(0,1)$, then $q_{\alpha}(x,\beta)=\sigma_{\alpha}(x,\beta)\equiv 0$ by Remark \ref{speccon} and Lemma \ref{ArnoLax}. So \eqref{e:16},\eqref{e:18} yield
\begin{equation}\label{e:20}
	F(s;0,1)=\exp\bigg(-\int_s^{\infty}\sigma_0(x,0)\d x\bigg)=\exp\bigg(\int_s^{\infty}(x-s)q_0(x,0)\d x\bigg),\ \ s\in\mathbb{R},
\end{equation}
where $\sigma_0(x,0)$ solves \eqref{e:4} with $\alpha=0$ and $q_0(x,0)$ solves \eqref{e:5}, also with $\alpha=0$. This, see also \eqref{e:17},\eqref{e:18}, makes \eqref{e:20} consistent with the Tracy-Widom and Forrester-Witte formula for $F_2(s)$, by \eqref{e:2},\eqref{e:7} and \eqref{e:8}.
\end{rem}
Our result about the limiting Fredholm determinant \eqref{e:14} is a modest improvement over \cite[Theorem $1.2$]{IKO} and \cite[Theorem $5$]{WXZ}. In \cite{IKO} the limiting kernel $A_s^{\alpha 1}(x,y)$ was identified in the context of the model
\begin{equation*}
	Z_{n,N}^{-1}|\det M|^{2\alpha}\e^{-N\textnormal{tr}\,V(M)}\d M,\ \ \ \alpha>-\frac{1}{2},
\end{equation*}
defined on $n\times n$ Hermitian matrices $M$, assuming $V$ is one-cut regular and assuming the origin is a right endpoint of the limiting mean eigenvalue density. In turn, $A_s^{\alpha1}(x,y)$ appeared in a particular double scaling limit $n,N\rightarrow\infty$. For general $\beta\in\mathbb{C}\setminus(-\infty,0)$, \cite{WXZ} identified $A_s^{\alpha\beta}(x,y)$ as limit for the reproducing kernel of the orthonormal polynomials associated with the Gaussian Fisher-Hartwig weight
\begin{equation*}
	w(x;\alpha,\mu,\omega)=\e^{-x^2}|x-\mu|^{2\alpha}\begin{cases}1,&x\leq\mu\\ \omega,&x>\mu\end{cases},
\end{equation*}
once $\mu$ is suitably edge-scaled. In \eqref{e:14} we lift the pointwise convergence of the reproducing kernel to the Fredholm convergence \eqref{e:14}, for general one-cut regular $V$ and $\alpha>-1,\beta\notin(-\infty,0)$. The connection to Painlev\'e functions in \eqref{e:16}, resp. \eqref{e:18}, is seemingly new and it shows that, although $A_s^{\alpha\beta}$ is constructed in terms of Painlev\'e-XXXIV data, its Fredholm determinant on $L^2(0,\infty)$ is not more complicated and also expressed in terms of solutions to Painlev\'e-XXXIV. Note that \eqref{e:16} uses the additional constraint $\beta\in\mathbb{C}\setminus(-\infty,1)$ which links to the smoothness of the underlying Painlev\'e functions.\bigskip

Next, we move to the first factor in \eqref{e:8}, the ratio of Hankel determinants. We address the same for quadratic $V(x)=x^2$ at the beginning.
\begin{theo}\label{BS:2} Let $\alpha>-1,\beta\in\mathbb{C}\setminus(-\infty,0)$ and $s\in\mathbb{R}$ be fixed. Then, as $n\rightarrow\infty$,
\begin{equation}\label{e:21}
	\ln\Bigg[\frac{D_n\big(\lambda_n,\alpha,\beta;x^2\big)}{D_n\big(\lambda_n,0,\beta;x^2\big)}\Bigg]=\frac{n\alpha}{2}(1-\ln 2)+\alpha s\sqrt[3]{n}+\frac{\alpha^2}{6}\ln n+\eta_{\alpha}(s,\beta)+\mathcal{O}\bigg(\frac{\ln n}{\sqrt[3]{n}}\bigg),
\end{equation}
with $\lambda_n$ as in \eqref{e:12} and where $s\mapsto\eta_{\alpha}(s,\beta)$ is such that
\begin{equation}\label{e:22}
	\frac{\d\eta_{\alpha}}{\d s}(s,\beta)=\sigma_{\alpha}(s,\beta)-\sigma_0(s,\beta)
\end{equation}
is the difference of two solutions of \eqref{e:4}, both with boundary constraint \eqref{e:17}. Moreover, as $n\rightarrow\infty$,
\begin{equation}\label{e:23}
	\ln\Bigg[\frac{D_n(\lambda_n,0,\beta;x^2)}{D_n(\lambda_n,0,1;x^2)}\Bigg]=-\int_s^{\infty}\sigma_0(x,\beta)\d x+\mathcal{O}\big(n^{-\frac{2}{3}}\big),
\end{equation}
in terms of $\sigma_0(x,\beta)$ that solves \eqref{e:4} with parameter $\alpha=0$ and boundary constraint \eqref{e:17}.
\end{theo}
\begin{rem} The large $n$-asymptotic of the Selberg integral $D_n(\lambda_n,0,1;x^2)$ is known, see for instance \cite{W}. Likewise the ratio asymptotics \eqref{e:23} is known from \cite[Theorem $2$]{BCI} where the same is written in terms of the Ablowitz-Segur solution to \eqref{e:6}. Here, we rederive \eqref{e:23} en route while proving \eqref{e:21}.
\end{rem}
Our proof workings for \eqref{e:21} produce an explicit formula for $\eta_{\alpha}(s,\beta)$ in terms of $\alpha$-quadratures over quantities from the Painlev\'e-XXXIV RHP \ref{YattRHP}, see \eqref{g:17} below. While lengthy, and hence our presentation of \eqref{e:22} in Theorem \ref{BS:2}, that formula gives $\eta_0(s,\beta)\equiv 0$. Moreover, identities \eqref{e:16} and \eqref{e:21} yield the following $s$-differentiable asymptotic:
\begin{cor} Let $\phi=\chi_{(s,\infty)}$ denote the indicator function on $(s,\infty)\subset\mathbb{R}$, $\phi_n$ be as in \eqref{e:13} and $\lambda_n$ as in \eqref{e:12}. Then, as $n\rightarrow\infty$,
\begin{equation}\label{e:24}
	E_n\big[\phi_n;\lambda_n,\alpha,\beta;x^2\big]=\exp\left(\frac{n\alpha}{2}(1-\ln 2)+\alpha s\sqrt[3]{n}+\frac{\alpha^2}{6}\ln n+\Xi_{\alpha}(s,\beta)+o(1)\right),
\end{equation}
with
\begin{equation*}
	\Xi_{\alpha}(s,\beta):=\eta_{\alpha}(s,\beta)-\int_s^{\infty}\big(\sigma_{\alpha}(x,\beta-1)-\sigma_{\alpha}(x,\beta)+\sigma_0(x,\beta)\big)\d x,
\end{equation*}
written in terms of \eqref{e:22} and \eqref{e:4},\eqref{e:17}. The asymptotic \eqref{e:24} is differentiable with respect to $(s,\alpha)$ and it holds uniformly in $(s,\alpha,\beta)\in\mathbb{R}\times(-1,\infty)\times(\mathbb{C}\setminus(-\infty,1))$ on compact sets.
\end{cor}
Result \eqref{e:24} aligns nicely with the Forrester-Witte result \cite[Proposition $26$]{FW}. Indeed, in \cite{FW} we find for $\beta=1$ and $\phi=\chi_{(s,\infty)}$,
\begin{equation}\label{e:25}
	\lim_{n\rightarrow\infty}\frac{\d}{\d s}\ln\Bigg[\e^{-\alpha s\sqrt[3]{n}}E_n\big[\phi_n;\lambda_n,\alpha,1;x^2\big]\Bigg]=\sigma_{\alpha}(s,0),
\end{equation}
just as we obtain it from \eqref{e:24}. Still, \eqref{e:24} is much stronger than \eqref{e:25} as we capture the full leading order large $n$-asymptotic of $E_n$, for general $\beta\in\mathbb{C}\setminus(-\infty,1)$. Next, we move from the Gaussian $V(x)=x^2$ in Theorem \ref{BS:2} to arbitrary one-cut regular $V$.
\begin{theo}\label{BS:3} Suppose $V$ is as in Assumption \ref{1-cut}. Let $\alpha>-1,\beta\in\mathbb{C}\setminus(-\infty,0)$ and $s\in\mathbb{R}$. Then, as $n\rightarrow\infty$, using the shorthand $W(x):=V(x)-x^2$,
\begin{align}
	\ln\Bigg[&\frac{D_n\big(\lambda_n,\alpha,\beta;V(x)\big)}{D_n(\lambda_n,\alpha,\beta;x^2)}\Bigg]=-n^2\int_E\frac{1}{2}\bigg(h_V(x)+\frac{1}{\pi}\bigg)\sqrt{2-x^2}\, W(x)\d x-\frac{n\alpha}{2\pi}\int_E\frac{W(x)}{\sqrt{2-x^2}}\d x+\frac{n\alpha}{2}W(\sqrt{2})\nonumber\\
	&+\frac{3\alpha s\sqrt[3]{n}}{2\sqrt{2}}\frac{(\pi h_V(\sqrt{2}))^{\frac{1}{3}}-1}{\pi h_V(\sqrt{2})-1}W'(\sqrt{2})+\Theta_{\alpha}(s,\beta)\ln\big(\pi h_V(\sqrt{2})\big)-\frac{1}{24}\ln\big(\pi h_V(-\sqrt{2})\big)+\mathcal{O}\big(n^{-\frac{1}{3}}\big),\label{e:26}
\end{align}
uniformly in the parameters $(s,\alpha,\beta)$ on compact sets in their admissible domain and where
\begin{equation*}
	\Theta_{\alpha}(s,\beta)=\frac{s^3}{6}-sa-\frac{2}{3}\big(Q_1^{21}+2Q_2^{12}\big),\hspace{1.5cm} \Theta_0(s,1)=-\frac{1}{24},
\end{equation*}
depends on $a=a(s,\alpha,\beta), Q_1^{21}=Q_1^{21}(s,\alpha,\beta)$ and $Q_2^{12}=Q_2^{12}(s,\alpha,\beta)$ that are defined in Lemma \ref{ArnoLax}.
\end{theo}
The first three terms in the right hand side of \eqref{e:26} are universal for ratio asymptotics of Hankel determinants with a single Fisher-Hartwig singularity and one-cut regular $V$, even when the singularity does not approach the soft edge, compare \cite[Proposition $5.5$]{BWW}. Only the fourth and fifth term in \eqref{e:26} feel the edge-scaling of $\lambda_n$: we observe fractional exponents of $n$ and Painlev\'e functions appear. Still, term four and five in \eqref{e:26} are consistent with the known results in \cite[Proposition $5.5$]{BWW} once $(\alpha,\beta)=(0,1)$ and necessarily the sixth term in \eqref{e:26} is equal to the corresponding left endpoint contribution in \cite[Proposition $5.5$]{BWW}. We provide a more detailed discussion and contextual placement of \eqref{e:26} after the next subsection.

\subsection{A probabilistic corollary} Ratio asymptotics of the type \eqref{e:21} or \eqref{e:26} are oftentimes useful in the derivation of limit theorems for the eigenvalue point process that underwrites \eqref{e:1}. For instance, see \cite[Section $1.1$]{BCI} or \cite[page $7520-7522$]{C}, the moment generating function of the number of eigenvalues greater than $\lambda_n$ relates to the ratio
\begin{equation*}
	\frac{D_n(\lambda_n,0,\beta;V(x))}{D_n(\lambda_n,0,1;V(x))}
\end{equation*}
and so all moments of the same random variable can be computed in terms of Painlev\'e-XXXIV data for large $n$ through \eqref{e:23} and \eqref{e:26}. Likewise, extreme value statistics in the thinned variant \cite{BP} of \eqref{e:1} are expressible as ratios of Hankel determinants with a single Fisher-Hartwig factor and so are averaged characteristic polynomials for the thinned process with a spectral gap. We won't focus on these quantities but instead state a limit theorem for the logarithm of the absolute value of the soft-edge scaled characteristic polynomial in the ensemble \eqref{e:1}.
\begin{cor} Suppose $M$ is drawn from \eqref{e:1} with $V$ as in Assumptoin \ref{1-cut} and let $\lambda_n$ be as in \eqref{e:12}. Then for any fixed $s\in\mathbb{R}$,
\begin{equation}\label{e:27}
	\sqrt{\frac{3}{\ln n}}\bigg[\ln\big|\det(M-\lambda_n I )\big|-n\rho_1-s\rho_2\sqrt[3]{n}\bigg]\Rightarrow N(0,1),
\end{equation}
as $n\rightarrow\infty$. Here, $N(0,1)$ is the standard normal random variable and $\rho_k$ abbreviate
\begin{align*}
	\rho_1:=\frac{1}{2}(1-\ln 2)-\frac{1}{2\pi}\int_E\frac{W(x)}{\sqrt{2-x^2}}\d x+\frac{1}{2}W(\sqrt{2}),\ \ \ 
	\rho_2:=1+\frac{3}{2\sqrt{2}}\frac{(\pi h_V(\sqrt{2}))^{\frac{1}{3}}-1}{\pi h_V(\sqrt{2})-1}W'(\sqrt{2}),
\end{align*}
with $W(x)=V(x)-x^2$.
\end{cor}
\begin{proof} Introduce the random variable
\begin{equation*}
	X_n(s):=\sqrt{\frac{3}{\ln n}}\bigg[\ln\big|\det(M-\lambda_n I )\big|-n\rho_1-s\rho_2\sqrt[3]{n}\bigg],\ \ \ n\in\mathbb{Z}_{\geq 2},\ \ s\in\mathbb{R}.
\end{equation*}
Expansion \eqref{e:21} and \eqref{e:26} yield for its moment generating function the following large $n$-asymptotic,
\begin{equation*}
	\mathbb{E}_n\Big(e^{tX_n(s)}\Big)=\big(1+o(1)\big)\exp\Big\{\big(\Theta_{\alpha}(s,\beta)-\Theta_0(s,\beta)\big)\ln(\pi h_V(\sqrt{2}))+\eta_{\alpha}(s,1)\Big\}\e^{\frac{1}{2}t^2},\ \ \alpha=t\sqrt{\frac{3}{\ln n}},
\end{equation*}
that holds pointwise in $t\in\mathbb{R}$. This, by uniformity in $\alpha>-1$ on compact sets, yields the claim \eqref{e:27}.
\end{proof}
Results of the form \eqref{e:27} are rather typical in many random matrix models. In fact, for the model \eqref{e:1} with one-cut regular $V$, \cite[Theorem $1.1$]{BWW} establishes convergence of
\begin{equation}\label{e:28}
	\frac{|\det(M-\lambda I)|^{\alpha}}{\mathbb{E}_n|\det(M-\lambda I)|^{\alpha}}\d\lambda
\end{equation}
to a Gaussian multiplicative chaos measure in the interior of $E\ni\lambda$, for suitable $\alpha\geq 0$. This is achieved by deriving asymptotics of the type \eqref{e:21},\eqref{e:26} but for $\lambda$ fixed in the interior of $E$. The limit for \eqref{e:28} in \cite{BWW} sheds light on the set of points $\lambda$ where the centered log-characteristic polynomial 
\begin{equation*}
	\ln|\det(M-\lambda I)|-\mathbb{E}_n\ln|\det(M-\lambda I)|
\end{equation*}
is exceptionally large, and this is an area of modern probability theory that has attracted substantial interest over the past decade, cf. \cite{FHK}. Our \eqref{e:27} does not offer new insights into logarithmically correlated fields, or chaos measures. It simply constitutes a CLT for the edge-scaled log-characteristic polynomial and showcases the appearance of Painlev\'e-XXXIV in the error control of \eqref{e:27}.
\subsection{Discussion and outline}
We are foremost motivated by the beautiful algebraic investigations of Forrester-Witte in \cite{FW} on \eqref{e:7} for $(\beta=1,V(x)=x^2)$ and the generalization of \eqref{e:25} to $\beta\neq 1$ and general one-cut regular $V$. While Okamoto's $\tau$-function theory might not be the ideal tool in attacking \eqref{e:7} in such a broad setting, the Riemann-Hilbert method employed here does the trick. En route, thanks to \eqref{e:8}, our analytic workings touch on two topics in integrable systems theory that have a long history. Namely, the asymptotic analysis of Hankel determinants with Fisher-Hartwig weights is classical, albeit mostly for bulk singularities. We refer the interested reader to the recent works \cite{C} and \cite{CFWW} that contain many pointers to the relevant literature. Edge-scaled Fisher-Hartwig singularities are somewhat under-appreciated, as evidenced by the short series of papers \cite{BCI,IK,WXZ}, none of which go beyond the Gaussian weight. Likewise, the expression of the limiting Fredholm determinant $F(s;\alpha,\beta)$ in \eqref{e:16} in terms of Painlev\'e-XXXIV is challenging to obtain within Okamoto's framework but it follows rather easily from the Riemann-Hilbert method by using dressing transformation techniques. These techniques are also a classical subject in modern integrable systems theory, compare for instance the monograph \cite{FT}.\smallskip

We now conclude our Introduction by a short overview over the paper's remainder. First, in Section \ref{sec2} we briefly review the Riemann-Hilbert approach of Fokas-Its-Kitaev \cite{FIK} to the orthonormal polynomials that underwrite all three factors in \eqref{e:8}, once $\lambda$ is edge-scaled as in \eqref{e:12}. This includes the central differential identities \eqref{a:9},\eqref{a:10} and \eqref{a:11}. Then, in Section \ref{sec3}, we carry out the necessary steps in the Deift-Zhou \cite{DZ,DKMVZ} asymptotic roadmap that culminate in Theorem \ref{Arnotheo1}. En route a central role is played by the Painlev\'e-XXXIV model problem \ref{YattRHP} that replaces the more commonly used Airy model problem \ref{Air} near $\lambda_n$. After that, in Section \ref{sec4}, we proof Theorem \ref{BS:1} by first computing \eqref{e:14} and after that by establishing \eqref{e:16} through an application of the Its-Izergin-Korepin-Slavnov \cite{IIKS} framework for integrable integral operators. The proofs of Theorems \ref{BS:2} and \ref{BS:3} are tedious, although the overall logic is clear: one derives asymptotics for the right hand sides in \eqref{a:9},\eqref{a:10},\eqref{a:11} and integrates them in parameter space. However, given the slow decay of the approximation error \eqref{a:35}, the execution of the same strategy requires us to iterate the integral representation \eqref{a:34} several times. This makes our workings in Section \ref{sec5}, leading to Theorem \ref{BS:2}, and in Section \ref{sec6}, leading to Theorem \ref{BS:3}, very lenghty. Lastly, Appendix \ref{tedious} and Appendix \ref{appPain} list certain technical aspects of our nonlinear steepest descent analysis and of the Painlev\'e-XXXIV model problem.

%%%%%%%%%%%%%%%%%%%%%%%%%%%%%%%%%%%%%%%%%%%%%%%%%%%%%%%%%%%%%%

\section{The Riemann-Hilbert method for orthogonal polynomials}\label{sec2} For any one-cut regular potential $V$, with data $(h_V,\ell_V)$, see \eqref{e:10},\eqref{e:11}, we consider the convex combination
\begin{equation}\label{a:1}
	V_{\theta}(x):=(1-\theta)x^2+\theta V(x),\ \ \ x\in\mathbb{R},\ \ \ \theta\in[0,1].
\end{equation}
It interpolates between the quadratic at $\theta=0$ and $V$ at $\theta=1$ and it does not leave the one-cut regular class:
\begin{lem}\label{convex} Suppose $V:\mathbb{R}\rightarrow\mathbb{R}_{\geq 0}$ is one-cut regular with data $(h_V,\ell_V)$ and $V_{\theta}:\mathbb{R}\rightarrow\mathbb{R}_{\geq 0}$ given by \eqref{a:1} for any $\theta\in[0,1]$. Then $V_{\theta}$ is one-cut regular for all $\theta\in[0,1]$ with data $(h_{V_\theta},\ell_{V_\theta})$ given by
\begin{equation}\label{a:2}
	h_{V_\theta}(x)=(1-\theta)\frac{1}{\pi}+\theta h_V(x),\ \ x\in\mathbb{R};\ \ \ \ \ \ \ \ \ell_{V_\theta}=(1-\theta)(1+\ln 2)+\theta\ell_V,
\end{equation}
for all $\theta\in[0,1]$.
\end{lem}
\begin{proof} The density function 
\begin{equation*}
	\rho_{V_{\theta}}(x):=\bigg[(1-\theta)\frac{1}{\pi}+\theta h_{V}(x)\bigg]\sqrt{(2-x^2)_+},\ \ \ x\in\mathbb{R},
\end{equation*}
is seen to satisfy the strict version of \eqref{e:10} for all $\theta\in[0,1]$ with $E=[-\sqrt{2},\sqrt{2}\,]$ and $\ell_{V_{\theta}}$ as in \eqref{a:2}. Since $x\mapsto V_{\theta}(x)$ is real analytic for all $\theta\in[0,1]$, the density of its equilibrium measure thus coincides with the above $\rho_{V_{\theta}}$ and $V_{\theta}$ is one-cut regular.
\end{proof}
Setting for all $\theta\in[0,1]$, compare \eqref{e:12} with the convention $\lambda_{n1}\equiv\lambda_n$,
\begin{equation}\label{a:3}
	\lambda_{n\theta}:=\sqrt{2}+\frac{s}{(n\tau_{\theta})^{\frac{2}{3}}},\ \ \ s\in\mathbb{R};\ \ \ \ \ \ \ \ \ \tau_{\theta}:=\pi h_{V_{\theta}}(\sqrt{2})2^{\frac{3}{4}}>0,
\end{equation}
we now consider the following well-known Riemann-Hilbert problem (RHP).
\begin{problem}\label{FIKbeast} Let $s\in\mathbb{R},\alpha\in\mathbb{R}_{>-1},\beta\in\mathbb{C}\setminus(-\infty,0),n\in\mathbb{N}$ and $\theta\in[0,1]$. Find a function $X(z)=X(z;s,\alpha,\beta,n,\theta)\in\mathbb{C}^{2\times 2}$ such that
\begin{enumerate}
	\item[(1)] $X(z)$ is analytic for $z\in\mathbb{C}\setminus\mathbb{R}$ and extends continuously to the closer upper and lower half-planes away from a neighborhood of $z=\lambda_{n\theta}$.
	\item[(2)] The continuous limiting values $X_{\pm}(z)=\lim_{\epsilon\downarrow 0}X(z\pm\im\epsilon)$ at $z\in\mathbb{R}\setminus\{\lambda_{n\theta}\}$ satisfy the condition
	\begin{equation}\label{a:4}
		X_+(z)=X_-(z)\begin{bmatrix}1 & w(z)\\ 0&1\end{bmatrix},\ \ \ \ \ \ w(z)=w(z;s,\alpha,\beta,n,\theta):=\omega_{\alpha\beta}(z-\lambda_{n\theta})\e^{-nV_{\theta}(z)}.
	\end{equation}
	\item[(3)] As $z\rightarrow\lambda_{n\theta},z\notin\mathbb{R}$, the first column of $X(z)$ remains bounded but its second column satisfies
	\begin{equation*}
		X^{j2}(z)=\mathcal{O}(|z-\lambda_{n\theta}|^{\alpha}),\ \alpha\in(-1,0);\ \ \ X^{j2}(z)=\mathcal{O}(\ln|z-\lambda_{n\theta}|),\ \alpha=0;\ \ \ X^{j2}(z)=\mathcal{O}(1),\ \alpha>0.
	\end{equation*}
	\item[(4)] As $z\rightarrow\infty$ and $z\notin\mathbb{R}$, $X(z)$ satisfies the asymptotic normalization
	\begin{equation*}
		X(z)=\Big\{I+\mathcal{O}\big(z^{-1}\big)\Big\}z^{n\sigma_3},\ \ \ \sigma_3:=\begin{bmatrix}1&0\\ 0&-1\end{bmatrix}.
	\end{equation*}
\end{enumerate}
\end{problem}
General theory, cf. \cite{FIK}, underwriting RHP \ref{FIKbeast} asserts unique solvability of the problem, for given parameters $(s,\alpha,\beta,n,\theta)$, precisely when the orthonormal polynomials $p_{n,n}(x)$ and $p_{n-1,n}(x)$ exist, where 
\begin{equation}\label{a:5}
	p_{j,n}(x)=\kappa_{j,n}\Big(x^j+\delta_{j,n}x^{j-1}+\gamma_{j,n}x^{j-2}+\mathcal{O}\big(x^{j-3}\big)\Big),\ \ \ \ \ \kappa_{j,n}\neq 0,\ \ \ j\in\mathbb{Z}_{\geq 0},
\end{equation}
and where we require, dropping $(s,\alpha,\beta,\theta)$ from our notation unless absolutely necessary,
\begin{equation*}
	\int_{-\infty}^{\infty}p_{j,n}(x)p_{k,n}(x)w(x)\d x=\begin{cases}1,&j=k\\ 0,&j\neq k\end{cases}.
\end{equation*}
In fact, if RHP \ref{FIKbeast} is solvable then its solution is of the form
\begin{equation*}
	X(z)=\begin{bmatrix}\kappa_{n,n}^{-1}\,p_{n,n}(z) & \frac{\kappa_{n,n}^{-1}}{2\pi\im}\int_{-\infty}^{\infty}p_{n,n}(x)\frac{w(x)}{x-z}\d x\smallskip\\ -2\pi\im\kappa_{n-1,n}\,p_{n-1,n}(z)&-\kappa_{n-1,n}\int_{-\infty}^{\infty}p_{n-1,n}(x)\frac{w(x)}{x-z}\d x\end{bmatrix},\ \ \ z\in\mathbb{C}\setminus\mathbb{R},
\end{equation*}
and we have, cf. \cite[Chapter $3.2$]{D},
\begin{align}
	\kappa_{n-1,n}^2=&\,\frac{\im}{2\pi}\lim_{z\rightarrow\infty}\big(X^{21}(z)z^{-n+1}\big),\hspace{3cm}\  \delta_{n,n}=\lim_{z\rightarrow\infty}\Big(\big(X^{11}(z)-z^n\big)z^{-n+1}\Big),\label{a:6}\\
	\gamma_{n,n}=&\,\lim_{z\rightarrow\infty}\Big(\big(X^{11}(z)-z^n-\delta_{n,n}z^{n-1}\big)z^{-n+2}\Big),\ \ \ \ \ \kappa_{n,n}^{-2}=-2\pi\im\lim_{z\rightarrow\infty}\big(X^{12}(z)z^{n+1}\big).\label{a:7}
\end{align}
The upcoming workings will establish, for fixed $s\in\mathbb{R},\alpha\in\mathbb{R}_{>-1},\beta\in\mathbb{C}\setminus(-\infty,0)$ and any $\theta\in[0,1]$, existence of $n_0=n_0(s,\alpha,\beta,\theta)\in\mathbb{N}$ so that RHP \ref{FIKbeast} is solvable for all $n\geq n_0$. We will be able to calculate the large $n$-asymptotic of its solution $X(z)$ and from it the asymptotic of \eqref{e:8} will follow,  to a certain extent, from an interplay of the below differential identities. These identities are commonly used in the analysis of moment determinants. Recall the Hankel determinant, with $w$ in \eqref{a:4},
\begin{equation*}
	D_n\big(\lambda_{n\theta},\alpha,\beta;V_{\theta}(x)\big)=\frac{1}{n!}\int_{\mathbb{R}^n}\prod_{1\leq j<k\leq n}|x_k-x_j|^2\prod_{m=1}^nw(x_m)\d x_m.
\end{equation*}
\begin{lem}[{\cite[Lemma $1$]{WXZ}}] Let $s\in\mathbb{R},\alpha\in\mathbb{R}_{>-1},\beta\in\mathbb{C}\setminus(-\infty,0),n\in\mathbb{N}$ be such that $D_k(\lambda_{n0},\alpha,\beta;x^2)\neq 0$ for all $k=1,\ldots,n+1$.
Then, 
\begin{equation}\label{a:9}
	\frac{\partial}{\partial\lambda_{n0}}\ln D_n\big(\lambda_{n0},\alpha,\beta;x^2\big)=2n\delta_{n,n},
\end{equation}
using a fixed branch for the logarithm and with $\delta_{n,n}$ as in \eqref{a:5},\eqref{a:6} at $\theta=0$.
\end{lem}
\begin{lem}[{\cite[Proposition $2$]{Kra}}] Pick $s\in\mathbb{R},\alpha\in\mathbb{R}_{>-1},\beta\in\mathbb{C}\setminus(-\infty,0),n\in\mathbb{N}$ so $D_k(\lambda_{n0},\alpha,\beta;x^2)\neq 0$ for all $k=1,\ldots,n+1$. Then, 
\begin{align}
	&\frac{\partial}{\partial\alpha}\ln D_n\big(\lambda_{n0},\alpha,\beta;x^2\big)=-\frac{n}{2}\frac{\partial}{\partial\alpha}\ln\kappa_{n-1,n}^2+\frac{1}{2}(n+\alpha)\frac{\partial}{\partial\alpha}\ln\kappa_{n,n}^{-2}-n\bigg(\frac{\kappa_{n-1,n}}{\kappa_{n,n}}\bigg)^2\frac{\partial}{\partial\alpha}\ln\bigg(\frac{\kappa_{n-1,n}^2}{\kappa_{n,n}^2}\bigg)\nonumber\\
	&+2n\frac{\partial}{\partial\alpha}\bigg(\gamma_{n,n}-\frac{1}{2}\delta_{n,n}^2\bigg)+\alpha\kappa_{n,n}^{-1}\frac{\partial}{\partial\alpha}\big(p_{n,n}(\lambda_{n0})\big)\mathring{X}^{22}(\lambda_{n0})+2\pi\im\alpha\kappa_{n-1,n}\frac{\partial}{\partial\alpha}\big(p_{n-1,n}(\lambda_{n0})\big)\mathring{X}^{12}(\lambda_{n0}),\label{a:10}
\end{align}
using a fixed branch for the logarithms. In \eqref{a:10}, $\mathring{X}^{j2}(\lambda_{n0}),j=1,2$ denote the regularized entries
\begin{equation*}
	\mathring{X}^{j2}(\lambda_{n0}):=\lim_{\substack{z\rightarrow\lambda_{n0}\\ \Im z>0}}\big(X^{j2}(z)-S_j(z)\big),\ \ \ \ S_j(z):=\begin{cases}\displaystyle\frac{\im\pi\e^{-\im\frac{\pi}{2}\alpha}}{\cos(\frac{\pi}{2}\alpha)}\,f_j(\lambda_n)(z-\lambda_n)^{\alpha},&\alpha\in(-1,0)\smallskip\\ \displaystyle\im\pi f_j(\lambda_n),&\alpha=0\bigskip\\ \displaystyle 0,&\alpha\in(0,\infty)\end{cases},
\end{equation*}
with the entire functions $f_1(z):=\frac{\kappa_{n,n}^{-1}}{2\pi\im}p_{n,n}(z)\e^{-nz^2},f_2(z):=-\kappa_{n-1,n}p_{n-1,n}(z)\e^{-nz^2}$
%\begin{equation*}
%	f_1(z):=\frac{\kappa_{n,n}^{-1}}{2\pi\im}p_{n,n}(z)\e^{-nz^2},\ \ \ \ \ \ \ \ \ f_2(z):=-\kappa_{n-1,n}p_{n-1,n}(z)\e^{-nz^2},
%\end{equation*}
and we use $\kappa_{n-1,n}^2,\delta_{n,n}$, $\gamma_{n,n},\kappa_{n,n}^{-2}$ as in \eqref{a:5},\eqref{a:6},\eqref{a:7}, all at $\theta=0$.
\end{lem}
\begin{lem}[{\cite[Lemma $3.7$]{BWW}}] Assume $s\in\mathbb{R},\alpha\in\mathbb{R}_{>-1},\beta\in\mathbb{C}\setminus(-\infty,0),n\in\mathbb{N}$ and $\theta\in[0,1]$ are such that $D_k(\lambda_{n\theta},\alpha,\beta;V_{\theta}(x))\neq 0$ for all $k=1,\ldots,n+1$. Then, 
\begin{equation}\label{a:11}
	\frac{\partial}{\partial\theta}\ln D_n\big(\lambda_{n\theta},\alpha,\beta;V_{\theta}(x)\big)=-\frac{n}{2\pi\im}\int_{-\infty}^{\infty}\bigg\{X(z)^{-1}X'(z)\bigg\}^{21}w(z)\big(V(z)-z^2\big)\d z,
\end{equation}
using a fixed branch for the logarithm and with $X(z)$ as in RHP \ref{FIKbeast} so $X'(z)=\frac{\d}{\d z}X(z)$.
\end{lem}
Identities \eqref{a:9},\eqref{a:10} and \eqref{a:11} are valid away from the exceptional sets that underwrite the vanishing of some (or possibly all) of the $\{D_k\}_{k=1}^{n+1}$. Yet, our upcoming analysis will establish non-vanishing of $D_n$ for $n\geq n_0$ and \textit{all} $s\in\mathbb{R},\alpha>-1,\beta\notin(-\infty,0),\theta\in[0,1]$, plus its large $n$-asymptotic. This will be sufficient to conclude validity of \eqref{e:21},\eqref{e:23},\eqref{e:26} in the entire parameter domain by general regularity properties of $D_n$ in its various parameters. We now continue with the nonlinear steepest descent analysis of RHP \ref{FIKbeast}.

%%%%%%%%%%%%%%%%%%%%%%%%%%%%%%%%%%%%%%%%%%%%%%%%%%%%%%%%%%%%%%
\section{Nonlinear steepest descent analysis of RHP \ref{FIKbeast} for general parameters}\label{sec3}

Several steps are needed for the asymptotic, as $n\rightarrow\infty$, resolution of RHP \ref{FIKbeast} and we present them in their natural chronological order.
\subsection{First transformation: translation and \textbf{\textit{g}}-function} We first move the singularity $\lambda_{n\theta}$ in RHP \ref{FIKbeast} to $z=\sqrt{2}$ by the simple translation
\begin{equation}\label{a:12}
	Y(z;s,\alpha,\beta,n,\theta):=X(z+s_{n\theta};s,\alpha,\beta,n,\theta),\ \ \ \ z\in\mathbb{C}\setminus\mathbb{R};\ \ \ \ \ \ \ \ s_{n\theta}:=\frac{s}{(n\tau_{\theta})^{\frac{2}{3}}}\in\mathbb{R}.
\end{equation}
Once done, the potential $z\mapsto V_{\theta}(z)$ in condition $(2)$ of RHP \ref{FIKbeast} is replaced by the $n$-dependent one $z\mapsto V_{\theta}(z+s_{n\theta})$, so the $g$-function for our analysis is a shifted version of the one underlying \eqref{a:1}. In detail, the relevant analytic and asymptotic properties of the $g$-function corresponding to \eqref{a:1} are summarized below.
\begin{lem}\label{glemma} Let $n\in\mathbb{N},\theta\in[0,1]$ and $V_{\theta},h_{V_{\theta}}$ as in \eqref{a:1},\eqref{a:2}. Set
\begin{equation}\label{a:13}
	 \rho_{V_{\theta}}(x):=h_{V_{\theta}}(x)\sqrt{(2-x^2)_+},\ x\in\mathbb{R};\ \ \ \ \ \ g_{\theta}(z):=\int_E\ln(z-x)\rho_{V_{\theta}}(x)\d x,\ \ \ z\in\mathbb{C}\setminus(-\infty,\sqrt{2}\,],
\end{equation}
with the principal branch in the definition of the logarithm $\ln:\mathbb{C}\setminus(-\infty,0]\rightarrow\mathbb{C}$. Then, for all $\theta\in[0,1]$,
\begin{enumerate}
	\item[(i)] $z\mapsto g_{\theta}(z)$ is analytic in $\mathbb{C}\setminus(-\infty,\sqrt{2}\,]$ and $z\mapsto\e^{ng_{\theta}(z)}$ is analytic in $\mathbb{C}\setminus E$.
	\item[(ii)] $z\mapsto g_{\theta}(z)$ admits limiting values $g_{\theta\pm}(z):=\lim_{\epsilon\downarrow 0}g_{\theta}(z\pm\im\epsilon)$ on $\mathbb{R}\setminus\{\pm\sqrt{2}\}$ such that for $z<-\sqrt{2}$,
	\begin{equation*}
		g_{\theta\pm}(z)=\int_E\ln|z-x|\rho_{V_{\theta}}(x)\d x\pm\im\pi,
	\end{equation*}
	for $z\in(-\sqrt{2},\sqrt{2})$,
	\begin{equation*}
		g_{\theta\pm}(z)=\int_E\ln|z-x|\rho_{V_{\theta}}(x)\d x\pm\im\pi\int_z^{\sqrt{2}}\rho_{V_{\theta}}(x)\d x,
	\end{equation*}
	and for $z>\sqrt{2}$,
	\begin{equation*}
		g_{\theta\pm}(z)=\int_E\ln|z-x|\rho_{V_{\theta}}(x)\d x.
	\end{equation*}
	\item[(iii)] As $z\rightarrow\infty$ in $\mathbb{C}\setminus(-\infty,\sqrt{2}\,]$,
	\begin{equation*}
		g_{\theta}(z)=\ln z-\frac{1}{z}\int_Ex\rho_{V_{\theta}}(x)\d x-\frac{1}{2z^2}\int_Ex^2\rho_{V_{\theta}}(x)\d x+\mathcal{O}\big(z^{-3}\big).
	\end{equation*}
	\item[(iv)] Making explicit the left hand side in \eqref{e:10}, we have for $z>\sqrt{2}$,
	\begin{equation*}
		g_{\theta+}(z)+g_{\theta-}(z)-V_{\theta}(z)+\ell_{V_{\theta}}=-2\pi\int_{\sqrt{2}}^zh_{V_{\theta}}(x)\sqrt{x^2-2}\,\d x,
	\end{equation*}
	and for $z<-\sqrt{2}$,
	\begin{equation*}
		g_{\theta+}(z)+g_{\theta-}(z)-V_{\theta}(z)+\ell_{V_{\theta}}=-2\pi\int_z^{-\sqrt{2}}h_{V_{\theta}}(x)\sqrt{x^2-2}\,\d x.
	\end{equation*}
\end{enumerate}
\end{lem}
In order to put \eqref{a:13} to good use, we transform the RHP for $Y(z)$ via the $g$-function as follows. Introduce
\begin{equation}\label{a:14}
	T(z;s,\alpha,\beta,n,\theta):=\exp\bigg[\frac{n}{2}\ell_{V_{\theta}}\sigma_3\bigg]Y(z;s,\alpha,\beta,n,\theta)\exp\bigg[-n\Big(g_{\theta}(z+s_{n\theta})+\frac{1}{2}\ell_{V_{\theta}}\Big)\sigma_3\bigg],\ \ z\in\mathbb{C}\setminus\mathbb{R},
\end{equation}
with $\ell_{V_{\theta}}$ as in \eqref{a:2} and $g_{\theta}(z)$ as in \eqref{a:13}. What results for $T(z)$ is the following RHP:
\begin{problem}\label{gfct} Let $s\in\mathbb{R},\alpha>-1,\beta\notin(-\infty,0),n\in\mathbb{N}$ and $\theta\in[0,1]$. The function $T(z)=T(z;s,\alpha,\beta,n,\theta)\in\mathbb{C}^{2\times 2}$ defined in \eqref{a:14} is uniquely determined by the following four properties:
\begin{enumerate}
	\item[(1)] $T(z)$ is analytic for $z\in\mathbb{C}\setminus\mathbb{R}$ and extends continuously to the closed upper and lower half-planes away from a neighborhood of $z=\sqrt{2}$.
	\item[(2)] The pointwise limits $T_{\pm}(z):=\lim_{\epsilon\downarrow 0}T(z\pm\im\epsilon)$ at $z\in\mathbb{R}\setminus\{\sqrt{2}\}$ satisfy $T_+(z)=T_-(z)G_T(z)$ with jump
	\begin{equation}\label{a:15}
		G_T(z)=G_T(z;s,\alpha,\beta,n,\theta)=\begin{bmatrix}\e^{-n\pi(z)} & \omega_{\alpha\beta}(z-\sqrt{2})\e^{n\eta(z)}\smallskip\\ 0 & \e^{n\pi(z)}\end{bmatrix},
	\end{equation}
	which we write with the shorthand
	\begin{equation*}
		\pi(z):=g_{\theta+}(z+s_{n\theta})-g_{\theta-}(z+s_{n\theta})=\begin{cases}0,&z+s_{n\theta}>\sqrt{2}\\ 2\pi\im\int_{z+s_{n\theta}}^{\sqrt{2}}\rho_{V_{\theta}}(x)\d x,&z+s_{n\theta}\in(-\sqrt{2},\sqrt{2})\\ 2\pi\im,&z+s_{n\theta}<-\sqrt{2}\end{cases},
	\end{equation*}
	and the shorthand
	\begin{align*}
		\eta(z):=g_{\theta+}(z\,+\,&s_{n\theta})+g_{\theta-}(z+s_{n\theta})-V_{\theta}(z+s_{n\theta})+\ell_{V_{\theta}}\\
		=&\,\,2\int_E\ln|z+s_{n\theta}-x|\rho_{V_{\theta}}(x)\d x-V_{\theta}(z+s_{n\theta})+\ell_{V_{\theta}},\ \ \ \ z\in\mathbb{R}\setminus\{\pm\sqrt{2}-s_{n\theta}\}.
	\end{align*}
	\item[(3)] As $z\rightarrow\sqrt{2}$, $z\notin\mathbb{R}$, the first column of $T(z)$ remains bounded but its second column satisfies
	\begin{equation*}
		T^{j2}(z)=\mathcal{O}(|z-\sqrt{2}|^{\alpha}),\ \alpha\in(-1,0);\ \ \ T^{j2}(z)=\mathcal{O}(\ln|z-\sqrt{2}|),\ \alpha=0;\ \ \ T^{j2}(z)=\mathcal{O}(1),\ \alpha>0.
	\end{equation*}	
	\item[(4)] As $z\rightarrow\infty$, $T(z)$ satisfies the normalization $T(z)=I+\mathcal{O}(z^{-1})$.
\end{enumerate}
\end{problem}
The triangular structure of $G_T(z)$ in \eqref{a:15} motivates our next transformation.
\subsection{Second transformation: opening of lenses} We have $\eta(z)=0$ for $z+s_{n\theta}\in(-\sqrt{2},\sqrt{2})$ by the variational equality right above \eqref{e:10} and since $x\mapsto h_{V_{\theta}}(x)$ is real analytic, the function
\begin{equation}\label{a:16}
	z\mapsto\xi(z)=\xi(z;s,n,\theta):=2\pi\int_{\sqrt{2}}^{z+s_{n\theta}}h_{V_{\theta}}(w)\big(w^2-2\big)^{\frac{1}{2}}\d w,
\end{equation}
is analytic for $z+s_{n\theta}\in\mathbb{C}\setminus(-\infty,\sqrt{2}]$, for any $\theta\in [0,1]$, in the domain of analyticity of $h_{V_{\theta}}$. The path of integration in \eqref{a:16} does not cross $E\subset\mathbb{R}$, it stays in the domain of analyticity of $h_{V_{\theta}}$, and we choose the principal branch for the complex square root with cut on $E$. With this convention \eqref{a:16} yields
\begin{equation*}
	\xi_+(x):=\lim_{y\downarrow 0}\xi(x+\im y)=-\pi(x),\ \ \ \ \ \xi_-(x):=\lim_{y\downarrow 0}\xi(x-\im y)=\pi(x)\ \ \ \ \ \ \ \ \forall\,x+s_n\in(-\sqrt{2},\sqrt{2}),
\end{equation*}
as well as
\begin{equation*}
	\lim_{y\downarrow 0}\frac{\d}{\d y}\xi(x+\im y)=-2\pi\rho_{V_{\theta}}(x+s_{n\theta})<0\ \ \ \ \ \ \ \ \forall\,x+s_{n\theta}\in(-\sqrt{2},\sqrt{2}).
\end{equation*}
Thus, by the Cauchy-Riemann equations, $\Re(\xi(z))<0$ locally in the upper and lower half-planes, more detailed for $z$ in two $(\alpha,\beta,\theta)$-independent vicinities
\begin{equation*}
	\Omega_+\subset\{z\in\mathbb{C}:\ \Im z>0\}\ \ \ \ \ \ \textnormal{and}\ \ \ \ \ \ \ \Omega_-\subset\{z\in\mathbb{C}:\ \Im z<0\},
\end{equation*}
such that $\Omega_{\pm}$ lie in the strip $\{z=x+\im y:\ x+s_{n\theta}\in(-\sqrt{2},\sqrt{2}),\ y\in\mathbb{R}\}$ and in the domain of analyticity of $h_{V_{\theta}}$. Moving forward, we assume from now on that $n\geq n_0$ is sufficiently large so that $-\sqrt{2}-s_{n\theta}<\sqrt{2}$, which is admissible since $s\in\mathbb{R}$ is kept fixed, compare \eqref{a:12}. Then, considering the lens shaped domain $L_+\cup L_-$ shown in Figure \ref{fig1} with $L_{\pm}\subset\Omega_{\pm}$, we open a lens around $[-\sqrt{2}-s_{n\theta},\sqrt{2}]$ by defining 
\begin{equation}\label{a:17}
	S(z;s,\alpha,\beta,n,\theta):=T(z;s,\alpha,\beta,n,\theta)\begin{cases}L(z)^{-1},&z\in L_+\\ L(z),&z\in L_-\\ I,&\textnormal{else}\end{cases},
\end{equation}
in terms of the lower triangular multiplier
\begin{equation*}
	L(z)=L(z;s,\alpha,n,\theta):=\begin{bmatrix}1&0\\ (\sqrt{2}-z)^{-\alpha}\e^{n\xi(z)} & 1\end{bmatrix},
\end{equation*}
where $\mathbb{C}\setminus[\sqrt{2},\infty)\ni z\mapsto (\sqrt{2}-z)^{-\alpha}$ has its cut on $[\sqrt{2},\infty)\subset\mathbb{R}$ so that $(\sqrt{2}-z)^{-\alpha}>0$ for $z<\sqrt{2}$. The relevant analytic and asymptotic properties of $S(z)$ are summarized below.
\begin{problem}\label{opRHP} Let $s\in\mathbb{R},\alpha>-1,\beta\notin(-\infty,0),\theta\in[0,1]$ and $\mathbb{N}\ni n\geq n_0$ so that $-\sqrt{2}-s_{n\theta}<\sqrt{2}$. The function $S(z)=S(z;s,\alpha,\beta,n,\theta)\in\mathbb{C}^{2\times 2}$ defined in \eqref{a:17} is uniquely determined by the following four properties:
\begin{enumerate}
	\item[(1)] $S(z)$ is analytic for $z\in\mathbb{C}\setminus\Sigma_S$, with $\Sigma_S$ shown in Figure \ref{fig1}. It extends continuously to the closure of $\mathbb{C}\setminus\Sigma_S$ away from $z=\sqrt{2}$.
		\begin{figure}[tbh]
	\begin{tikzpicture}[xscale=0.9,yscale=0.9]
	\draw [thick, color=red, decoration={markings, mark=at position 0.15 with {\arrow{>}}},decoration={markings, mark=at position 0.4 with {\arrow{>}}}, decoration={markings, mark=at position 0.6 with {\arrow{>}}}, decoration={markings, mark=at position 0.85 with {\arrow{>}}}, postaction={decorate}] (-5,0) -- (5,0);
		\draw[thick,color=red, decoration={markings, mark=at position 0.5 with {\arrow{<}}}, postaction={decorate} ] (2.5,0) arc (38.65980825:141.3401918:3.201562119);
		\draw[thick,color=red, decoration={markings, mark=at position 0.5 with {\arrow{<}}}, postaction={decorate} ] (2.5,0) arc (-38.65980825:-141.3401918:3.201562119);
	\node [below] at (2.7,-0.15) {{\footnotesize $\sqrt{2}$}};
	\node [below] at (-3.2,-0.15) {{\footnotesize $-\sqrt{2}-s_{n\theta}$}};
	\node [right] at (-1.9,1.4) {{\footnotesize $\partial L_+$}};
	\node [right] at (-0.4,0.6) {{\footnotesize $L_+$}};
	\node [right] at (-1.9,-1.4) {{\footnotesize $\partial L_-$}};
	\node [right] at (-0.4,-0.6) {{\footnotesize $L_-$}};
	\draw [fill, color=black] (2.5,0) circle [radius=0.06];
	\draw [fill, color=black] (-2.5,0) circle [radius=0.06];
\end{tikzpicture}
\caption{The oriented jump contour $\Sigma_S$, shown in red, for $S(z)$ in the complex $z$-plane.}
\label{fig1}
\end{figure}
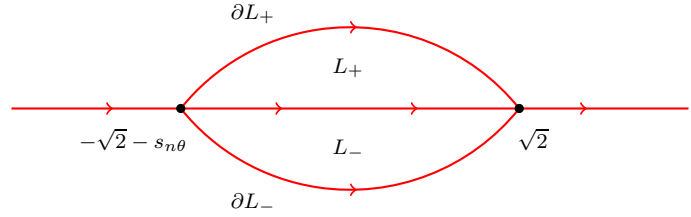
	\item[(2)] The non-tangential limiting values $S_{\pm}(z),z\in\Sigma_S\setminus\{\sqrt{2}\}$ from either side of $\Sigma_S$ satisfy the relation $S_+(z)=S_-(z)G_S(z)$ with $G_S(z)=G_S(z;s,\alpha,\beta,n,\theta)$ given by
	\begin{equation*}
		G_S(z)=\begin{bmatrix}1&|z-\sqrt{2}|^{\alpha}\e^{n\eta(z)}\\ 0&1\end{bmatrix},\ \ z<-\sqrt{2}-s_{n\theta};\ \ \ \ \ \ \ \ G_S(z)=L(z),\ \ z\in\partial L_+\cup\partial L_-.
	\end{equation*}
	For the remaining parts of $\Sigma_S$ on the real line, if $\sqrt{2}-s_{n\theta}<\sqrt{2}$, then
	\begin{align*}
		G_S(z)=&\,\begin{bmatrix}0&|z-\sqrt{2}|^{\alpha}\\ -|z-\sqrt{2}|^{-\alpha}&0\end{bmatrix},\ \ z\in(-\sqrt{2}-s_{n\theta},\sqrt{2}-s_{n\theta});\\
		G_S(z)=&\,\begin{bmatrix}0&|z-\sqrt{2}|^{\alpha}\e^{-n\xi(z)}\\ -|z-\sqrt{2}|^{-\alpha}\e^{n\xi(z)} & 0\end{bmatrix},\ \ z\in(\sqrt{2}-s_{n\theta},\sqrt{2});\\
		G_S(z)=&\,\begin{bmatrix}1&\beta|z-\sqrt{2}|^{\alpha}\e^{n\eta(z)}\\ 0&1\end{bmatrix},\ \ z\in(\sqrt{2},\infty).
	\end{align*}
	If however $\sqrt{2}<\sqrt{2}-s_{n\theta}$, then instead
	\begin{align*}
		G_S(z)=&\,\begin{bmatrix}0&|z-\sqrt{2}|^{\alpha}\\ -|z-\sqrt{2}|^{-\alpha}&0\end{bmatrix},\ \ z\in(-\sqrt{2}-s_{n\theta},\sqrt{2});\\
		G_S(z)=&\,\begin{bmatrix}\e^{-n\pi(z)} & |z-\sqrt{2}|^{\alpha}\\ 0&\e^{n\pi(z)}\end{bmatrix},\ \ z\in(\sqrt{2},\sqrt{2}-s_{n\theta});\\
		G_S(z)=&\,\begin{bmatrix}1&\beta|z-\sqrt{2}|^{\alpha}\e^{n\eta(z)}\\ 0&1\end{bmatrix},\ \ z\in(\sqrt{2}-s_{n\theta},\infty).
	\end{align*}
	\item[(3)] As $z\rightarrow\sqrt{2},z\notin\Sigma_S$, the behavior of $S(z)$ is as the behavior of $T(z)$ in condition $(3)$ of RHP \ref{gfct}, provided $z$ lies outside of $L_{\pm}$. Inside $L_{\pm}$ the relevant behavior is modified according to \eqref{a:17}.
	\item[(4)] As $z\rightarrow\infty,z\notin\Sigma_S$, $S(z)$ satisfies the normalization $S(z)=I+\mathcal{O}(z^{-1})$.
\end{enumerate}
\end{problem}
Given our discussion of $\Re(\xi(z))$ below \eqref{a:16}, the jump matrix $G_S(z)=G_S(z;s,\alpha,\beta,n,\theta)$ is asymptotically, as $n\rightarrow\infty$, close to the identity matrix, away from $z=-\sqrt{2}-s_{n\theta}$, from $z=\sqrt{2}$ and from the line segment in between. In more detail:
\begin{prop}\label{snorm1} Let $s\in\mathbb{R},\alpha>-1,\beta\notin(-\infty,0),\epsilon\in(0,1)$. There exist $n_0=n_0(s,\alpha,\beta,\epsilon)\in\mathbb{N}$ and $c_j=c_j(s,\alpha,\beta,\epsilon)>0$ so that
\begin{align*}
	\|G_S(\cdot;s,\alpha,\beta,n,\theta)-I\|_{L^2\cap L^{\infty}(-\infty,-\sqrt{2}-s_{n\theta}-\epsilon)}\leq c_1\e^{-nc_2},\ \ 
	\|G_S(\cdot;s,\alpha,\beta,n,\theta)-I\|_{L^2\cap L^{\infty}(m+\epsilon,\infty)}\leq c_3\e^{-nc_4},
\end{align*}
for all $n\geq n_0$ and $\theta\in[0,1]$ with $m:=\max\{\sqrt{2}-s_{n\theta},\sqrt{2}\}$. Moreover, there exist $d_j=d_j(s,\alpha,\epsilon)>0$ so that
\begin{equation*}
	\|G_S(\cdot;s,\alpha,\beta,n,\theta)-I\|_{L^2\cap L^{\infty}(\partial L_{\pm}\setminus(\mathbb{D}_{\epsilon}(-\sqrt{2}-s_{n\theta})\cup\mathbb{D}_{\epsilon}(\sqrt{2})))}\leq d_1\e^{-nd_2}
\end{equation*}
is valid for all $n\geq n_0$ and $\theta\in[0,1]$, with the open disk $\mathbb{D}_{\epsilon}(z_0):=\{z\in\mathbb{C}:\,|z-z_0|<\epsilon\}$.
\end{prop}
The small norm estimates in Proposition \ref{snorm1} make precise the notion of $G_S(z),z\notin[-\sqrt{2}-s_{n\theta},\sqrt{2}]$ being close to the identity matrix as $n\rightarrow\infty$, uniformly in $\theta\in[0,1]$ and for fixed $s\in\mathbb{R},\alpha>-1,\beta\notin(-\infty,0)$. Consequently, RHP \ref{opRHP} is localized on the line segment $(-\sqrt{2}-s_{n\theta},\sqrt{2})$ and in the vicinities of its endpoints. As such we now move on to the necessary local analysis.
\subsection{The outer parametrix} As in \cite[page $148$]{V}, we let
\begin{equation*}
	j(z):=z+\big(z^2-1\big)^{\frac{1}{2}},\ \ \ \ z\in\mathbb{C}\setminus[-1,1]\ \ \ \ \ \ \textnormal{with}\ \ \ \ J(z)>0\ \ \textnormal{for}\ z>1,
\end{equation*}
denote the conformal map from $\mathbb{C}\setminus[-1,1]$ to the exterior of the unit disk $\mathbb{D}_1(0)$, with the principal branch for the complex square root chosen. In turn, the Szeg\H{o} function
\begin{equation}\label{a:18}
	\mathcal{D}(z)=\mathcal{D}(z;s,n,\theta):=\bigg(\frac{z-\sqrt{2}}{j\big(1+2(z-\sqrt{2})/(2\sqrt{2}+s_{n\theta})\big)}\bigg)^{\frac{\alpha}{2}},\ \ \ \ z\in\mathbb{C}\setminus[-\sqrt{2}-s_{n\theta},\sqrt{2}],
\end{equation}
with principal branches for all fractional exponents, is analytic off $[-\sqrt{2}-s_{n\theta},\sqrt{2}]\subset\mathbb{R}$ and it satisfies
\begin{equation*}
	\mathcal{D}_+(z)\mathcal{D}_-(z)=|z-\sqrt{2}|^{\alpha},\ \ \ z\in(-\sqrt{2}-s_{n\theta},\sqrt{2});\ \ \ \ \ \ \ \ \ \ \ \chi:=\lim_{z\rightarrow\infty}\mathcal{D}(z)=\bigg(\frac{2\sqrt{2}+s_{n\theta}}{4}\bigg)^{\frac{\alpha}{2}}.
\end{equation*}
The same Szeg\H{o} function \eqref{a:18} allows us to define
\begin{equation}\label{a:19}
	P(z;s,\alpha,n,\theta):=\chi^{\sigma_3}\frac{1}{2}\begin{bmatrix}\nu(z)+\nu(z)^{-1} & -\im\big(\nu(z)-\nu(z)^{-1}\big)\smallskip\\ \im\big(\nu(z)-\nu(z)^{-1}\big) & \nu(z)+\nu(z)^{-1}\end{bmatrix}\mathcal{D}(z)^{-\sigma_3},\ \ \ z\in\mathbb{C}\setminus[-\sqrt{2}-s_{n\theta},\sqrt{2}],
\end{equation}
where
\begin{equation*}
	\nu(z)=\nu(z;s,n,\theta):=\bigg(\frac{z-\sqrt{2}}{z+\sqrt{2}+s_{n\theta}}\bigg)^{\frac{1}{4}},\ \ \ z\in\mathbb{C}\setminus[-\sqrt{2}-s_{n\theta},\sqrt{2}],
\end{equation*}
also uses principal branches such that $\lim_{z\rightarrow+\infty}\nu(z)=1$. The relevant analytic and asymptotic properties of the outer parametrix \eqref{a:19} are summarized below.
\begin{problem}[Outer RHP]\label{outerArno} Let $s\in\mathbb{R},\alpha>-1,\theta\in[0,1]$ and suppose $\mathbb{N}\ni n\geq n_0$ is such that $-\sqrt{2}-s_{n\theta}<\sqrt{2}$. Then $P(z)=P(z;s,\alpha,n,\theta)\in\mathbb{C}^{2\times 2}$ defined in \eqref{a:19} has the following properties:
\begin{enumerate}
	\item[(1)] $z\mapsto P(z)$ is analytic for $z\in\mathbb{C}\setminus[-\sqrt{2}-s_{n\theta},\sqrt{2}]$, it admits continuous boundary values $P_{\pm}(z):=\lim_{\epsilon\downarrow 0}P(z\pm\im\epsilon)$ for $z\in(-\sqrt{2}-s_{n\theta},\sqrt{2})$ and $P_{\pm}\in L^2[-\sqrt{2}-s_{n\theta},\sqrt{2}]$.
	\item[(2)] The pointwise limiting values $P_{\pm}(z)$ on $(-\sqrt{2}-s_{n\theta},\sqrt{2})$ satisfy
	\begin{equation*}
		P_+(z)=P_-(z)\begin{bmatrix}0 & |z-\sqrt{2}|^{\alpha}\\ -|z-\sqrt{2}|^{-\alpha} & 0\end{bmatrix}.
	\end{equation*}
	\item[(3)] As $z\rightarrow\infty$ we have the asymptotic $P(z)=I+\sum_{k=1}^2P_kz^{-k}+\mathcal{O}(z^{-3})$ with $P_k=P_k(s,\alpha,n,\theta)$ equal
	\begin{align*}
		P_1=&\,\frac{1}{4}(2\sqrt{2}+s_{n\theta})\,\chi^{\sigma_3}\begin{bmatrix}\alpha & \im\smallskip\\ -\im & -\alpha\end{bmatrix}\chi^{-\sigma_3},\\
		P_2=&\,\frac{1}{8}(2\sqrt{2}+s_{n\theta})\,\chi^{\sigma_3}\Bigg\{\frac{1}{4}(2\sqrt{2}+s_{n\theta})\begin{bmatrix}1+\alpha^2& -2\im\alpha\smallskip\\ -2\im\alpha & 1+\alpha^2\end{bmatrix}+\begin{bmatrix}\frac{\alpha}{4}(2\sqrt{2}-3s_{n\theta}) & -\im s_{n\theta}\smallskip\\ \im s_{n\theta} & -\frac{\alpha}{4}(2\sqrt{2}-3s_{n\theta})\end{bmatrix}\Bigg\}\chi^{-\sigma_3}.
	\end{align*}
\end{enumerate}
\end{problem}
\subsection{The local parametrix at $z=\sqrt{2}$} The construction near $z=\sqrt{2}$ is not as explicit as \eqref{a:19} or as \eqref{a:28} and makes use of the model problem \ref{YattRHP} which links to Painlev\'e-XXXIV. Let $Q(\zeta)=Q(\zeta;x,\alpha,\beta)\in\mathbb{C}^{2\times 2}$ denote the unique solution of RHP \ref{YattRHP} for $x\in\mathbb{R},\alpha>-1$ and $\beta\in\mathbb{C}\setminus(-\infty,0)$. Fix $\epsilon>0$ small, let $s\in\mathbb{R},\theta\in[0,1]$ and $n\geq n_0$ sufficiently large so that $\sqrt{2}-s_{n\theta}\in\mathbb{D}_{\epsilon}(\sqrt{2})$. Having ensured the same we consider the map
\begin{equation}\label{a:20}
	\mathbb{D}_{\epsilon}(\sqrt{2})\ni z\mapsto\zeta^b(z):=\bigg[\frac{3\pi}{2}\int_{\sqrt{2}}^{z+s_{n\theta}}h_{V_{\theta}}(w)\big(w^2-2\big)^{\frac{1}{2}}\d w\bigg]^{\frac{2}{3}},\ \ \ \ \theta\in[0,1],
\end{equation}
with path of integration off $E$ and in the domain of analyticity of $h_{V_{\theta}}$, and with the principal branch for the complex square root with cut on $E$.
\begin{prop}\label{conf1} The map $z\mapsto\zeta^b(z)$ in \eqref{a:20}, defined with $s\in\mathbb{R},\theta\in[0,1]$ and $n\geq n_0$ so that $\sqrt{2}-s_{n\theta}\in\mathbb{D}_{\epsilon}(\sqrt{2})$, is conformal near $z=\sqrt{2}$ for any $\theta\in[0,1]$. Moreover, as $\omega= z+s_{n\theta}-\sqrt{2}\rightarrow 0$,
\begin{align*}
	\zeta^b(z)=\sqrt{2}\big(\pi h_{V_{\theta}}(\sqrt{2})\big)^{\frac{2}{3}}\omega&\,\Bigg[1+\frac{2}{5}\bigg\{\frac{h_{V_{\theta}}'(\sqrt{2})}{h_{V_{\theta}}(\sqrt{2})}+\frac{1}{4\sqrt{2}}\bigg\}\omega+\Bigg\{\frac{2}{7}\left\{\frac{h_{V_{\theta}}''(\sqrt{2})}{2h_{V_{\theta}}(\sqrt{2})}+\frac{1}{4\sqrt{2}}\frac{h_{V_{\theta}}'(\sqrt{2})}{h_{V_{\theta}}(\sqrt{2})}-\frac{1}{64}\right\}\\
	&\hspace{5cm}-\frac{1}{25}\bigg\{\frac{h_{V_{\theta}}'(\sqrt{2})}{h_{V_{\theta}}(\sqrt{2})}+\frac{1}{4\sqrt{2}}\bigg\}^2\Bigg\}\omega^2+\mathcal{O}\big(\omega^3\big)\Bigg],
\end{align*}
and by Taylor expansion,
\begin{equation*}
	\zeta^b(z)-\zeta^b(\sqrt{2})=\zeta_1^b(z-\sqrt{2})\left[1+\sum_{j=1}^2\zeta_{j+1}^b(z-\sqrt{2})^j+\mathcal{O}\big((z-\sqrt{2})^3\big)\right],\ \ \ z-\sqrt{2}\rightarrow 0,
\end{equation*}
with $z$-independent coefficients $\zeta_j^b$ such that, as $n\rightarrow\infty$, for any fixed $s\in\mathbb{R}$ but uniformly in $\theta\in[0,1]$,
\begin{align*}
	\zeta_1^b=&\,\sqrt{2}\big(\pi h_{V_{\theta}}(\sqrt{2})\big)^{\frac{2}{3}}\left[1+\frac{4s}{5(n\tau_{\theta})^{\frac{2}{3}}}\left\{\frac{h_{V_{\theta}}'(\sqrt{2})}{h_{V_{\theta}}(\sqrt{2})}+\frac{1}{4\sqrt{2}}\right\}+\mathcal{O}\big(n^{-\frac{4}{3}}\big)\right],\\
	\zeta_2^b=&\,\frac{2}{5}\left\{\frac{h_{V_{\theta}}'(\sqrt{2})}{h_{V_{\theta}}(\sqrt{2})}+\frac{1}{4\sqrt{2}}\right\}\Bigg[1+\frac{s}{(n\tau_{\theta})^{\frac{2}{3}}}\Bigg\{\frac{15}{7}\left\{\frac{h_{V_{\theta}}''(\sqrt{2})}{2h_{V_{\theta}}(\sqrt{2})}+\frac{1}{4\sqrt{2}}\frac{h_{V_{\theta}}'(\sqrt{2})}{h_{V_{\theta}}(\sqrt{2})}-\frac{1}{64}\right\}\\
	&\hspace{5.5cm}-\frac{11}{10}\bigg\{\frac{h_{V_{\theta}}'(\sqrt{2})}{h_{V_{\theta}}(\sqrt{2})}+\frac{1}{4\sqrt{2}}\bigg\}^2\Bigg\}\bigg\{\frac{h_{V_{\theta}}'(\sqrt{2})}{h_{V_{\theta}}(\sqrt{2})}+\frac{1}{4\sqrt{2}}\bigg\}^{-1}+\mathcal{O}\big(n^{-\frac{4}{3}}\big)\Bigg],\\
	\zeta_3^b=&\,\Bigg\{\frac{2}{7}\left\{\frac{h_{V_{\theta}}''(\sqrt{2})}{2h_{V_{\theta}}(\sqrt{2})}+\frac{1}{4\sqrt{2}}\frac{h_{V_{\theta}}'(\sqrt{2})}{h_{V_{\theta}}(\sqrt{2})}-\frac{1}{64}\right\}
	-\frac{1}{25}\bigg\{\frac{h_{V_{\theta}}'(\sqrt{2})}{h_{V_{\theta}}(\sqrt{2})}+\frac{1}{4\sqrt{2}}\bigg\}^2\Bigg\}\Bigg[1+\mathcal{O}\big(n^{-\frac{2}{3}}\big)\Bigg].
\end{align*}
\end{prop}
Next, we employ \eqref{a:20} in the definition of the parametrix at $z=\sqrt{2}$: let $s\in\mathbb{R},\theta\in[0,1]$ and $n\geq n_0$ once more such that $\sqrt{2}-s_{n\theta}\in\mathbb{D}_{\epsilon}(\sqrt{2})$ with $\epsilon>0$ small. Using $Q(\zeta)=Q(\zeta;x,\alpha,\beta)$ of RHP \ref{YattRHP}, consider
\begin{equation}\label{a:21}
	M(z;s,\alpha,\beta,n,\theta):=E^b(z)Q\Big(n^{\frac{2}{3}}\big(\zeta^b(z)-\zeta^b(\sqrt{2})\big);n^{\frac{2}{3}}\zeta^b(\sqrt{2}),\alpha,\beta\Big)\exp\left[\frac{2n}{3}\big(\zeta^b(z)\big)^{\frac{3}{2}}\sigma_3\right](z-\sqrt{2})^{-\frac{\alpha}{2}\sigma_3},
\end{equation}
for $z\in\mathbb{D}_{\epsilon}(\sqrt{2})\setminus\Sigma_S$ and $\theta\in[0,1]$. All fractional exponents in \eqref{a:21} are taken with their principal branch and the multiplier $E^b(z)=E^b(z;s,\alpha,n,\theta)$ equals
\begin{equation}\label{a:22}
	E^b(z):=P(z)(z-\sqrt{2})^{\frac{\alpha}{2}\sigma_3}\e^{\im\frac{\pi}{4}\sigma_3}\frac{1}{\sqrt{2}}\begin{bmatrix}1 & -1\\ 1 & 1\end{bmatrix}\Big(n^{\frac{2}{3}}\big(\zeta^b(z)-\zeta^b(\sqrt{2})\big)\Big)^{\frac{1}{4}\sigma_3},\  z\in\mathbb{D}_{\epsilon}(\sqrt{2})\setminus[-\sqrt{2}-s_{n\theta},\sqrt{2}],
\end{equation}
with $P(z)=P(z;s,\alpha,n,\theta)$ from \eqref{a:19} and principal branches throughout. A priori, $z\mapsto E^b(z)$ is analytic at $z=\sqrt{2}$ off $[-\sqrt{2}-s_{n\theta},\sqrt{2}]\subset\mathbb{R}$, but things improve at second glance.
\begin{prop} Let $s\in\mathbb{R},\alpha>-1,\theta\in[0,1]$ and $n\geq n_0$ so that $\sqrt{2}-s_{n\theta}\in\mathbb{D}_{\epsilon}(\sqrt{2})$ with $\epsilon>0$ small. The function $z\mapsto E^b(z)$ defined in \eqref{a:22} is analytic at $z=\sqrt{2}$ for all $\theta\in[0,1]$. Moreover, as $z\rightarrow \sqrt{2}$,
\begin{equation}\label{a:23}
	E^b(z)=E_0^b+E_1^b(z-\sqrt{2})+\mathcal{O}\big((z-\sqrt{2})^2\big),%E_2^b(z-\sqrt{2})^2+\mathcal{O}\big((z-\sqrt{2})^3\big),
\end{equation}
with $z$-independent, $2\times 2$ matrix-valued, coefficients
\begin{align*}
	E_0^b=&\,\chi^{\sigma_3}\e^{\im\frac{\pi}{4}\sigma_3}\frac{1}{\sqrt{2}}\begin{bmatrix}1 & -1-\alpha\\ 1 & \ \,1-\alpha\end{bmatrix}\Big(n^{\frac{2}{3}}\zeta_1^b(2\sqrt{2}+s_{n\theta})\Big)^{\frac{1}{4}\sigma_3},\\
	E_1^b=&\,\chi^{\sigma_3}\e^{\im\frac{\pi}{4}\sigma_3}\frac{1}{\sqrt{2}}\Bigg(\frac{1}{2\sqrt{2}+s_{n\theta}}\begin{bmatrix}\frac{1}{4}+\alpha+\frac{1}{2}\alpha^2 & \ \ \,\frac{1}{4}-\frac{1}{12}\alpha-\frac{1}{2}\alpha^2-\frac{1}{6}\alpha^3\smallskip\\
	\frac{1}{4}-\alpha+\frac{1}{2}\alpha^2 & -\frac{1}{4}-\frac{1}{12}\alpha+\frac{1}{2}\alpha^2-\frac{1}{6}\alpha^3\end{bmatrix}\\
	&\hspace{8cm}+\frac{\zeta_2^b}{4}\begin{bmatrix}1 & \ \ \,1+\alpha\smallskip\\ 1 & -1+\alpha\end{bmatrix}\Bigg)\Big(n^{\frac{2}{3}}\zeta_1^b(2\sqrt{2}+s_{n\theta})\Big)^{\frac{1}{4}\sigma_3}.
\end{align*}
%\begin{align*}
%	&E_2^b=\frac{\chi^{\sigma_3}\e^{\im\frac{\pi}{4}\sigma_3}}{\sqrt{2}(b-a)^2}\Bigg(\begin{bmatrix}-\frac{3}{32}-\frac{5}{12}\alpha-\frac{1}{24}\alpha^2+\frac{1}{6}\alpha^3+\frac{1}{24}\alpha^4 & -\frac{5}{32}+\frac{29}{480}\alpha+\frac{7}{24}\alpha^2+\frac{1}{24}\alpha^3-\frac{1}{24}\alpha^4-\frac{1}{120}\alpha^5\smallskip\\ -\frac{3}{32}+\frac{5}{12}\alpha-\frac{1}{24}\alpha^2-\frac{1}{6}\alpha^3+\frac{1}{24}\alpha^4 & \,\,\,\,\,\frac{5}{32}+\frac{29}{480}\alpha-\frac{7}{24}\alpha^2+\frac{1}{24}\alpha^3+\frac{1}{24}\alpha^4-\frac{1}{120}\alpha^5\end{bmatrix}\\
%	&\hspace{0.25cm}+\frac{\zeta_2^b}{4}(b-a)\begin{bmatrix}\frac{1}{4}+\alpha+\frac{1}{2}\alpha^2&-\frac{1}{4}+\frac{1}{12}\alpha+\frac{1}{2}\alpha^2+\frac{1}{6}\alpha^3\smallskip\\ \frac{1}{4}-\alpha+\frac{1}{2}\alpha^2 & \,\,\,\,\frac{1}{4}+\frac{1}{12}\alpha-\frac{1}{2}\alpha^2+\frac{1}{6}\alpha^3\end{bmatrix}+\frac{1}{4}\begin{bmatrix}1&-1-\alpha\\ 1&1-\alpha\end{bmatrix}\begin{bmatrix}-\frac{3}{8}(\zeta_2^b)^2+\zeta_3^b&0\smallskip\\ 0&\frac{5}{8}(\zeta_2^b)^2-\zeta_3^b\end{bmatrix}\\
%	&\hspace{0.5cm}\times(b-a)^2\Bigg)\Big(n^{\frac{2}{3}}\zeta_1^b(b-a)\Big)^{\frac{1}{4}\sigma_3}.
%\end{align*}
\end{prop}
\begin{proof} Seeing that $\zeta_1^b>0$ for $n\geq n_0$, compare Proposition \ref{conf1}, the Taylor expansion of $\zeta^b(z)-\zeta^b(\sqrt{2})$ combined with RHP \ref{outerArno} yields for $z\in(-\sqrt{2}-s_{n\theta},\sqrt{2})\cap\mathbb{D}_{\epsilon}(\sqrt{2})$,
\begin{align*}
	E_+^b(z)=&\,\lim_{\delta\downarrow 0}E^b(z+\im\delta)\\
	=&\,P_-(z)|z-\sqrt{2}|^{\frac{\alpha}{2}\sigma_3}\begin{bmatrix}0 & 1\\ -1 & 0\end{bmatrix}\e^{\im\frac{\pi}{2}\alpha\sigma_3}\e^{\im\frac{\pi}{4}\sigma_3}\frac{1}{\sqrt{2}}\begin{bmatrix}1 & -1\\ 1 & 1\end{bmatrix}\e^{\im\frac{\pi}{2}\sigma_3}\Big(n^{\frac{2}{3}}\big(\zeta^b(z)-\zeta^b(\sqrt{2})\big)\Big)_-^{\frac{1}{4}\sigma_3}\\
	=&\,P_-(z)(z-\sqrt{2})_-^{\frac{\alpha}{2}\sigma_3}\e^{\im\frac{\pi}{4}\sigma_3}\frac{1}{\sqrt{2}}\begin{bmatrix}1 & -1\\ 1 & 1\end{bmatrix}\Big(n^{\frac{2}{3}}\big(\zeta^b(z)-\zeta^b(\sqrt{2})\big)\Big)_-^{\frac{1}{4}\sigma_3}=\lim_{\delta\downarrow 0}E^b(z-\im\delta)=E_-^b(z).
\end{align*}
Thus, $z\mapsto E^b(z)$ extends analytically to $\mathbb{D}_{\epsilon}(\sqrt{2})\setminus\{\sqrt{2}\}$ and since $|z-\sqrt{2}|^{\frac{1}{2}}\|E^b(z)\|$ remains bounded as $z\rightarrow \sqrt{2}$, see \eqref{a:18},\eqref{a:19} and \eqref{a:22}, $z\mapsto E^b(z)$ extends analytically to all of $\mathbb{D}_{\epsilon}(\sqrt{2})$. Expansion \eqref{a:23} affirms the same local analyticity and it follows from straightforward calculations using the Puiseux expansion
\begin{align*}
	P(z)&(z-\sqrt{2})^{\frac{\alpha}{2}\sigma_3}=\chi^{\sigma_3}\e^{\im\frac{\pi}{4}\sigma_3}\frac{1}{\sqrt{2}}\bigg\{\omega^{-\frac{1}{4}}\begin{bmatrix}\varsigma(\omega) & 0\smallskip\\ \varsigma(\omega) & 0\end{bmatrix}+\omega^{\frac{1}{4}}\begin{bmatrix}0 & -\iota(\omega)-\alpha\varsigma(\omega)\smallskip\\ 0 & \ \ \iota(\omega)-\alpha\varsigma(\omega)\end{bmatrix}+\omega^{\frac{3}{4}}\begin{bmatrix}\ \ \,\alpha\iota(\omega)+\frac{1}{2}\alpha^2\varsigma(\omega) & 0\smallskip\\ -\alpha\iota(\omega)+\frac{1}{2}\alpha^2\varsigma(\omega) & 0\end{bmatrix}\\
	&\,+\omega^{\frac{5}{4}}\begin{bmatrix}0 & -\frac{1}{2}\alpha^2\iota(\omega)+\frac{1}{6}\alpha(1-\alpha^2)\varsigma(\omega)\smallskip\\
	0 & \ \ \frac{1}{2}\alpha^2\iota(\omega)+\frac{1}{6}\alpha(1-\alpha^2)\varsigma(\omega)\end{bmatrix}+\omega^{\frac{7}{4}}\begin{bmatrix}
	-\frac{1}{6}\alpha(1-\alpha^2)\iota(\omega)-\frac{1}{24}\alpha^2(4-\alpha^2)\varsigma(\omega) &0\smallskip\\
	\ \,\,\frac{1}{6}\alpha(1-\alpha^2)\iota(\omega)-\frac{1}{24}\alpha^2(4-\alpha^2)\varsigma(\omega)&0\end{bmatrix}\\
	&+\omega^{\frac{9}{4}}\begin{bmatrix}0&\,\,\,\,\frac{1}{24}\alpha^2(4-\alpha^2)\iota(\omega)-\alpha(\frac{3}{40}-\frac{1}{12}\alpha^2+\frac{1}{120}\alpha^4)\varsigma(\omega)\smallskip\\
	0&-\frac{1}{24}\alpha^2(4-\alpha^2)\iota(\omega)-\alpha(\frac{3}{40}-\frac{1}{12}\alpha^2+\frac{1}{120}\alpha^4)\varsigma(\omega)\end{bmatrix}+\mathcal{O}\big(\omega^{\frac{11}{4}}\big)\Bigg\}\frac{1}{\sqrt{2}}\begin{bmatrix}1 & 1\\ -1 & 1 \end{bmatrix}\e^{-\im\frac{\pi}{4}\sigma_3},%\ \ \ \ \ \ \omega=\frac{z-b}{b-a}\downarrow 0,
\end{align*}
for $\omega=\frac{z-\sqrt{2}}{2\sqrt{2}+s_{n\theta}}\downarrow 0$ and where $\omega\mapsto\varsigma(\omega):=(1+\omega)^{\frac{1}{4}}$ and $\omega\mapsto\iota(\omega):=(1+\omega)^{-\frac{1}{4}}$ are analytic at $\omega=0$.
\end{proof}
Having established the analytic properties of $z\mapsto E^b(z)$ we now summarize the relevant analytic and asymptotic properties of $z\mapsto M(z)$ as defined in \eqref{a:21}.
\begin{problem}[Local RHP at $z=\sqrt{2}$]\label{localb} Let $s\in\mathbb{R},\alpha>-1,\beta\notin(-\infty,0),\theta\in[0,1]$ and $n\geq n_0$ so that $\sqrt{2}-s_{n\theta}\in\mathbb{D}_{\epsilon}(\sqrt{2})$ with $\epsilon>0$ small. The function $M(z)=M(z;s,\alpha,\beta,n,\theta)\in\mathbb{C}^{2\times 2}$ defined in \eqref{a:21} has the following properties:
\begin{enumerate}
	\item[(1)] $z\mapsto M(z)$ is analytic for $z\in\mathbb{D}_{\epsilon}(\sqrt{2})\setminus\Sigma_S$, possibly after a local contour deformation near $z=\sqrt{2}$ in Figure \ref{fig1}, and extends continuously to the closure of $\mathbb{D}_{\epsilon}(\sqrt{2})\setminus\Sigma_S$ away from $z=\sqrt{2}$.
	\item[(2)] The non-tangential limiting values $M_{\pm}(z)$ for $z\in\Sigma_S\setminus\{\sqrt{2}\}$ near $z=\sqrt{2}$ obey
	\begin{equation*}
		M_+(z)=M_-(z)G_S(z),
	\end{equation*}
	with $G_S(z)=G_S(z;s,\alpha,\beta,n,\theta)$ exactly as specified in condition $(2)$ of RHP \ref{opRHP}.
	\item[(3)] As $z\rightarrow\sqrt{2},z\notin\Sigma_S$, $z\mapsto M(z)$ is weakly singular, however $\mathbb{D}_{\epsilon}(\sqrt{2})\ni z\mapsto S(z)M(z)^{-1}$ is analytic.
	\item[(4)] As $n\rightarrow\infty$ with $s\in\mathbb{R},\alpha>-1,\beta\notin(-\infty,0)$, the function $M(z)$ in \eqref{a:21} matches onto $P(z)$ in \eqref{a:19} as follows,
	\begin{equation}\label{a:25}
		M(z)=\left\{I+\sum_{k=1}^3M_k(z;s,\alpha,\beta,n,\theta)n^{-\frac{k}{3}}+\mathcal{O}\big(n^{-\frac{4}{3}}\big)\right\}P(z),
	\end{equation}
	uniformly in $0<r_1\leq|z-\sqrt{2}|\leq r_2<\frac{\epsilon}{2}$ and $\theta\in[0,1]$ for any fixed $r_1,r_2$. The coefficients $M_k(z)=M_k(z;s,\alpha,\beta,n,\theta)$ in \eqref{a:25} are
	bounded in $n$ and equal, with the shorthand $\hbar(z):=\zeta^b(z)-\zeta^b(\sqrt{2})$,
	\begin{align*}
		M_1(z)=&\,E_{\ast}^b(z)\begin{bmatrix}0 &1\smallskip\\ 0 & 0\end{bmatrix}E_{\ast}^{b}(z)^{-1}\frac{Q_1^{12}}{\hbar(z)}+D_1(z),\\
		M_2(z)=&\,E_{\ast}^b(z)\begin{bmatrix}1 & 0\\
		0 & -1\end{bmatrix}E_{\ast}^b(z)^{-1}\frac{Q_1^{11}}{\hbar(z)}+D_2(z)+E_{\ast}^b(z)\begin{bmatrix}0 & 1\\ 0 & 0\end{bmatrix}E_{\ast}^b(z)^{-1}\frac{D_1(z)\,Q_1^{12}}{\hbar(z)},\\
		M_3(z)=&\,E_{\ast}^b(z)\begin{bmatrix}0 & 0\\ 1 &0\end{bmatrix}E_{\ast}^b(z)^{-1}\frac{Q_1^{21}}{\hbar(z)}+E_{\ast}^b(z)\begin{bmatrix}0 & 1\\ 0 & 0\end{bmatrix}E_{\ast}^b(z)^{-1}\frac{Q_2^{12}}{\hbar^2(z)}+D_3(z)\\
		&\hspace{0.5cm}+E_{\ast}^b(z)\begin{bmatrix}0 & 1\\ 0 & 0\end{bmatrix}E_{\ast}^b(z)^{-1}\frac{D_2(z)\,Q_1^{12}}{\hbar(z)}+E_{\ast}^b(z)\begin{bmatrix}1 & 0\\
		0 & -1\end{bmatrix}E_{\ast}^b(z)^{-1}\frac{D_1(z)\,Q_1^{11}}{\hbar(z)}.
	\end{align*}
	Here $E_{\ast}^b(z):=E^b(z)n^{-\frac{1}{6}\sigma_3}$, we write $Q_k^{ij}=Q_k^{ij}(n^{\frac{2}{3}}\zeta^b(\sqrt{2}),\alpha,\beta)$ for the scalar entries of the coefficients $Q_k(x,\alpha,\beta)$ in condition $(4)$ of RHP \ref{YattRHP} and $D_k(z)$ is made explicit in \eqref{nasty3}.
\end{enumerate}
\end{problem}
The jump behavior in condition $(2)$ above follows from RHP \ref{YattRHP}, condition $(2)$, the fact that $z\mapsto E^b(z)$ is analytic at $z=\sqrt{2}$ and the definition of $z\mapsto\zeta^b(z)$ in \eqref{a:20}, compare Proposition \ref{conf1}. Likewise, the singular behavior of $z\mapsto M(z)$ at $z=\sqrt{2}$ is a consequence of condition $(3)$ in RHP \ref{YattRHP} and of the structure of $M(z)$ in \eqref{a:21}. The matching \eqref{a:25} is less straightforward and we devote Appendix \ref{tedious} to its derivation.

\begin{rem}\label{unif} While \eqref{a:25} is uniform in $(s,\alpha,\beta)\in\mathbb{R}\times(-1,\infty)\times(\mathbb{C}\setminus(-\infty,0))$ on compact sets, once $\alpha=0$, \cite[Section $4.5$]{BCI} establishes uniformity of \eqref{a:25} also for $s\geq s_0$. This will be useful later on.
\end{rem}

\subsection{The local parametrix at $z=-\sqrt{2}-s_{n\theta}$} The construction near $z=-\sqrt{2}-s_{n\theta}$ is standard and involves the Airy parametrix: let $\textnormal{Ai}(\zeta)$ denote the Airy function, $\textnormal{Ai}'(\zeta)$ its derivative and assemble, with the help of Figure \ref{fig2},
\begin{equation}\label{a:26}
	A(\zeta):=\sqrt{2\pi}\,\e^{-\im\frac{\pi}{4}}\begin{bmatrix}\textnormal{Ai}(\zeta) & \e^{\im\frac{\pi}{3}}\textnormal{Ai}(\e^{-\im\frac{2\pi}{3}}\zeta)\smallskip\\ \textnormal{Ai}'(\zeta) & \e^{-\im\frac{\pi}{3}}\textnormal{Ai}'(\e^{-\im\frac{2\pi}{3}}\zeta)\end{bmatrix}\begin{cases}I,&\zeta\in\Omega_1\\ 
	\bigl[\begin{smallmatrix}1 & 0\\ -1 & 1\end{smallmatrix}\bigr],&\zeta\in\Omega_2\\ 
	\bigl[\begin{smallmatrix}1 & -1\\ 0 & 1\end{smallmatrix}\bigr],&\zeta\in\Omega_4\\
	\bigl[\begin{smallmatrix}1 & -1\\ 0 & 1\end{smallmatrix}\bigr]\bigl[\begin{smallmatrix}1 & 0\\ 1 & 1\end{smallmatrix}\bigr],&\zeta\in\Omega_3
	\end{cases}.
\end{equation}
Then $A(\zeta)$ is uniquely characterized by the below four properties.
\begin{problem}\label{Air} The function $A(\zeta)\in\mathbb{C}^{2\times 2}$ defined in \eqref{a:26} is such that:
\begin{enumerate}
	\item[(1)] $\zeta\mapsto A(\zeta)$ is analytic for $\zeta\in\mathbb{C}\setminus\Sigma_Q$ with $\Sigma_Q$ shown in Figure \ref{fig2}. On $\Sigma_Q$, $A(\zeta)$ admits continuous limiting values $A_{\pm}(\zeta)$ as one approaches $\Sigma_Q$ from either sides of $\mathbb{C}\setminus\Sigma_Q$.
	\item[(2)] The non-tangential limiting values $A_{\pm}(\zeta)$ on $\Sigma_Q\ni\zeta$ satisfy $A_+(\zeta)=A_-(\zeta)G_A(\zeta)$ with $G_A(\zeta)=G_Q(\zeta;0,1)$ as in condition $(2)$ of RHP \ref{YattRHP}. 
	\item[(3)] $\zeta\mapsto A(\zeta)$ is bounded as $\zeta\rightarrow 0$ and $\zeta\notin\Sigma_Q$.
	\item[(4)] As $\zeta\rightarrow\infty$ and $\zeta\notin\Sigma_Q$, $A(\zeta)$ is normalized as follows,
	\begin{equation*}
		A(\zeta)=\Bigg\{I-\frac{7}{48\zeta}\begin{bmatrix}0 & 0\\ 1 & 0\end{bmatrix}+\frac{5}{48\zeta^2}\begin{bmatrix}0 & 1\\ 0 & 0\end{bmatrix}+\mathcal{O}\big(\zeta^{-3}\big)\Bigg\}\zeta^{-\frac{1}{4}\sigma_3}\frac{1}{\sqrt{2}}\begin{bmatrix}1 & 1\\ -1 & 1\end{bmatrix}\e^{-\im\frac{\pi}{4}\sigma_3}\e^{-\frac{2}{3}\zeta^{\frac{3}{2}}\sigma_3}.
	\end{equation*}
\end{enumerate}
\end{problem}
Equipped with $A(\zeta)$ in \eqref{a:26}, we fix $\epsilon>0$ small, let $s\in\mathbb{R},\theta\in[0,1]$ and consider the map
\begin{equation}\label{a:27}
	\mathbb{D}_{\epsilon}(a)\ni z\mapsto\zeta^a(z):=\e^{\im\pi}\left[\frac{3\pi}{2}\int_{-\sqrt{2}}^{z+s_{n\theta}}h_{V_{\theta}}(w)\big(w^2-2\big)^{\frac{1}{2}}\d w\right]^{\frac{2}{3}},\ \ \ \ \ \ \theta\in[0,1],
\end{equation}
with path of integration off $E$ and in the domain of analyticity of $h_{V_{\theta}}$, and the principal branch for the complex square root with cut on $E$.
\begin{prop}\label{conf2} The map $z\mapsto\zeta^a(z)$ in \eqref{a:27}, defined with parameters $s\in\mathbb{R},\theta\in[0,1]$ and $n\geq n_0$, is conformal near $z=-\sqrt{2}-s_{n\theta}$ for any $\theta\in[0,1]$. Moreover, as $\omega=z+\sqrt{2}+s_{n\theta}\rightarrow 0$,
\begin{align*}
	\zeta^a(z)=\sqrt{2}\big(\pi h_{V_{\theta}}(-\sqrt{2})\big)^{\frac{2}{3}}\omega&\Bigg[1+\frac{2}{5}\left\{\frac{h_{V_{\theta}}'(-\sqrt{2})}{h_{V_{\theta}}(-\sqrt{2})}-\frac{1}{4\sqrt{2}}\right\}\omega+\mathcal{O}\big(\omega^2\big)\Bigg].
	%\Bigg\{\frac{2}{7}\left\{\frac{h_{V_{\theta}}''(-\sqrt{2})}{2h_{V_{\theta}}(-\sqrt{2})}-\frac{1}{4\sqrt{2}}\frac{h_{V_{\theta}}'(-\sqrt{2})}{h_{V_{\theta}}(-\sqrt{2})}-\frac{1}{64}\right\}\\
	%&\hspace{4.5cm}-\frac{1}{25}\left\{\frac{h_{V_{\theta}}'(-\sqrt{2})}{h_{V_{\theta}}(-\sqrt{2})}-\frac{1}{4\sqrt{2}}\right\}^2\Bigg\}\omega^2+\mathcal{O}\big(\omega^3\big)\Bigg].
\end{align*}
\end{prop}
Next, we utilize \eqref{a:27} in the definition of the parametrix at $z=-\sqrt{2}-s_{n\theta}$: let $s\in\mathbb{R},\theta\in[0,1],n\geq n_0$ and $\epsilon>0$ small. Using $A(\zeta)$ of RHP \ref{Air}, consider
\begin{equation}\label{a:28}
	N(z;s,\alpha,n,\theta):=E^a(z)A\Big(n^{\frac{2}{3}}\e^{-\im\pi}\zeta^a(z)\Big)\sigma_3\exp\left[\frac{2n}{3}\big(\e^{-\im\pi}\zeta^a(z)\big)^{\frac{3}{2}}\sigma_3\right](\sqrt{2}-z)^{-\frac{\alpha}{2}\sigma_3},
\end{equation}
for $z\in\mathbb{D}_{\epsilon}(-\sqrt{2}-s_{n\theta})\setminus\Sigma_S$ and $\theta\in[0,1]$. All fractional exponents in \eqref{a:28} are taken with their principal branch and the multiplier $E^a(z)=E^a(z;s,\alpha,n,\theta)$ equals
\begin{equation}\label{a:29}
	E^a(z):=P(z)(\sqrt{2}-z)^{\frac{\alpha}{2}\sigma_3}\e^{\im\frac{\pi}{4}\sigma_3}\frac{\sigma_3}{\sqrt{2}}\begin{bmatrix}1 & -1\\ 1 & 1\end{bmatrix}\Big(n^{\frac{2}{3}}\e^{-\im\pi}\zeta^a(z)\Big)^{\frac{1}{4}\sigma_3},\ \ z\in\mathbb{D}_{\epsilon}(-\sqrt{2}-s_{n\theta})\setminus[-\sqrt{2}-s_{n\theta},\sqrt{2}],
\end{equation}
with $P(z)=P(z;s,\alpha,n,\theta)$ from \eqref{a:19} and principal branches. Here is the analogue of \eqref{a:23} for $E^a(z)$.
\begin{prop} Let $s\in\mathbb{R},\alpha>-1,\theta\in[0,1]$ and $n\geq n_0$. The function $z\mapsto E^a(z)$ defined in \eqref{a:29} is analytic at $z=-\sqrt{2}-s_{n\theta}$ for all $\theta\in[0,1]$. Moreover, as $z\rightarrow -\sqrt{2}-s_{n\theta}$,
\begin{equation}\label{a:30}
	E^a(z)=E_0^a+E_1^a(z+\sqrt{2}+s_{n\theta})+\mathcal{O}\big((z+\sqrt{2}+s_{n\theta})^2\big),
\end{equation}
with $z$-independent, $2\times 2$ matrix-valued, coefficients
\begin{align*}
	E_0^a=&\,\chi^{\sigma_3}\e^{\im\frac{\pi}{4}\sigma_3}\frac{\sigma_3}{\sqrt{2}}\begin{bmatrix}1 & -1-\alpha\\ 1 & \ \ 1-\alpha\end{bmatrix}\Big(n^{\frac{2}{3}}\sqrt{2}\big(\pi h_{V_{\theta}}(-\sqrt{2})\big)^{\frac{2}{3}}(2\sqrt{2}+s_{n\theta})\Big)^{\frac{1}{4}\sigma_3},\\
	E_1^a=&\,\chi^{\sigma_3}\e^{\im\frac{\pi}{4}\sigma_3}\frac{\sigma_3}{\sqrt{2}}\Bigg(\frac{1}{-2\sqrt{2}-s_{n\theta}}\begin{bmatrix}\frac{1}{4}+\alpha+\frac{1}{2}\alpha^2 & \ \ \,\frac{1}{4}-\frac{1}{12}\alpha-\frac{1}{2}\alpha^2-\frac{1}{6}\alpha^3\smallskip\\
	\frac{1}{4}-\alpha+\frac{1}{2}\alpha^2 & -\frac{1}{4}-\frac{1}{12}\alpha+\frac{1}{2}\alpha^2-\frac{1}{6}\alpha^3\end{bmatrix}+\frac{1}{10}\begin{bmatrix}1 &\ \ 1+\alpha\\ 1 & -1+\alpha\end{bmatrix}\\
	&\hspace{5cm}\times\left\{\frac{h_{V_{\theta}}'(-\sqrt{2})}{h_{V_{\theta}}(-\sqrt{2})}-\frac{1}{4\sqrt{2}}\right\}\Bigg)\Big(n^{\frac{2}{3}}\sqrt{2}\big(\pi h_{V_{\theta}}(-\sqrt{2})\big)^{\frac{2}{3}}(2\sqrt{2}+s_{n\theta})\Big)^{\frac{1}{4}\sigma_3}.
\end{align*}
\end{prop}
\begin{proof} We have $\e^{-\im\pi}\zeta^a(z)>0$ for $z<-\sqrt{2}-s_{n\theta}$, see Proposition \ref{conf2}, and so with RHP \ref{outerArno}, for $z\in(-\sqrt{2}-s_{n\theta},\sqrt{2})\cap\mathbb{D}_{\epsilon}(-\sqrt{2}-s_{n\theta})$,
\begin{align*}
	E_+^a(z)=&\,\lim_{\delta\downarrow 0}E^a(z+\im\delta)\\
	=&\,P_-(z)(\sqrt{2}-z)^{\frac{\alpha}{2}\sigma_3}\begin{bmatrix}0 & 1\\ -1 &0\end{bmatrix}\e^{\im\frac{\pi}{4}\sigma_3}\frac{\sigma_3}{\sqrt{2}}\begin{bmatrix}1 & -1\\ 1 & 1\end{bmatrix}\e^{-\im\frac{\pi}{2}\sigma_3}\Big(n^{\frac{2}{3}}\e^{-\im\pi}\zeta^a(z)\Big)^{\frac{1}{4}\sigma_3}_-\\
	=&\,P_-(z)(\sqrt{2}-z)^{\frac{\alpha}{2}\sigma_3}\e^{\im\frac{\pi}{4}\sigma_3}\frac{\sigma_3}{\sqrt{2}}\begin{bmatrix}1 & -1\\ 1 & 1\end{bmatrix}\Big(n^{\frac{2}{3}}\e^{-\im\pi}\zeta^a(z)\Big)^{\frac{1}{4}\sigma_3}_-=\lim_{\delta\downarrow 0}E^a(z-\im\delta)=E_-^a(z).
\end{align*}
Thus, $z\mapsto E^a(z)$ extends analytically to $\mathbb{D}_{\epsilon}(-\sqrt{2}-s_{n\theta})\setminus\{-\sqrt{2}-s_{n\theta}\}$ and since $|z+\sqrt{2}+s_{n\theta}|^{\frac{1}{2}}\|E^a(z)\|$ remains bounded as $z\rightarrow -\sqrt{2}-s_{n\theta}$, see \eqref{a:18}, \eqref{a:10} and \eqref{a:29}, $z\mapsto E^a(z)$ extends analytically to all of $\mathbb{D}_{\epsilon}(-\sqrt{2}-s_{n\theta})$. The Taylor expansion \eqref{a:30} confirms the same and it follows from straightforward calculations using the Puiseux expansion
\begin{align*}
	P(z)&(\sqrt{2}-z)^{\frac{\alpha}{2}\sigma_3}=\chi^{\sigma_3}\e^{\im\frac{\pi}{4}\sigma_3}\frac{\sigma_3}{\sqrt{2}}\Bigg\{\omega^{-\frac{1}{4}}\begin{bmatrix}\varsigma(\omega) & 0\smallskip\\ \varsigma(\omega)&0\end{bmatrix}+\omega^{\frac{1}{4}}\begin{bmatrix}0 & -\iota(w)-\alpha\varsigma(\omega)\smallskip\\ 0 & \ \ \,\iota(\omega)-\alpha\varsigma(\omega)\end{bmatrix}+\omega^{\frac{3}{4}}\begin{bmatrix}\ \ \,\alpha\iota(\omega)+\frac{1}{2}\alpha^2\varsigma(\omega) & 0\smallskip\\ -\alpha\iota(\omega)+\frac{1}{2}\alpha^2\varsigma(\omega) & 0\end{bmatrix}\\
	&\,\,+\omega^{\frac{5}{4}}\begin{bmatrix}0 & -\frac{1}{2}\alpha^2\iota(\omega)+\frac{1}{6}\alpha(1-\alpha^2)\varsigma(\omega)\smallskip\\ 0 & \ \ \,\frac{1}{2}\alpha^2\iota(\omega)+\frac{1}{6}\alpha(1-\alpha^2)\varsigma(\omega)\end{bmatrix}+\mathcal{O}\big(\omega^{\frac{7}{4}}\big)\Bigg\}\frac{1}{\sqrt{2}}\begin{bmatrix}1 & 1\\ -1 & 1\end{bmatrix}\sigma_3\,\e^{-\im\frac{\pi}{4}\sigma_3},\ \ \omega=\frac{-\sqrt{2}-s_{n\theta}-z}{2\sqrt{2}+s_{n\theta}}\downarrow 0,
\end{align*}
where, as before, $\omega\mapsto\varsigma(\omega)=(1+\omega)^{\frac{1}{4}}$ and $\omega\mapsto\iota(\omega)=(1+\omega)^{-\frac{1}{4}}$.
\end{proof}
Keeping in mind the analytic properties of $z\mapsto E^a(z)$, we now summarize the relevant analytic and asymptotic properties of $z\mapsto N(z)$ as defined in \eqref{a:28}.
\begin{problem}[Local RHP at $z=-\sqrt{2}-s_{n\theta}$]\label{locala} Let $s\in\mathbb{R},\alpha>-1,\theta\in[0,1],n\geq n_0$ and $\epsilon>0$ small. The function $N(z)=N(z;s,\alpha,n,\theta)\in\mathbb{C}$ defined in \eqref{a:28} has the following properties:
\begin{enumerate}
	\item[(1)] $z\mapsto N(z)$ is analytic for $z\in\mathbb{D}_{\epsilon}(-\sqrt{2}-s_{n\theta})\setminus\Sigma_S$, possibly after a local contour deformation near $z=-\sqrt{2}-s_{n\theta}$ in Figure \ref{fig1}, and extends continuously to the closure of $\mathbb{D}_{\epsilon}(-\sqrt{2}-s_{n\theta})\setminus\Sigma_S$ away from $z=-\sqrt{2}-s_{n\theta}$.
	\item[(2)] The non-tangential limiting values $N_{\pm}(z)$ for $z\in\Sigma_S\setminus\{-\sqrt{2}-s_{n\theta}\}$ near $z=-\sqrt{2}-s_{n\theta}$ obey
	\begin{equation*}
		N_+(z)=N_-(z)G_S(z),
	\end{equation*}
	with $G_S(z)=G_S(z;s,\alpha,\beta,n,\theta)$ exactly as specified in condition $(2)$ of RHP \ref{opRHP}.
	\item[(3)] $z\mapsto N(z)$ is bounded as $z\rightarrow -\sqrt{2}-s_{n\theta},z\notin\Sigma_S$ and $\mathbb{D}_{\epsilon}(-\sqrt{2}-s_{n\theta})\ni z\mapsto S(z)N(z)^{-1}$ is analytic.
	\item[(4)] As $n\rightarrow\infty$ with $s\in\mathbb{R},\alpha>-1$, the function $N(z)$ in \eqref{a:28} matches onto $P(z)$ in \eqref{a:19} as follows,
	\begin{equation}\label{a:31}
		N(z)=\left\{I+N_1(z;s,\alpha,n,\theta)n^{-1}+\mathcal{O}\big(n^{-\frac{5}{3}}\big)\right\}P(z),
	\end{equation}
	uniformly in $0<r_1\leq|z+\sqrt{2}+s_{n\theta}|\leq r_2<\frac{\epsilon}{2}$ and $\theta\in[0,1]$ for any fixed $r_1,r_2$. The coefficient $N_1(z)=N_1(z;s,\alpha,n,\theta)$ in \eqref{a:31} is bounded in $n$ and equals
	\begin{equation*}
		N_1(z)=\frac{7}{48}E_{\ast}^a(z)\begin{bmatrix}0 & 0\\ 1 & 0\end{bmatrix}E_{\ast}^a(z)^{-1}\big(\zeta^a(z)\big)^{-1}+\frac{5}{48}E_{\ast}^a(z)\begin{bmatrix}0 & 1\\ 0 & 0\end{bmatrix}E_{\ast}^a(z)^{-1}\big(\zeta^a(z)\big)^{-2}.
	\end{equation*}
	Here, $E_{\ast}^a(z):=E^a(z)n^{-\frac{1}{6}\sigma_3}$.
\end{enumerate}
\end{problem}
The jump behavior in condition $(2)$ above follows from RHP \ref{Air}, condition $(2)$, the fact that $z\mapsto E^a(z)$ is analytic at $z=-\sqrt{2}-s_{n\theta}$ and the definition of $z\mapsto\zeta^a(z)$ in \eqref{a:27}, compare Proposition \ref{conf2}.
\begin{rem}\label{speccon} Utilizing \eqref{a:26}, the solution of RHP \ref{YattRHP} for $\alpha=0$ and $\beta=1$ equals
\begin{equation}\label{a:32}
	Q(\zeta;x,0,1)=\sqrt{2\pi}\,\e^{-\im\frac{\pi}{4}}\begin{bmatrix}1&0\,\smallskip\\ -\frac{1}{4}x^2&1\,\end{bmatrix}\begin{bmatrix}\textnormal{Ai}(\zeta+x) & \e^{\im\frac{\pi}{3}}\textnormal{Ai}\big(\e^{-\im\frac{2\pi}{3}}(\zeta+x)\big)\smallskip\\ \textnormal{Ai}'(\zeta+x) & \e^{-\im\frac{\pi}{3}}\textnormal{Ai}'\big(\e^{-\im\frac{2\pi}{3}}(\zeta+x)\big)\end{bmatrix}\begin{cases}I,&\zeta\in\Omega_1\\ 
	\bigl[\begin{smallmatrix}1 & 0\\ -1 & 1\end{smallmatrix}\bigr],&\zeta\in\Omega_2\\ 
	\bigl[\begin{smallmatrix}1 & -1\\ 0 & 1\end{smallmatrix}\bigr],&\zeta\in\Omega_4\\
	\bigl[\begin{smallmatrix}1 & -1\\ 0 & 1\end{smallmatrix}\bigr]\bigl[\begin{smallmatrix}1 & 0\\ 1 & 1\end{smallmatrix}\bigr],&\zeta\in\Omega_3
	\end{cases},
\end{equation}
where $x\in\mathbb{R}$ is arbitrary. Consequently, compare Lemma \ref{ArnoLax}, with $(\alpha,\beta)=(0,1)$ throughout,
\begin{equation*}
	a=\frac{x^2}{4},\ \ \ b=\frac{x}{2}\bigg(1-\frac{x^3}{8}\bigg),\ \ \ Q_1^{21}=\frac{1}{12}\bigg(x^3-\frac{x^6}{16}-\frac{7}{4}\bigg),\ \ \ Q_2^{12}=\frac{1}{48}\bigg(5-5x^3+\frac{x^6}{8}\bigg),
\end{equation*}
\begin{equation*}
	Q_2^{11}=\frac{1}{192}\bigg(35-\frac{7}{5}x^5+\frac{1}{32}x^8\bigg).
\end{equation*}
\end{rem}
The outcomes of RHP \ref{outerArno}, \ref{localb} and \ref{locala} conclude our parametrix analysis. We now compare the underlying model functions locally to $S(z)$.

\subsection{Third transformation: ratio problem} Suppose $s\in\mathbb{R},\alpha>-1,\beta\notin(-\infty,0),\theta\in[0,1]$ and $n\geq n_0$ so $\sqrt{2}-s_{n\theta}\in\mathbb{D}_{\epsilon}(\sqrt{2})$ with $0<\epsilon<\sqrt{2}$ small. Using the invertible functions \eqref{a:19},\eqref{a:21} and \eqref{a:28}, define
\begin{equation}\label{a:33}
	R(z;s,\alpha,\beta,n,\theta):=S(z;s,\alpha,\beta,n,\theta)\begin{cases}M(z)^{-1},&z\in\mathbb{D}_{\epsilon}(\sqrt{2})\\ N(z)^{-1},&z\in\mathbb{D}_{\epsilon}(-\sqrt{2}-s_{n\theta})\\ P(z)^{-1},&z\in\mathbb{C}\setminus(\Sigma_S\cup\overline{\mathbb{D}_{\epsilon}(-\sqrt{2}-s_{n\theta})}\cup\overline{\mathbb{D}_{\epsilon}(\sqrt{2})})
	\end{cases},
\end{equation}
and obtain the following properties of $R(z)$.
\begin{problem}[Small norm problem]\label{ratioArno} Let $s\in\mathbb{R},\alpha>-1,\beta\notin(-\infty,0),\theta\in[0,1]$ and $n\geq n_0$ so $\sqrt{2}-s_n\in\mathbb{D}_{\epsilon}(\sqrt{2})$ with $0<\epsilon<\sqrt{2}$ small. The function $R(z)=R(z;s,\alpha,\beta,n,\theta)\in\mathbb{C}^{2\times 2}$ defined in \eqref{a:33} has the following properties:
\begin{enumerate}
	\item[(1)] $z\mapsto R(z)$ is analytic for $z\in\mathbb{C}\setminus\Sigma_R$ with $\Sigma_R$ shown in Figure \ref{fig3}. On $\Sigma_R$, $R(z)$ admits continuous boundary values $R_{\pm}(z)$ as one approaches $\Sigma_R$ non-tangentially from either side of $\mathbb{C}\setminus\Sigma_R$.
	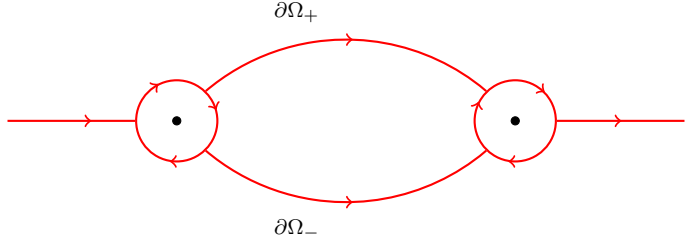
\begin{figure}[tbh]
	\begin{tikzpicture}[xscale=0.9,yscale=0.9]
	\draw [thick, color=red,decoration={markings, mark=at position 0.4 with {\arrow{>}}}, postaction={decorate}] (-5,0) -- (-1.9,0);
	\draw [thick, color=red,decoration={markings, mark=at position 0.7 with {\arrow{>}}}, postaction={decorate}] (1.9,0) -- (5,0);
		\draw[thick,color=red, decoration={markings, mark=at position 0.5 with {\arrow{<}}}, postaction={decorate} ] (2.5,0) arc (38.65980825:141.3401918:3.201562119);
		\draw[thick,color=red, decoration={markings, mark=at position 0.5 with {\arrow{<}}}, postaction={decorate} ] (2.5,0) arc (-38.65980825:-141.3401918:3.201562119);
	\draw [fill=white, thick, white] (2.5,0) circle [radius=0.6];
	\draw [thick, color=red, decoration={markings, mark=at position 0.15 with {\arrow{<}}}, decoration={markings, mark=at position 0.45 with {\arrow{<}}}, decoration={markings, mark=at position 0.75 with {\arrow{<}}},postaction={decorate}] (2.5,0) circle [radius=0.6];
	\draw [fill, color=black] (2.5,0) circle [radius=0.06];
	\draw [fill=white, thick, white] (-2.5,0) circle [radius=0.6];
	\draw [thick, color=red, decoration={markings, mark=at position 0.07 with {\arrow{<}}}, decoration={markings, mark=at position 0.35 with {\arrow{<}}}, decoration={markings, mark=at position 0.75 with {\arrow{<}}},postaction={decorate}] (-2.5,0) circle [radius=0.6];
	\draw [fill, color=black] (-2.5,0) circle [radius=0.06];
	\node [right] at (-1.2,1.6) {{\footnotesize $\partial\Omega_+$}};
	\node [right] at (-1.2,-1.6) {{\footnotesize $\partial\Omega_-$}};
	%\node [below] at (2.7,0) {{\footnotesize $\sqrt{2}$}};
	%\node [below] at (-2.7,0) {{\footnotesize $a$}};
\end{tikzpicture}
\caption{The oriented jump contour $\Sigma_R$, shown in red, for $R(z)$ in the complex $z$-plane.}
\label{fig3}
\end{figure}
	\item[(2)] The limiting values on $\Sigma_R$ satisfy $R_+(z)=R_-(z)G_R(z)$ with $G_R(z)=G_R(z;s,\alpha,\beta,n,\theta)$ given by
	\begin{equation*}
		G_R(z)=I+(\sqrt{2}-z)^{-\alpha}\e^{n\xi(z)}P(z)\begin{bmatrix}0 & 0\\ 1 & 0\end{bmatrix}P(z)^{-1},\ \ \ z\in\partial\Omega_+\cup\partial\Omega_-,
	\end{equation*}
	and
	\begin{equation*}
		G_R(z)=I+|z-\sqrt{2}|^{\alpha}\e^{n\eta(z)}P(z)\begin{bmatrix}0 & 1\\ 0 & 0\end{bmatrix}P(z)^{-1},\ \ \ z\in\mathbb{R}\setminus[-\sqrt{2}-s_{n\theta}-\epsilon,\sqrt{2}+\epsilon],
	\end{equation*}
	on the parts of $\Sigma_R$ disjoint from the circles $\partial\mathbb{D}_{\epsilon}(-\sqrt{2}-s_{n\theta})\cup\partial\mathbb{D}_{\epsilon}(\sqrt{2})$. On those
	\begin{equation*}
		G_R(z)=N(z)P(z)^{-1},\ \ z\in\partial\mathbb{D}_{\epsilon}(-\sqrt{2}-s_{n\theta});\hspace{1cm}G_R(z)=M(z)P(z)^{-1},\ \ z\in\partial\mathbb{D}_{\epsilon}(\sqrt{2}).
	\end{equation*}
	\item[(3)] As $z\rightarrow\infty$, $R(z)$ satisfies the normalization $R(z)=I+\mathcal{O}(z^{-1})$.
\end{enumerate}
\end{problem}
We emphasize that $z\mapsto R(z)$ is analytic in $\mathbb{D}_{\epsilon}(-\sqrt{2}-s_{n\theta})\cup\mathbb{D}_{\epsilon}(\sqrt{2})$ and on $(-\sqrt{2}-s_{n\theta}+\epsilon,\sqrt{2}-\epsilon\subset\mathbb{R}$, because of the analyticity of the maps
\begin{equation*}
	\mathbb{D}_{\epsilon}(-\sqrt{2}-s_{n\theta})\ni z\mapsto S(z)N(z)^{-1},\hspace{2cm}\mathbb{D}_{\epsilon}(\sqrt{2})\ni z\mapsto S(z)M(z)^{-1},
\end{equation*}
and because of RHP \ref{outerArno}. What's more, the jump matrix in RHP \ref{ratioArno} is close to the identity matrix for large $n$, uniformly in $\theta\in[0,1]$, compare \eqref{a:25},\eqref{a:31} and Proposition \ref{snorm1}, in particular since $P(z)=P(z;s,\alpha,n,\theta)$ is bounded in $n$. This is made precise in the following result.
\begin{prop}\label{Arnonorm2} Let $s\in\mathbb{R},\alpha>-1,\beta\in\mathbb{C}\setminus(-\infty,0)$ and $0<\epsilon<\sqrt{2}$ small. There exist $n_0=n_0(s,\alpha,\beta,\epsilon)\in\mathbb{N}$ and $c=c(s,\alpha,\beta,\epsilon)>0$ so that
\begin{equation*}
	\|G_R(\cdot;s,\alpha,\beta,n,\theta)-I\|_{L^2\cap L^{\infty}(\Sigma_R)}\leq\frac{c}{\sqrt[3]{n}}\ \ \ \ \ \forall\,n\geq n_0,\ \ \theta\in[0,1].
\end{equation*}
\end{prop}
Hence, RHP \ref{ratioArno} is a small norm problem and so the general theory of \cite{DZ,DKMVZ} applies to its solvability, for large $n\geq n_0$. We summarize our findings:
\begin{theo}\label{Arnotheo1} For every $s\in\mathbb{R},\alpha>-1,\beta\in\mathbb{C}\setminus(-\infty,0)$ and $0<\epsilon<\sqrt{2}$ small, there exist $n_0=n_0(s,\alpha,\beta,\epsilon)\in\mathbb{N}$ and $c=c(s,\alpha,\beta,\epsilon)>0$ so that RHP \ref{ratioArno} is uniquely solvable for all $n\geq n_0$ and all $\theta\in[0,1]$. Its solution can be computed from the integral representation
\begin{equation}\label{a:34}
	R(z)=I+\frac{1}{2\pi\im}\int_{\Sigma_R}R_-(w)\big(G_R(w)-I)\frac{\d w}{w-z},\ \ \ z\in\mathbb{C}\setminus\Sigma_R,
\end{equation}
by using the small norm estimate
\begin{equation}\label{a:35}
	\|R_-(\cdot;s,\alpha,\beta,n,\theta)-I\|_{L^2(\Sigma_R)}\leq \frac{c}{\sqrt[3]{n}}\ \ \ \ \forall\,n\geq n_0,\ \theta\in[0,1].
\end{equation}
\end{theo}
The last Theorem concludes our nonlinear steepest descent analysis of RHP \ref{FIKbeast} and thus our asymptotic legwork. We now begin to extract the large $n$-asymptotics of $E_n$ in \eqref{e:8}. We begin with the Fredholm determinant $F_n$.
%%%%%%%%%%%%%%%%%%%%%%%%%%%%%%%%%%%%%%%%%%%%%%%%%%%%%%%%%%%%%%

\section{The Fredholm determinant $F_n$ - proof of Theorem \ref{BS:1}}\label{sec4}
Recall from \eqref{e:8} and \eqref{e:9} the formula
\begin{equation}\label{f:1}
	F_n[\phi;\lambda_n,\alpha,\beta]=1+\sum_{\ell=1}^n\frac{(-1)^{\ell}}{\ell!}\int_{J^{\ell}}\det\Big[K_{n,\alpha,\beta}^{\lambda_n}(x_j,x_k)\Big]_{j,k=1}^{\ell}\prod_{m=1}^{\ell}\phi(x_m)\d x_m,
\end{equation}
with the reproducing kernel
\begin{equation*}
	K_{n,\alpha,\beta}^{\lambda_n}(x,y)=\omega_{\alpha\beta}^{\frac{1}{2}}(x-\lambda_n)\e^{-\frac{n}{2}V(x)}\omega_{\alpha\beta}^{\frac{1}{2}}(y-\lambda_n)\e^{-\frac{n}{2}V(y)}\sum_{\ell=0}^{n-1}p_{\ell,n}(x)p_{\ell,n}(y)
\end{equation*}
expressed in terms of \eqref{a:5}, where $\lambda_n\equiv\lambda_{n1}$ and where the test function $\phi:\mathbb{R}\rightarrow\mathbb{C}$ has support $J$.
\begin{proof}[Proof of \eqref{e:14}] We take $\theta=1$ in RHP \ref{FIKbeast} and all subsequent RHPs
\begin{equation}\label{f:2}
	X(z)\stackrel{\eqref{a:12}}{\longrightarrow}Y(z)\stackrel{\eqref{a:14}}{\longrightarrow}T(z)\stackrel{\eqref{a:17}}{\longrightarrow}S(z)\stackrel{\eqref{a:33}}{\longrightarrow}R(z).
\end{equation}
Given the well-known Riemann-Hilbert representation of the correlation kernel $K_{n,\alpha,\beta}^{\lambda_n}$, i.e. the formula
\begin{equation*}
	K_{n,\alpha,\beta}^{\lambda_n}(x,y)=\frac{1}{2\pi\im(x-y)}\omega_{\alpha\beta}^{\frac{1}{2}}(x-\lambda_n)\e^{-\frac{n}{2}V(x)}\omega_{\alpha\beta}^{\frac{1}{2}}(y-\lambda_n)\e^{-\frac{n}{2}V(y)}
	\Big(X_+^{-1}(y)X_+(x)\Big)^{21},\  x,y\in\mathbb{R}\setminus\{\lambda_n\},
\end{equation*}
we now exploit \eqref{f:2} and apply the written out transformations. What results is, for sufficiently large $n\geq n_0$, see also \cite[page $39,40$]{IKO}, here with the shorthand 
\begin{equation*}
	u:=\sqrt{2}+\frac{x-s}{(n\tau)^{\frac{2}{3}}},\ \ \  v:=\sqrt{2}+\frac{y-s}{(n\tau )^{\frac{2}{3}}},
\end{equation*}
where $\tau\equiv\tau_1=\pi h_V(\sqrt{2})2^{\frac{3}{4}}$, see \eqref{a:12} and \eqref{a:3},
\begin{align}
	\frac{1}{(n\tau)^{\frac{2}{3}}}&K_{n,\alpha,\beta}^{\lambda_n}\bigg(\sqrt{2}+\frac{x}{(n\tau)^{\frac{2}{3}}},\sqrt{2}+\frac{y}{(n\tau)^{\frac{2}{3}}}\bigg)=\frac{1}{2\pi\im(x-y)}\Psi\Big(n^{\frac{2}{3}}\big(\zeta^b(v)-\zeta^b(\sqrt{2})\big);n^{\frac{2}{3}}\zeta^b(\sqrt{2}),\alpha,\beta\Big)^{\top}\nonumber\\
	&\times\begin{bmatrix}0&1\\ -1&0\end{bmatrix}E^b(v)^{-1}R(v)^{-1}R(u)E^b(u)\Psi\Big(n^{\frac{2}{3}}\big(\zeta^b(u)-\zeta^b(\sqrt{2})\big);n^{\frac{2}{3}}\zeta^b(\sqrt{2}),\alpha,\beta\Big),\ \ x,y\neq s.\label{f:3}
\end{align}
In \eqref{f:3} we use $\Psi^{\top}$ to denote the matrix transpose of $\Psi$ and the same denotes the $2\times 1$ column vector-valued function
\begin{equation*}
	\Psi(z;x,\alpha,\beta)=\begin{bmatrix}\psi_1(z;x,\alpha,\beta)\\ \psi_2(z;x,\alpha,\beta)\end{bmatrix},\ \ \ z\in\mathbb{R}\setminus\{0\},\ \ x\in\mathbb{R},\ \alpha>-1,\ \beta\notin(-\infty,0)
\end{equation*}
defined in \eqref{e:15}. At this point we pass to the limit $n\rightarrow\infty$ in the right hand side of \eqref{f:3}: by Proposition \ref{conf1} and because of the chosen $\tau$, one finds
\begin{equation*}
	n^{\frac{2}{3}}\zeta^b(\sqrt{2})=s+\mathcal{O}\big(n^{-\frac{2}{3}}\big),\ n^{\frac{2}{3}}\big(\zeta^b(u)-\zeta^b(\sqrt{2})\big)=x-s+\mathcal{O}\big(n^{-\frac{2}{3}}\big),\ n^{\frac{2}{3}}\big(\zeta^b(v)-\zeta^b(\sqrt{2})\big)=y-s+\mathcal{O}\big(n^{-\frac{2}{3}}\big),
\end{equation*}
pointwise in $x,y,s\in\mathbb{R}$, as $n\rightarrow\infty$. Then, arguing exactly as in \cite[$(2.99),(2.100)$]{IKO} we have
\begin{equation}\label{f:4}
	R(v)^{-1}R(u)=I+\mathcal{O}\bigg(\frac{x-y}{n}\bigg),\ \ \ \ E^b(v)^{-1}E^b(u)=I+\mathcal{O}\bigg(\frac{x-y}{\sqrt[3]{n}}\bigg),
\end{equation}
as $n\rightarrow\infty$ for fixed $x,y\neq s$. Inserting \eqref{f:4} into \eqref{f:3} we have thus established the limiting behavior
\begin{align}
	\frac{1}{(n\tau)^{\frac{2}{3}}}&\,K_{n,\alpha,\beta}^{\lambda_n}\bigg(\sqrt{2}+\frac{x}{(n\tau)^{\frac{2}{3}}},\sqrt{2}+\frac{y}{(n\tau)^{\frac{2}{3}}}\bigg)\nonumber\\
	=&\,\frac{\psi_2(x-s;s,\alpha,\beta)\psi_1(y-s;s,\alpha,\beta)-\psi_1(x-s;s,\alpha,\beta)\psi_2(y-s;s,\alpha,\beta)}{2\pi\im(x-y)}+\mathcal{O}\big(n^{-\frac{1}{3}}\big)\label{f:5}
\end{align}
for the reproducing kernel, which holds uniformly in $(x,y,s,\alpha,\beta)$ when chosen on compact sets of their admissible domains and for $x,y\neq s$. Denoting with $A_s^{\alpha\beta}(x-s,y-s)$ the limiting kernel in \eqref{f:5} and choosing $\phi_n$ edge-scaled as in \eqref{e:13}, i.e.
\begin{equation*}
	\phi_n(x)=\phi\Big((x-\sqrt{2})(n\tau)^{\frac{2}{3}}\Big),\ \ \ \ \tau\equiv\tau_1=\pi h_V(\sqrt{2})2^{\frac{3}{4}},
\end{equation*}
the change of variables $y_j=(x_j-\sqrt{2})(n\tau)^{\frac{2}{3}}$ in \eqref{f:1} yields
\begin{equation*}
	F_n[\phi_n;\lambda_n,\alpha,\beta]=1+\sum_{\ell=1}^n\frac{(-1)^{\ell}}{\ell!}\int_{J^{\ell}}\det\Bigg[\frac{1}{(n\tau)^{\frac{2}{3}}}K_{n,\alpha,\beta}^{\lambda_n}\bigg(\sqrt{2}+\frac{y_j}{(n\tau)^{\frac{2}{3}}},\sqrt{2}+\frac{y_k}{(\tau n)^{\frac{2}{3}}}\bigg)\Bigg]_{j,k=1}^{\ell}\prod_{m=1}^{\ell}\phi(y_m)\d y_m.
\end{equation*}
Here, $[K_{n,\alpha,\beta}^{\lambda_n}(x_j,x_k)]_{j,k=1}^{\ell}$ is positive definite for $\beta\geq 0$, and so by Hadamard's inequality, 
\begin{equation}\label{f:6}
	\det\bigg[\frac{1}{(n\tau)^{\frac{2}{3}}}K_{n,\alpha,\beta}^{\lambda_n}(x_j,x_k)\bigg]_{j,k=1}^{\ell}\leq\prod_{j=1}^{\ell}\frac{1}{(n\tau)^{\frac{2}{3}}}K_{n,\alpha,\beta}^{\lambda_n}\bigg(\sqrt{2}+\frac{y_j}{(n\tau)^{\frac{2}{3}}},\sqrt{2}+\frac{y_j}{(\tau n)^{\frac{2}{3}}}\bigg).
\end{equation}
Since, by the proof workings leading to \eqref{f:5},
\begin{equation*}
	\lim_{n\rightarrow\infty}\frac{1}{(n\tau)^{\frac{2}{3}}}K_{n,\alpha,\beta}^{\lambda_n}\bigg(\sqrt{2}+\frac{y}{(n\tau)^{\frac{2}{3}}},\sqrt{2}+\frac{y}{(n\tau)^{\frac{2}{3}}}\bigg)
\end{equation*}
exists uniformly in $y$ on any finite interval $(-A,A)\subset\mathbb{R}$, the $\ell$th term of the sum in $F_n$ above, for bounded test functions $\phi$ of bounded support, is bounded by $c^{\ell}/\ell!$ for some $0<c<\infty$. This is sufficient for the dominated convergence of the series in the formula for $F_n[\phi_n;\lambda_n,\alpha,\beta]$, provided $\phi$ is a bounded test function of bounded support. If the support of $\phi$ has no upper limit, while $\phi$ is still bounded, then we have to find an integrable, on $(A,\infty)$ with $A>0$ large, upper bound for the right hand side of \eqref{f:6} and we must check that the integral over $(A,\infty)^{\ell}$ of the bound yields a convergent series. This can be achieved by exploiting l'Hospital's formula,
\begin{equation*}
	K_{n,\alpha,\beta}^{\lambda_n}(y,y)=\frac{1}{2\pi\im}\omega_{\alpha\beta}(y-\lambda_n)\Big(X_+^{-1}(y)X_+'(y)\Big)^{21},\ \ y\in\mathbb{R}\setminus\{\lambda_n\},
\end{equation*}
and tracing back the chain \eqref{f:2}. What results are two formul\ae\, for each factor in \eqref{f:6} in terms of either $P(z)$ (if $(y-s)/(\tau n)^{2/3}>\epsilon$) in \eqref{a:19} or $Q(z)$ in RHP \ref{YattRHP} (if $(y-s)/(\tau n)^{2/3}<\epsilon$). Either way, Theorem \ref{Arnotheo1}, condition $(4)$ in RHP \ref{YattRHP} and \eqref{a:16} guarantee existence of $c_j=c_j(s,\alpha,\beta)>0$ so that 
\begin{equation*}
	\frac{1}{(n\tau)^{\frac{2}{3}}}K_n\bigg(\sqrt{2}+\frac{y}{(n\tau)^{\frac{2}{3}}},\sqrt{2}+\frac{y}{(n\tau)^{\frac{2}{3}}}\bigg)\leq c_1\exp\left[-c_2y^{\frac{3}{2}}\right]
\end{equation*}
holds for $y>A$, with $A>0$ large. Consequently the integral of the product \eqref{f:6} over $(A,\infty)^{\ell}$ is bounded above by $c^{\ell}$ with a positive $c$, and this is sufficient for the dominated convergence of the series in the formula for $F_n[\phi_n;\lambda_n,\alpha,\beta]$. In turn, \eqref{e:14} follows from \eqref{f:5} for any fixed $s\in\mathbb{R},\alpha>-1$ and $\beta\geq 0$. To have \eqref{e:14} for all $\beta\in\mathbb{C}\setminus(-\infty,0)$, we notice that \eqref{f:3} and $A_s^{\alpha\beta}(x-s,y-s)$ are analytic in $\beta\notin(-\infty,0)$, for fixed $x,y\neq s\in\mathbb{R}$ and $\alpha>-1$, and bounded on all compact sets $K\subset\mathbb{C}\setminus(-\infty,0)$. So \eqref{e:14} holds also for $\beta\in\mathbb{C}\setminus(-\infty,0)$ on compact subsets, by an application of Vitali's convergence theorem and the identity theorem for analytic functions.
\end{proof}
Moving on, we choose $\phi(x)=1,x>s$ and $\phi(x)=0,x<s$ in \eqref{e:14}, so that in the wording of Theorem \ref{BS:1},
\begin{equation}\label{f:7}
	F(s;\alpha,\beta)=1+\sum_{\ell=1}^{\infty}\frac{(-1)^{\ell}}{\ell!}\int_{(0,\infty)^{\ell}}\det\Big[A_s^{\alpha\beta}(x_j,x_k)\Big]_{j,k=1}^{\ell}\prod_{m=1}^{\ell}\d x_m.
\end{equation}
The structure of $A_s^{\alpha}(x,y)$ identifies \eqref{f:7} as Fredholm determinant of an integral operator on $L^2(0,\infty)$ with integrable kernel, cf. \cite{IIKS}. This means we can associate a  canonical RHP with $F(s;\alpha,\beta)$ and \eqref{e:16} will eventually follow. Start by defining the $2\times 1$ column vectors
\begin{equation*}
	f(x)=\begin{bmatrix}f_{1}(x)\\ f_{2}(x)\end{bmatrix}:=\sqrt{\frac{\im}{2\pi}}\begin{bmatrix}\psi_1(x;s,\alpha,\beta)\smallskip\\ \psi_2(x;s,\alpha,\beta)\end{bmatrix},\ \ \ h(x)=\begin{bmatrix}h_{1}(x)\\ h_{2}(x)\end{bmatrix}:=\sqrt{\frac{\im}{2\pi}}\begin{bmatrix}\ \ \psi_2(x;s,\alpha,\beta)\smallskip\\ -\psi_1(x;s,\alpha,\beta)\end{bmatrix},
\end{equation*}
for $x\in\mathbb{R}\setminus\{0\},s\in\mathbb{R},\alpha>-1$ and $\beta\in\mathbb{C}\setminus(-\infty,1)$, with some arbitrary but fixed branch for the complex square root. The additional restriction on $\beta$ will ensure invertibility of $I-A_s^{\alpha\beta}$ on $L^2(0,\infty)$ and thus the non-vanishing of $F(s;\alpha,\beta)$ - see Proposition \ref{solva}. For now we consider the following RHP.
\begin{problem}\label{Arnoint} Let $(s,\alpha,\beta)\in\mathbb{R}\times(-1,\infty)\times(\mathbb{C}\setminus(-\infty,1))$. Find $W(z)=W(z;s,\alpha,\beta)\in\mathbb{C}^{2\times 2}$ with the following properties:
\begin{enumerate}
	\item[(1)] $z\mapsto W(z)$ is analytic for $z\in\mathbb{C}\setminus[0,\infty)$ with continuous boundary values $W_{\pm}(z)=\lim_{\epsilon\downarrow 0}W(z\pm\im\epsilon)$ on $(0,\infty)$.
	\item[(2)] The values $W_{\pm}(z)$ satisfy $W_+(z)=W_-(z)G_W(z),z\in(0,\infty)$ with $G_W(z)=G_W(z;s,\alpha,\beta)$ given by
	\begin{equation*}
		G_W(z)=\begin{bmatrix}1+\psi_1(z)\psi_2(z)& -\psi_1^2(z)\smallskip\\ \psi_2^2(z) & 1-\psi_1(z)\psi_2(z)\end{bmatrix},\ \ \ \ \psi_j(z)=\psi_j(z;s,\alpha,\beta).
	\end{equation*}
	\item[(3)] The limiting values $W_{\pm}(z)$ are in $L^2(0,\infty)$.
	\item[(4)] As $z\rightarrow\infty$ and $z\notin(0,\infty)$, $W(z)$ satisfies the normalization
	\begin{equation*}
		W(z)=I+\frac{W_1}{z}+\mathcal{O}\big(z^{-2}\big),\ \ \ \ \ \ W_1=W_1(s,\alpha,\beta)\in\mathbb{C}^{2\times 2}.
	\end{equation*}
\end{enumerate}
\end{problem}
By the general theory of \cite{IIKS}, RHP \ref{Arnoint} is uniquely solvable in $I+L^2(0,\infty)$, for given $(s,\alpha,\beta)$, if and only if the operator $I-A_s^{\alpha\beta}$ is invertible on $L^2(0,\infty)$. And if solvable, then the unique solution to RHP \ref{Arnoint} is of the form
\begin{equation*}
	W(z)=I-\int_0^{\infty}F(w)h(w)^{\top}\frac{\d w}{w-z},\ \ \ z\in\mathbb{C}\setminus[0,\infty),
\end{equation*}
where, independent of the choice of limiting values, for $z>0$,
\begin{equation*}
	F(z):=\big((I-A_s^{\alpha\beta})^{-1}f\big)(z)=W_{\pm}(z)f(z)\in\mathbb{C}^{2\times 1}.
\end{equation*}
Moreover, another standard fact from \cite{IIKS} says that the resolvent $R_s^{\alpha\beta}=(I-A_s^{\alpha\beta})^{-1}A_s^{\alpha\beta}$, if existent, has kernel
\begin{equation*}
	R_s^{\alpha\beta}(x,y)=\frac{F(x)^{\top}H(y)}{x-y},\ \ \ x,y>0;\hspace{1.5cm} H(y)=\big(W_{\pm}^{\top}(y)\big)^{-1}h(y)\in\mathbb{C}^{2\times 1},\ \ y>0.
\end{equation*}
In order to prove solvability of RHP \ref{Arnoint} and in turn \eqref{e:16} we employ an undressing transformation to the same problem: By \eqref{e:15} and condition $(2)$ in RHP \ref{YattRHP}, the functions
\begin{equation*}
	\psi_1(z;s,\alpha,\beta)=Q_{\pm}^{11}(z;s,\alpha,\beta)\ \ \ \ \textnormal{and}\ \ \ \ \psi_2(z;s,\alpha,\beta)=Q_{\pm}^{21}(z;s,\alpha,\beta)
\end{equation*}
are independent of the choice of limiting values $\pm$ for $z>0$. Hence, choosing $Q_-^{j1}$ in the definition of $\psi_j$ we can factor $G_W(z)$ in RHP \ref{Arnoint} with the help of the solution $Q(\zeta)=Q(\zeta;s,\alpha,\beta)$ to RHP \ref{YattRHP} as follows,
\begin{equation*}
	G_W(z)=Q_-(z)\begin{bmatrix}1&-1\\ 0&1\end{bmatrix}Q_-(z)^{-1},\ \ \ z>0.
\end{equation*}
Consequently, the undressing transformation
\begin{equation}\label{f:8}
	U(z;s,\alpha,\beta):=W(z;s,\alpha,\beta)Q(z;s,\alpha,\beta),\ \ \ \ z\in\mathbb{C}\setminus\Sigma_Q,\ \ \ s\in\mathbb{R},\ \alpha>-1,\ \beta\in\mathbb{C}\setminus(-\infty,1),
\end{equation}
with contour $\Sigma_Q$ shown in Figure \ref{fig2}, leads us to our next RHP.
\begin{problem}[Undressed problem]\label{Arnoundress} Let $s\in\mathbb{R},\alpha>-1,\beta\in\mathbb{C}\setminus(-\infty,1)$ and $U(z)=U(z;s,\alpha,\beta)\in\mathbb{C}^{2\times 2}$ as in \eqref{f:8}. Then $U(z)$ has the following properties:
\begin{enumerate}
	\item[(1)] $z\mapsto U(z)$ is analytic for $z\in\mathbb{C}\setminus\Sigma_Q$ where $\Sigma_Q=\bigcup_{j=1}^4\Sigma_j$ is shown and described in RHP \ref{YattRHP}, compare Figure \ref{fig2}.	On $\Sigma_Q\setminus\{0\}$, $U(z)$ has continuous non-tangential boundary values $U_{\pm}(z)$.
	\item[(2)] The limiting values $U_{\pm}(z)$ on $\Sigma_Q\setminus\{0\}$ satisfy $U_+(z)=U_-(z)G_U(z)$ with $G_U(z)=G_U(z;\alpha,\beta)$ equal
	\begin{equation*}
		G_U(z)=\begin{bmatrix}1&0\\ \e^{\im\pi\alpha}&1\end{bmatrix},\ z\in\Sigma_2;\ \ \ G_U(z)=\begin{bmatrix}0 & 1\\ -1 & 0\end{bmatrix},\ z\in\Sigma_3;\
	\end{equation*}
	\begin{equation}\label{f:9}
		G_U(z)=\begin{bmatrix}1 & 0\\ \e^{-\im\pi\alpha} & 1\end{bmatrix},\ z\in\Sigma_4,;\ \ \ \ G_U(z)=\begin{bmatrix}1&\beta-1\\ 0&1\end{bmatrix},\ z\in\Sigma_1.
	\end{equation}
	\item[(3)] Near $z=0$, $U(z)$ admits the local representation
	\begin{equation*}
		U(z)=\widehat{U}(z)S(z)\mathcal{M}_j\Big|_{\beta\rightarrow\beta-1},\ \ \ \ z\in\big(\Omega_j\cap\mathbb{D}_{\epsilon}(0)\big)\setminus\{0\},
	\end{equation*}
	in terms of some $z\mapsto\widehat{U}(z)=\widehat{U}(z;s,\alpha,\beta)$ that is analytic and non-vanishing at $z=0$ and with $S(z),\mathcal{M}_j$ as in condition $(3)$ of RHP \ref{YattRHP}.
	\item[(4)] As $z\rightarrow\infty$ and $z\notin\Sigma_Q$, $U(z)$ is normalized as 
	\begin{equation*}
		U(z)=\Big\{I+\sum_{k=1}^2U_k(s,\alpha,\beta)z^{-k}+\mathcal{O}\big(z^{-3}\big)\Big\}z^{-\frac{1}{4}\sigma_3}\frac{1}{\sqrt{2}}\begin{bmatrix}1 & 1\\ -1 & 1\end{bmatrix}\e^{-\im\frac{\pi}{4}\sigma_3}\e^{-\varpi(z,s)\sigma_3},
	\end{equation*}
	with $\varpi(z,s)=\frac{2}{3}z^{\frac{3}{2}}+sz^{\frac{1}{2}}$, using principal branches throughout.
\end{enumerate}
\end{problem}
Evidently, compare RHP \ref{YattRHP} and \ref{Arnoundress},
\begin{equation}\label{f:10}
	z\mapsto U(z;s,\alpha,\beta)Q(z;s,\alpha,\beta-1)^{-1}
\end{equation}
is an entire function normalized to unity at $z=\infty$, provided $s\in\mathbb{R},\alpha>-1$ and $\beta\in\mathbb{C}\setminus(-\infty,1)$. Thus, RHP \ref{Arnoundress} is uniquely solvable and so, see \eqref{f:8}, the same goes for RHP \ref{Arnoint}. In turn, by \cite{IIKS}:
\begin{prop}\label{solva} Let $s\in\mathbb{R},\alpha>-1$ and $\beta\in(-\infty,1)$. Then $I-A_s^{\alpha\beta}$ is invertible on $L^2(0,\infty)$.
\end{prop}
\begin{proof}[Proof of \eqref{e:16}] Let $s\in\mathbb{R},\alpha>-1$ and $\beta\notin(-\infty,1)$ be arbitrary but fixed. Given the invertibility of $I-A_s^{\alpha\beta}$, $\mathbb{R}\ni s\mapsto \ln F(s;\alpha,\beta)$ is well-defined with some fixed branch, and using Jacobi's variational formula
\begin{equation}\label{f:11}
	\frac{\d}{\d s}\ln F(s;\alpha,\beta)=-\tr_{L^2(0,\infty)}\bigg((I-A_s^{\alpha\beta})^{-1}\frac{\d A_s^{\alpha\beta}}{\d s}\bigg).
\end{equation}
Since $\psi_j=Q_{\pm}^{j1}$ for $j\in\{1,2\}$, the Lax equation $\frac{\partial Q}{\partial s}=BQ$ in Lemma \ref{ArnoLax}, with $x\equiv s$, yields 
\begin{equation*}
	\frac{\partial\psi_1}{\partial s}=a\psi_1+\psi_2,\ \ \ \ \ \frac{\partial\psi_2}{\partial s}=(z+b)\psi_1-a\psi_2;\ \ \ \ \psi_j=\psi_j(z;s,\alpha,\beta),\ \ a=a(s,\alpha,\beta),\ \ b=b(s,\alpha,\beta),
\end{equation*}
and so,
\begin{equation}\label{f:12}
	\frac{\d }{\d s}A_s^{\alpha\beta}(x,y)=\frac{1}{2\pi\im}\psi_1(x;s,\alpha,\beta)\psi_1(y;s,\alpha,\beta)=-f_1(x)f_1(y).
\end{equation}
This means, exploiting the general setup of the integrable integral operators $A_s^{\alpha\beta}$,
\begin{align}
	\frac{\d}{\d s}\ln F(s;\alpha,\beta)\stackrel{\eqref{f:11}}{=}-\int_0^{\infty}\int_0^{\infty}&(I-A_s^{\alpha\beta})^{-1}(x,y)\frac{\d A_s^{\alpha\beta}}{\d s}(y,x)\d y\,\d x\stackrel{\eqref{f:12}}{=}\int_0^{\infty}\big((I-K_s^{\alpha})^{-1}f_1\big)(x)f_1(x)\,\d x\nonumber\\
	&=\int_0^{\infty}F_1(x)f_1(x)\,\d x=-\int_0^{\infty}F_1(x)h_2(x)\,\d x=-W_1^{12}(s,\alpha,\beta),\label{f:13}
\end{align}
and by \eqref{f:8}, compare also condition $(4)$ in RHPs \ref{Arnoundress}, \ref{YattRHP},
\begin{equation*}
	W_1(s,\alpha,\beta)=U_1(s,\alpha,\beta)-Q_1(s,\alpha,\beta)\stackrel{\eqref{f:10}}{=}Q_1(s,\alpha,\beta-1)-Q_1(s,\alpha,\beta).
\end{equation*}
Consequently, by \eqref{f:13} and Lemma \ref{ArnoLax},
\begin{equation}\label{f:14}
	\frac{\d}{\d s}\ln F(s;\alpha,\beta)=a(s,\alpha,\beta)-a(s,\alpha,\beta-1)=\sigma(s,\alpha,\beta-1)-\sigma(s,\alpha,\beta),
\end{equation}
where $s\mapsto\sigma(s,\alpha,\beta)$ solves the sigma-Painlev\'e-II equation \eqref{B:3} with large $s$-asymptotics written in \eqref{e:17}. Integrating \eqref{f:14} with respect to $s$, using $F(\infty;\alpha,\beta)=1$, results in \eqref{e:16}. 
\end{proof}

%%%%%%%%%%%%%%%%%%%%%%%%%%%%%%%%%%%%%%%%%%%%%%%%%%%%%%%%%%%%%%%%%%%%%%%%%%%%
\section{The ratio of Hankel determinants - proof of Theorem \ref{BS:2}}\label{sec5}
Recall the formula
\begin{equation*}
	D_n\big(\lambda,\alpha,\beta;V(x)\big)=\frac{1}{n!}\int_{\mathbb{R}^n}\prod_{1\leq j<k\leq n}|x_k-x_j|^2\prod_{\ell=1}^n\omega_{\alpha\beta}(x_{\ell}-\lambda)\e^{-nV(x_{\ell})}\d x_{\ell},
\end{equation*}
and the differential identities \eqref{a:9} and \eqref{a:10} for it. The right hand sides in those identities are expressible in terms of Riemann-Hilbert data via \eqref{a:6}, \eqref{a:7} and via $p_{n,n}(z)=\kappa_{n,n}X^{11}(z),p_{n-1,n}(z)=\frac{\im}{2\pi}\kappa_{n-1,n}^{-1}X^{21}(z)$. Thus, seeing that RHP \ref{FIKbeast} is solvable for large $n\geq n_0$ in the admissible parameter domain $s\in\mathbb{R},\alpha>-1,\beta\notin(-\infty,0)$ by Theorem \ref{Arnotheo1} and by the invertible chain of transformations \eqref{f:2}, we now choose $\theta=0$ and derive asymptotics for the right hand sides in \eqref{a:9}, \eqref{a:10}. The upcoming calculations are straightforward but tedious given the explicit appearance of $n$ in \eqref{a:9},\eqref{a:10}, while \eqref{a:35} decays only of order $1/\sqrt[3]{n}$.
\begin{lem} Write
\begin{equation*}
	R_k:=\frac{\im}{2\pi}\int_{\Sigma_R}R_-(w)\big(G_R(w)-I\big)w^{k-1}\d w,\ \ \ k\in\mathbb{N},
\end{equation*}
for the coefficients in the asymptotic series of $R(z)=R(z;s,\alpha,\beta,n,\theta=0)$ at $z=\infty$, compare \eqref{a:34}. Then
\begin{equation}\label{g:1}
	\kappa_{n-1,n}^2=\frac{\im}{2\pi}\e^{n\ell_{V_0}}\big(R_1^{21}+P_1^{21}\big),\ \ \ \ \ \ \kappa_{n,n}^{-2}=-2\pi\im\,\e^{-n\ell_{V_0}}\big(R_1^{12}+P_1^{12}\big),\ \ \ \ \ \ \delta_{n,n}=R_1^{11}+P_1^{11},
\end{equation}
\begin{equation}\label{g:2}
	\gamma_{n,n}=s_{n0}\big(R_1^{11}+P_1^{11}\big)+R_2^{11}+P_2^{11}+\big(R_1P_1\big)^{11}-\frac{n}{4},
\end{equation}
with $P_k=P_k(s,\alpha,n,\theta=0)$ written out in condition $(3)$ of RHP \ref{outerArno} and $\ell_{V_0}=1+\ln 2$, see \eqref{a:2}.
\end{lem}
 \begin{proof} When $\theta=0$, then
 \begin{align*}
 	g_0(z)=&\,\ln z-\frac{1}{4z^2}+\mathcal{O}\big(z^{-3}\big),\\
	R(z-s_{n0})P(z-s_{n0})=&\,I+\frac{1}{z}\big(R_1+P_1\big)+\frac{1}{z^2}\big(R_1P_1+R_2+P_2+s_{n0}(R_1+P_1)\big)+\mathcal{O}\big(z^{-3}\big),
\end{align*}
as $z\rightarrow\infty, z\notin\mathbb{R}$. The identities in \eqref{g:1},\eqref{g:2} now follow from the chain of transformations \eqref{f:2}.
 \end{proof}
\subsection{Asymptotics for $R_k$} Abbreviate $\Sigma_a:=\partial\mathbb{D}_{\epsilon}(-\sqrt{2}-s_{n0})$ and $\Sigma_b:=\partial\mathbb{D}_{\epsilon}(\sqrt{2})$, both oriented clockwise as shown in Figure \ref{fig3}. By Proposition \ref{snorm1} and estimates \eqref{a:25},\eqref{a:31},\eqref{a:35},
\begin{equation*}
	R_k=\frac{\im}{2\pi}\left[\oint_{\Sigma_a}+\oint_{\gamma_b}\right]\big(G_R(w)-I\big)w^{k-1}\d w+\frac{\im}{2\pi}\oint_{\Sigma_b}\big(R_-(w)-I\big)\big(G_R(w)-I\big)w^{k-1}\d w+\mathcal{O}\big(n^{-\frac{4}{3}}\big).
\end{equation*}
We thus first set out to compute asymptotics for $R_-(z),z\in\Sigma_b$.
\begin{lem} As $n\rightarrow\infty$, for any $\theta\in[0,1]$,
\begin{align}
	R_-(z)-I=&\left\{\bigg[\res_{w=\sqrt{2}}M_1(w)\bigg]\frac{1}{z-\sqrt{2}}-M_1(z)\right\}\frac{1}{\sqrt[3]{n}}\nonumber\\
	&\hspace{1cm}+\left\{M_1^2(z)-\bigg[\res_{w=\sqrt{2}}M_1(w)\bigg]\frac{M_1(z)}{z-\sqrt{2}}-M_2(z)+\frac{L}{z-\sqrt{2}}\right\}\frac{1}{\sqrt[3]{n^2}}+\mathcal{O}\big(n^{-1}\big),\label{g:3}
\end{align}
uniformly in $z\in\Sigma_b$ and in $(s,\alpha,\beta)\in\mathbb{R}\times(-1,\infty)\times(\mathbb{C}\setminus(-\infty,0))$ chosen on compact sets. Here
\begin{align*}
	L=\sigma\Bigg\{\frac{Q_1^{11}}{\zeta_1^b}\begin{bmatrix}-\alpha&1+\alpha\\ 1-\alpha&\alpha\end{bmatrix}+\frac{s^4}{32}\frac{1}{\tau_{\theta}^{2/3}}\begin{bmatrix}0&1\\ 1&0\end{bmatrix}-\frac{s^2}{4}\frac{Q_1^{12}}{\sqrt{\zeta_1^b}}\frac{1}{\tau_{\theta}^{1/3}}\begin{bmatrix}0&1\\ 1&0\end{bmatrix}+\frac{\alpha}{2\zeta_1^b}\big(Q_1^{12}\big)^2\begin{bmatrix}1&-1\\ 1&-1\end{bmatrix}\Bigg\}\sigma^{-1},
\end{align*}
and $Q_k^{ij}=Q_k^{ij}(n^{\frac{2}{3}}\zeta^b(\sqrt{2}),\alpha,\beta)$ as in RHP \ref{YattRHP}. Also, $\sigma:=\chi^{\sigma_3}\e^{\im\frac{\pi}{4}\sigma_3}$ and $M_j(z)$ as in RHP \ref{localb}.
\end{lem}
\begin{proof} We will iterate \eqref{a:34}. First, by \eqref{a:25},\eqref{a:31} and \eqref{a:35}, for $z\in\Sigma_b$, via residue theorem,
\begin{align}
	R_-(z)-I=&\,\frac{1}{2\pi\im}\oint_{\Sigma_b}\big(G_R(w)-I\big)\frac{\d w}{w-z_-}+\mathcal{O}\big(n^{-\frac{2}{3}}\big)\stackrel{\eqref{a:25}}{=}\frac{1}{2\pi\im}\oint_{\Sigma_b}M_1(w)n^{-\frac{1}{3}}\frac{\d w}{w-z_-}+\mathcal{O}\big(n^{-\frac{2}{3}}\big)\nonumber\\
	=&\,-\bigg\{M_1(z)-\bigg[\res_{w=\sqrt{2}}M_1(w)\bigg]\frac{1}{z-\sqrt{2}}\bigg\}\frac{1}{\sqrt[3]{n}}+\mathcal{O}\big(n^{-\frac{2}{3}}\big).\label{g:4}
\end{align}
In turn, using \eqref{g:4} at once back in \eqref{a:34},
\begin{align*}
	&R_-(z)-I=\frac{1}{2\pi\im}\oint_{\Sigma_b}\big(R_-(w)-I\big)\big(G_R(w)-I\big)\frac{\d w}{w-z_-}+\frac{1}{2\pi\im}\oint_{\Sigma_b}\big(G_R(w)-I\big)\frac{\d w}{w-z_-}+\mathcal{O}\big(n^{-1}\big)\\
	=&\,\frac{\im}{2\pi}\oint_{\Sigma_b}\bigg\{M_1(w)-\bigg[\res_{\xi=\sqrt{2}}M_1(\xi)\bigg]\frac{1}{w-\sqrt{2}}\bigg\}\frac{M_1(w)}{\sqrt[3]{n^2}}\frac{\d w}{w-z_-}-\frac{\im}{2\pi}\oint_{\Sigma_b}\bigg\{\sum_{k=1}^2\frac{M_k(w)}{\sqrt[3]{n^k}}\bigg\}\frac{\d w}{w-z_-}+\mathcal{O}\big(n^{-1}\big)\\
	=&\,\left\{\bigg[\res_{\xi=\sqrt{2}}M_1(\xi)\bigg]\frac{1}{z-\sqrt{2}}-M_1(z)\right\}\frac{1}{\sqrt[3]{n}}+\left\{\big(M_1(z)\big)^2-\bigg[\res_{\xi=\sqrt{2}}M_1(\xi)\bigg]\frac{M_1(z)}{z-\sqrt{2}}-M_2(z)\right\}\frac{1}{\sqrt[3]{n^2}}\\
    &\hspace{2cm}-\frac{M_0^1M_{-1}^1-M_{-1}^2}{z-\sqrt{2}}\frac{1}{\sqrt[3]{n^2}}+\mathcal{O}\big(n^{-1}\big).
\end{align*}
Here, $M_k^j$ denotes the Laurent coefficients in the local expansion of
\begin{equation*}
    M_j(z)=\sum_{k=-j}^{\infty}M_k^j(z-\sqrt{2})^k,\ \ \ \ \ 0<|z-\sqrt{2}|<\epsilon,\ \ \ j\in\{1,2\},
\end{equation*}
and the same can be read off from \eqref{a:23}, from the explicit formul\ae\,in condition $(4)$ of RHP \ref{localb} and from Lemma \ref{ugly}. In fact, and which was already used in the calculation of $R_-(z)-I$ above,
\begin{equation*}  
	M_{-2}^2\stackrel{\eqref{a:23}}{=}\begin{bmatrix}0&0\\0&0\end{bmatrix},
\end{equation*}
followed by $L:=M_{-1}^2-M_0^1M_{-1}^1$ as written in the statement of the Lemma.
\end{proof}
Next we calculate the contour integrals for $R_k$ that do not involve $R_-(w)-I,w\in\Sigma_b$. 
\begin{prop} Let $k\in\mathbb{N}$. As $n\rightarrow\infty$, with $\sigma=\chi^{\sigma_3}\e^{\im\frac{\pi}{4}\sigma_3}$,
\begin{align}
    &\frac{\im}{2\pi}\oint_{\Sigma_a}\big(G_R(w)-I\big)\Big(\frac{w}{\varkappa}\Big)^{k-1}\d w=\sigma\Bigg\{\frac{5}{96\sqrt{2}}\sqrt{1-\frac{\varkappa}{\sqrt{2}}}\textcolor{brown}{\bigg\{}\frac{1}{2(\varkappa-\sqrt{2})}\begin{bmatrix}-(2\alpha^2+1)&-2\alpha^2-4\alpha-1\smallskip\\ 2\alpha^2-4\alpha+1&2\alpha^2+1\end{bmatrix}\nonumber\\
    &+\bigg(\frac{3}{20\sqrt{2}}+\frac{k-1}{\varkappa}\bigg)\begin{bmatrix}-1&-1\\ 1&1\end{bmatrix}\textcolor{brown}{\bigg\}}+\frac{7}{192}\Big(1-\frac{\varkappa}{\sqrt{2}}\Big)^{-\frac{1}{2}}\begin{bmatrix}\alpha^2-1&(\alpha+1)^2\smallskip\\ -(\alpha-1)^2 & 1-\alpha^2\end{bmatrix}\Bigg\}\frac{\sigma^{-1}}{n}+\mathcal{O}\big(n^{-\frac{5}{3}}\big),\label{g:5}
\end{align}
uniformly in $(s,\alpha,\beta)\in\mathbb{R}\times(-1,\infty)\in(\mathbb{C}\setminus(-\infty,0))$ chosen on compact sets. Here, $\varkappa:=-\sqrt{2}-s_{n0}$.
\end{prop}
\begin{proof} By \eqref{a:31} we have the Laurent expansion $N_1(z)=\sum_{k=-2}^{\infty}N_k^1(z-\varkappa)^k,0<|z-\varkappa|<\epsilon$, and thus via residue theorem,
\begin{equation*}
    \frac{\im}{2\pi}\oint_{\Sigma_a}\big(G_R(w)-I\big)\Big(\frac{w}{\varkappa}\Big)^{k-1}\d w=\Bigg\{\bigg(\frac{k-1}{\varkappa}\bigg)N_{-2}^1+N_{-1}^1\Bigg\}\frac{1}{n}+\mathcal{O}\big(n^{-\frac{5}{3}}\big).
\end{equation*}
It remains to compute
\begin{align*}
    N_{-2}^1\stackrel{\eqref{a:30}}{=}\frac{5}{96\sqrt{2}}\sqrt{1-\frac{\varkappa}{\sqrt{2}}}\,\sigma\begin{bmatrix}-1&-1\\ 1&1\end{bmatrix}\sigma^{-1},
\end{align*}
as well as
\begin{align*}
    N_{-1}^1=&\,\frac{7}{192}\Big(1-\frac{\varkappa}{\sqrt{2}}\Big)^{-\frac{1}{2}}\sigma\begin{bmatrix}\alpha^2-1&(\alpha+1)^2\\ -(\alpha-1)^2&1-\alpha^2\end{bmatrix}\sigma^{-1}\\
    &+\frac{5}{96\sqrt{2}}\sqrt{1-\frac{\varkappa}{\sqrt{2}}}\,\sigma\Bigg\{\frac{1}{2(\varkappa-\sqrt{2})}\begin{bmatrix}-(2\alpha^2+1)&-2\alpha^2-4\alpha-1\\ 2\alpha^2-4\alpha+1&2\alpha^2+1\end{bmatrix}+\frac{3}{20\sqrt{2}}\begin{bmatrix}-1&-1\\ 1&1\end{bmatrix}\Bigg\}\sigma^{-1}.
\end{align*}
The derivation of \eqref{g:5} is complete.
\end{proof}
\begin{prop} Let $k\in\mathbb{N}$. As $n\rightarrow\infty$, with $Q_k^{ij}=Q_k^{ij}(n^{\frac{2}{3}}\zeta^b(\sqrt{2}),\alpha,\beta)$ as in RHP \ref{YattRHP},
\begin{align}
    &\frac{\im}{2\pi}\oint_{\Sigma_b}\big(G_R(w)-I\big)\Big(\frac{w}{\sqrt{2}}\Big)^{k-1}\d w=\sigma\begin{bmatrix}-1&1\\ -1&1\end{bmatrix}\Bigg\{\frac{1}{2}\sqrt{\frac{\sqrt{2}-\varkappa}{\zeta_1^b}}\,Q_1^{12}-\frac{s^2}{8}\sqrt{1-\frac{\varkappa}{\sqrt{2}}}\Bigg\}\frac{\sigma^{-1}}{\sqrt[3]{n}}\nonumber\\
    &+\sigma\Bigg\{\frac{Q_1^{11}}{\zeta_1^b}\begin{bmatrix}-\alpha&1+\alpha\\ 1-\alpha&\alpha\end{bmatrix}+\frac{s^4}{32\sqrt{2}}\begin{bmatrix}1&0\\0&1\end{bmatrix}-\frac{s^2}{8}\frac{Q_1^{12}}{\sqrt{\zeta_1^b\sqrt{2}}}\begin{bmatrix}1-\alpha&1+\alpha\\ 1-\alpha&1+\alpha\end{bmatrix}\Bigg\}\frac{\sigma^{-1}}{\sqrt[3]{n^2}}\nonumber\\
    &+\sigma\Bigg\{\bigg(\frac{k-1}{\sqrt{2}}\bigg)\begin{bmatrix}-1&1\\-1&1\end{bmatrix}\textcolor{brown}{\bigg\{}\frac{1}{2\zeta_1^b}\sqrt{\frac{\sqrt{2}-\varkappa}{\zeta_1^b}}\,Q_2^{12}+\frac{s^3}{48\sqrt{2}}\sqrt{1-\frac{\varkappa}{\sqrt{2}}}\bigg(1-\frac{s^3}{16}\bigg)+\frac{s^4}{64\sqrt{2}}\sqrt{\frac{\sqrt{2}-\varkappa}{\zeta_1^b}}\,Q_1^{12}\nonumber\\
    &-\frac{s^2}{8\zeta_1^b}\sqrt{1-\frac{\varkappa}{\sqrt{2}}}\,Q_1^{11}\textcolor{brown}{\bigg\}}+\frac{1}{2(\zeta_1^b)^2}\sqrt{\frac{\zeta_1^b}{\sqrt{2}-\varkappa}}\,Q_1^{21}\begin{bmatrix}\alpha^2-1&-(\alpha+1)^2\\ (\alpha-1)^2&1-\alpha^2\end{bmatrix}-\frac{s^3}{96\sqrt{2}}\sqrt{1-\frac{\varkappa}{\sqrt{2}}}\textcolor{orange}{\bigg\{}\bigg(1-\frac{s^3}{16}\bigg)\nonumber\\
    &\times\frac{1}{\sqrt{2}-\varkappa}\begin{bmatrix}3&1\\ -1&-3\end{bmatrix}-\frac{3}{2\sqrt{2}}\bigg(1-\frac{s^3}{80}\bigg)\begin{bmatrix}1&-1\\ 1&-1\end{bmatrix}\textcolor{orange}{\bigg\}}-\frac{s^4}{128\sqrt{2}}\sqrt{\frac{\sqrt{2}-\varkappa}{\zeta_1^b}}\,Q_1^{12}\textcolor{blue}{\bigg\{}
    \bigg(\zeta_2^b+\frac{1}{5\sqrt{2}}\bigg)\begin{bmatrix}-1&1\\ -1&1\end{bmatrix}\nonumber\\
    &+\frac{1}{\sqrt{2}-\varkappa}\begin{bmatrix}2\alpha^2+1&-2\alpha^2-4\alpha-1\\ 2\alpha^2-4\alpha+1&-2\alpha^2-1\end{bmatrix}\textcolor{blue}{\bigg\}}-\frac{Q_2^{12}}{4\zeta_1^b}\sqrt{\frac{\sqrt{2}-\varkappa}{\zeta_1^b}}\textcolor{olive}{\bigg\{}\frac{1}{\sqrt{2}-\varkappa}\begin{bmatrix}2\alpha^2+1&-2\alpha^2-4\alpha-1\\ 2\alpha^2-4\alpha+1&-2\alpha^2-1\end{bmatrix}\nonumber\\
    &+3\zeta_2^b\begin{bmatrix}-1&1\\ -1&1\end{bmatrix}\textcolor{olive}{\bigg\}}+\frac{s^2}{8\zeta_1^b}\sqrt{1-\frac{\varkappa}{\sqrt{2}}}\,Q_1^{11}\textcolor{magenta}{\bigg\{}\frac{1}{\sqrt{2}-\varkappa}\begin{bmatrix}2\alpha^2-1&-2\alpha^2-4\alpha-1\\ 2\alpha^2-4\alpha+1&1-2\alpha^2\end{bmatrix}-\bigg(\frac{1}{2(\sqrt{2}-\varkappa)}-\frac{1}{20\sqrt{2}}\nonumber\\
    &-\zeta_2^b\bigg)\begin{bmatrix}-1&1\\ -1&1\end{bmatrix}\textcolor{magenta}{\bigg\}}\Bigg\}
    \frac{\sigma^{-1}}{n}+\mathcal{O}\big(n^{-\frac{4}{3}}\big),\label{g:6}
\end{align}
uniformly in $(s,\alpha,\beta)\in\mathbb{R}\times(-1,\infty)\times(\mathbb{C}\setminus(-\infty,0))$ on compact sets. Here, $\sigma=\chi^{\sigma_3}\e^{\im\frac{\pi}{4}\sigma_3},\varkappa=-\sqrt{2}-s_{n0}$.
\end{prop}
\begin{proof} By \eqref{a:25}, as $n\rightarrow\infty$
\begin{align}
    \frac{\im}{2\pi}\oint_{\Sigma_b}\big(G_R(w)-I\big)\Big(\frac{w}{\sqrt{2}}\Big)^{k-1}\d w=\frac{M_{-1}^1}{\sqrt[3]{n}}+&\,\Bigg\{\bigg(\frac{k-1}{\sqrt{2}}\bigg)M_{-2}^2+M_{-1}^2\Bigg\}\frac{1}{\sqrt[3]{n^2}}\nonumber\\
    &\hspace{1cm}+\Bigg\{\bigg(\frac{k-1}{\sqrt{2}}\bigg)M_{-2}^3+M_{-1}^3\Bigg\}\frac{1}{n}+\mathcal{O}\big(n^{-\frac{4}{3}}\big),\label{g:7}
\end{align}
with $M_k^j$ as Laurent coefficients in the expansion $M_j(z)=\sum_{k=-j}^{\infty}M_k^j(z-\sqrt{2})^k,0<|z-\sqrt{2}|<\epsilon$ for $j\in\{1,2,3\}$ with $M_{-3}^3=M_{-2}^2=\bigl[\begin{smallmatrix}0&0\\ 0&0\end{smallmatrix}\bigr]$. The non-zero coefficents follow again from \eqref{a:23}, condition $(4)$ in RHP \ref{localb} and Lemma \ref{ugly}. In detail,
\begin{align*}
    M_{-1}^1=&\,\,\sigma\begin{bmatrix}-1&1\\-1&1\end{bmatrix}\sigma^{-1}\Bigg\{\frac{1}{2}\sqrt{\frac{\sqrt{2}-\varkappa}{\zeta_1^b}}\,Q_1^{12}-\frac{s^2}{8}\sqrt{1-\frac{\varkappa}{\sqrt{2}}}\Bigg\};\\
    \bigg(\frac{k-1}{\sqrt{2}}\bigg)M_{-2}^2+M_{-1}^2=&\,\,\sigma\Bigg\{\frac{Q_1^{11}}{\zeta_1^b}\begin{bmatrix}-\alpha&1+\alpha\\ 1-\alpha&\alpha\end{bmatrix}+\frac{s^4}{32\sqrt{2}}\begin{bmatrix}1&0\\0&1\end{bmatrix}-\frac{s^2}{8}\frac{Q_1^{12}}{\sqrt{\zeta_1^b\sqrt{2}}}\begin{bmatrix}1-\alpha&1+\alpha\\ 1-\alpha&1+\alpha\end{bmatrix}\Bigg\}\sigma^{-1},
\end{align*}
and the lengthier
\begin{align*}
    M_{-2}^3&=\sigma\begin{bmatrix}-1&1\\ -1&1\end{bmatrix}\sigma^{-1}\\
    &\hspace{0.5cm}\times\Bigg\{\frac{1}{2\zeta_1^b}\sqrt{\frac{\sqrt{2}-\varkappa}{\zeta_1^b}}\,Q_2^{12}+\frac{s^3}{48\sqrt{2}}\sqrt{1-\frac{\varkappa}{\sqrt{2}}}\,\bigg(1-\frac{s^3}{16}\bigg)+\frac{s^4}{64\sqrt{2}}\sqrt{\frac{\sqrt{2}-\varkappa}{\zeta_1^b}}\,Q_1^{12}-\frac{s^2}{8\zeta_1^b}\sqrt{1-\frac{\varkappa}{\sqrt{2}}}\,Q_1^{11}\Bigg\};\\
	M_{-1}^3&=\sigma\Bigg\{\frac{1}{2(\zeta_1^b)^2}\sqrt{\frac{\zeta_1^b}{\sqrt{2}-\varkappa}}\,Q_1^{21}\begin{bmatrix}\alpha^2-1&-(\alpha+1)^2\\ (\alpha-1)^2&1-\alpha^2\end{bmatrix}-\frac{s^3}{96\sqrt{2}}\sqrt{1-\frac{\varkappa}{\sqrt{2}}}\textcolor{orange}{\bigg\{}\bigg(1-\frac{s^3}{16}\bigg)\frac{1}{\sqrt{2}-\varkappa}\begin{bmatrix}3&1\\ -1&-3\end{bmatrix}\\
    &-\frac{3}{2\sqrt{2}}\bigg(1-\frac{s^3}{80}\bigg)\begin{bmatrix}1&-1\\ 1&-1\end{bmatrix}\textcolor{orange}{\bigg\}}-\frac{s^4}{128\sqrt{2}}\sqrt{\frac{\sqrt{2}-\varkappa}{\zeta_1^b}}\,Q_1^{12}\textcolor{blue}{\bigg\{}\frac{1}{\sqrt{2}-\varkappa}\begin{bmatrix}2\alpha^2+1&-2\alpha^2-4\alpha-1\\ 2\alpha^2-4\alpha+1&-2\alpha^2-1\end{bmatrix}\\
    &+\bigg(\zeta_2^b+\frac{1}{5\sqrt{2}}\bigg)\begin{bmatrix}-1&1\\ -1&1\end{bmatrix}\textcolor{blue}{\bigg\}}-\frac{Q_2^{12}}{4\zeta_1^b}\sqrt{\frac{\sqrt{2}-\varkappa}{\zeta_1^b}}\textcolor{olive}{\bigg\{}\frac{1}{\sqrt{2}-\varkappa}\begin{bmatrix}2\alpha^2+1&-2\alpha^2-4\alpha-1\\ 2\alpha^2-4\alpha+1&-2\alpha^2-1\end{bmatrix}\\
    &+3\zeta_2^b\begin{bmatrix}-1&1\\ -1&1\end{bmatrix}\textcolor{olive}{\bigg\}}+\frac{s^2}{8\zeta_1^b}\sqrt{1-\frac{\varkappa}{\sqrt{2}}}\,Q_1^{11}\textcolor{magenta}{\bigg\{}\frac{1}{\sqrt{2}-\varkappa}\begin{bmatrix}2\alpha^2-1&-2\alpha^2-4\alpha-1\\ 2\alpha^2-4\alpha+1&1-2\alpha^2\end{bmatrix}\\
    &-\bigg(\frac{1}{2(\sqrt{2}-\varkappa)}-\frac{1}{20\sqrt{2}}-\zeta_2^b\bigg)\begin{bmatrix}-1&1\\ -1&1\end{bmatrix}\textcolor{magenta}{\bigg\}}\Bigg\}\sigma^{-1}.
\end{align*}
Combining the above formul\ae\,for $M_{-1}^1,M_{-1}^2,M_{-2}^3$ and $M_{-1}^3$ in \eqref{g:7} results in \eqref{g:6}.
\end{proof}
%Lastly we utilize \eqref{gh0} and calculate the below contour integral.
\begin{prop} Let $k\in\mathbb{N}$. As $n\rightarrow\infty$, with $Q_k^{ij}=Q_k^{ij}(n^{\frac{2}{3}}\zeta^b(\sqrt{2}),\alpha,\beta)$ as in RHP \ref{YattRHP},
\begin{align}
   & \frac{\im}{2\pi}\oint_{\Sigma_b}\big(R_-(w)-I\big)\big(G_R(w)-I\big)\Big(\frac{w}{\sqrt{2}}\Big)^{k-1}\d w=\sigma\Bigg\{\frac{s^4}{32\sqrt{2}}\begin{bmatrix}-1&1\\1&-1\end{bmatrix}+\frac{\alpha}{2\zeta_1^b}\big(Q_1^{12}\big)^2\begin{bmatrix}1&-1\\1&-1\end{bmatrix}\nonumber\\
    &+\frac{s^2}{8\sqrt{\zeta_1^b\sqrt{2}}}Q_1^{12}\begin{bmatrix}1-\alpha&-1+\alpha\\ -1-\alpha&1+\alpha\end{bmatrix}\Bigg\}\frac{\sigma^{-1}}{\sqrt[3]{n^2}}+\sigma\Bigg\{-\frac{s^4}{128\sqrt{2}}\sqrt{\frac{\sqrt{2}-\varkappa}{\zeta_1^b}}\,Q_1^{12}\textcolor{brown}{\bigg\{}\frac{1}{\sqrt{2}-\varkappa}\nonumber\\
    &\hspace{0.5cm}\times\begin{bmatrix}-2\alpha^2-1&2\alpha^2+4\alpha-7\\-2\alpha^2+4\alpha+7&2\alpha^2+1\end{bmatrix}+\bigg(\zeta_2^b+\frac{1}{5\sqrt{2}}\bigg)\begin{bmatrix}1&-1\\ 1&-1\end{bmatrix}\textcolor{brown}{\bigg\}}-\frac{\alpha^2}{2\zeta_1^b}\sqrt{\frac{\sqrt{2}-\varkappa}{\zeta_1^b}}\frac{(Q_1^{12})^3}{\sqrt{2}-\varkappa}\begin{bmatrix}1&-1\\ 1&-1\end{bmatrix}\nonumber\\
    &+\frac{s^2}{16\sqrt{\zeta_1^b\sqrt{2}}}\sqrt{\frac{\sqrt{2}-\varkappa}{\zeta_1^b}}\,\big(Q_1^{12}\big)^2\textcolor{orange}{\bigg\{}\frac{4\alpha}{\sqrt{2}-\varkappa}\begin{bmatrix}0&1\\ 1&0\end{bmatrix}+\bigg(\frac{1}{\sqrt{2}-\varkappa}+\frac{1}{10\sqrt{2}}\bigg)\begin{bmatrix}1&-1\\1&-1\end{bmatrix}\textcolor{orange}{\bigg\}}+\frac{s^6}{256\sqrt{2}}\sqrt{1-\frac{\varkappa}{\sqrt{2}}}\nonumber\\
    &\hspace{0.5cm}\times\textcolor{blue}{\bigg\{}\frac{1}{2}\bigg(\frac{1}{\sqrt{2}-\varkappa}+\frac{1}{10\sqrt{2}}\bigg)\begin{bmatrix}1&-1\\ 1&-1\end{bmatrix}-\frac{1}{\sqrt{2}-\varkappa}\begin{bmatrix}1&1\\ -1&-1\end{bmatrix}\textcolor{blue}{\bigg\}}\nonumber\\
    &+\frac{s^2}{8\zeta_1^b}\sqrt{1-\frac{\varkappa}{\sqrt{2}}}\,Q_1^{11}\textcolor{magenta}{\bigg\{}\frac{1}{\sqrt{2}-\varkappa}\begin{bmatrix}-2\alpha^2-1&2\alpha^2-1\\ -2\alpha^2+1&2\alpha^2+1\end{bmatrix}-\bigg(\zeta_2^b+\frac{1}{2(\sqrt{2}-\varkappa)}-\frac{1}{20\sqrt{2}}\bigg)\begin{bmatrix}-1&1\\ -1&1\end{bmatrix}\textcolor{magenta}{\bigg\}}\nonumber\\
    &-\frac{1}{4\zeta_1^b}\sqrt{\frac{\sqrt{2}-\varkappa}{\zeta_1^b}}\,Q_1^{11}Q_1^{12}\textcolor{olive}{\bigg\{}\frac{1}{\sqrt{2}-\varkappa}\begin{bmatrix}-6\alpha^2+1&6\alpha^2+4\alpha-1\\ -6\alpha^2+4\alpha+1 & 6\alpha^2-1\end{bmatrix}+\zeta_2^b\begin{bmatrix}1&-1\\1&-1\end{bmatrix}\textcolor{olive}{\bigg\}}\Bigg\}\frac{\sigma^{-1}}{n}+\mathcal{O}\big(n^{-\frac{4}{3}}\big),\label{g:8}
\end{align}
uniformly in $(s,\alpha,\beta)\in\mathbb{R}\times(-1,\infty)\times(\mathbb{C}\setminus(-\infty,0))$ on compact sets. Here, $\sigma=\chi^{\sigma_3}\e^{\im\frac{\pi}{4}\sigma_3},\varkappa=-\sqrt{2}-s_{n0}$.
\end{prop}
\begin{proof} By \eqref{g:3} and \eqref{a:25}, as $n\rightarrow\infty$,
\begin{align}
	\frac{\im}{2\pi}\oint_{\Sigma_b}\big(R_-(w)-I\big)\big(G_R(w)-I\big)&\Big(\frac{w}{\sqrt{2}}\Big)^{k-1}\d w=-\frac{M_0^1M_{-1}^1}{\sqrt[3]{n^2}}\label{g:9}\\
	&+\Big\{\big(M_0^1\big)^2M_{-1}^1+M_1^1\big(M_{-1}^1\big)^2-M_0^2M_{-1}^1-M_0^1M_{-1}^2\Big\}\frac{1}{n}+\mathcal{O}\big(n^{-\frac{4}{3}}\big)\nonumber
\end{align}
with $M_k^j$ as Laurent coefficients in the expansion $M_j(z)=\sum_{k=-j}^{\infty}M_k^j(z-\sqrt{2})^k,0<|z-\sqrt{2}|<\epsilon$. Here,
\begin{align*}
	M_0^1M_{-1}^1=\sigma\Bigg\{\frac{s^4}{32\sqrt{2}}\begin{bmatrix}1&-1\\ -1&1\end{bmatrix}-\frac{\alpha}{2\zeta_1^b}\big(Q_1^{12}\big)^2\begin{bmatrix}1&-1\\ 1&-1\end{bmatrix}-\frac{s^2}{8\sqrt{\zeta_1^b\sqrt{2}}}\,Q_1^{12}\begin{bmatrix}1-\alpha&-1+\alpha\\ -1-\alpha&1+\alpha\end{bmatrix}\Bigg\}\sigma^{-1},
\end{align*}
followed by
\begin{align*}
	&\big(M_0^1\big)^2M_{-1}^1=\sigma\Bigg\{\frac{s^4}{64\sqrt{2}}\sqrt{\frac{\sqrt{2}-\varkappa}{\zeta_1^b}}\,Q_1^{12}\bigg(\zeta_2^b+\frac{1}{10\sqrt{2}}-\frac{2(1+\alpha^2)}{\sqrt{2}-\varkappa}\bigg)-\frac{s^2}{16\sqrt{\zeta_1^b\sqrt{2}}}\sqrt{\frac{\sqrt{2}-\varkappa}{\zeta_1^b}}\big(Q_1^{12}\big)^2\\
	&\times\bigg(\zeta_2^b-\frac{1+4\alpha^2}{\sqrt{2}-\varkappa}\bigg)-\frac{\alpha^2}{2\zeta_1^b}\sqrt{\frac{\sqrt{2}-\varkappa}{\zeta_1^b}}\frac{(Q_1^{12})^3}{\sqrt{2}-\varkappa}+\frac{s^6}{256\sqrt{2}}\sqrt{1-\frac{\varkappa}{\sqrt{2}}}\bigg(\frac{1}{\sqrt{2}-\varkappa}-\frac{1}{10\sqrt{2}}\bigg)\Bigg\}\begin{bmatrix}1&-1\\1&-1\end{bmatrix}\sigma^{-1},
\end{align*}
and $M_1^1\big(M_{-1}^1\big)^2=\bigl[\begin{smallmatrix}0&0\\ 0&0\end{smallmatrix}\bigr]$. Next,
\begin{align*}
	M_0^2&\,M_{-1}^1=\sigma\Bigg\{\frac{Q_1^{11}}{\zeta_1^b}\textcolor{brown}{\bigg\{}\frac{2\alpha}{\sqrt{2}-\varkappa}\begin{bmatrix}-1-\alpha&1+\alpha\\ 1-\alpha&\alpha-1\end{bmatrix}-\zeta_2^b\begin{bmatrix}-1&1\\-1&1\end{bmatrix}\textcolor{brown}{\bigg\}}-\frac{s^4}{640}\begin{bmatrix}-1&1\\-1&1\end{bmatrix}-\frac{s^2}{8\sqrt{\zeta_1^b\sqrt{2}}}\,Q_1^{12}\\
	&\times\textcolor{orange}{\bigg\{}\bigg(\frac{1}{\sqrt{2}-\varkappa}+\frac{1}{10\sqrt{2}}+\zeta_2^b\bigg)\begin{bmatrix}1&-1\\ 1&-1\end{bmatrix}+\frac{1}{\sqrt{2}-\varkappa}\begin{bmatrix}-2\alpha^2-2\alpha-1&2\alpha^2+2\alpha+1\\ -2\alpha^2+2\alpha-1&2\alpha^2-2\alpha+1\end{bmatrix}\textcolor{orange}{\bigg\}}\Bigg\}\sigma^{-1}\\
	&\hspace{1cm}\times\Bigg\{\frac{1}{2}\sqrt{\frac{\sqrt{2}-\varkappa}{\zeta_1^b}}\,Q_1^{12}-\frac{s^2}{8}\sqrt{1-\frac{\varkappa}{\sqrt{2}}}\Bigg\},
\end{align*}
and finally
\begin{align*}
	&M_0^1M_{-1}^2=\sigma\Bigg\{\frac{1}{4\zeta_1^b}\sqrt{\frac{\sqrt{2}-\varkappa}{\zeta_1^b}}\,Q_1^{11}Q_1^{12}\textcolor{brown}{\bigg\{}\frac{1}{\sqrt{2}-\varkappa}\begin{bmatrix}-2\alpha^2+4\alpha+1&2\alpha^2-1\\ -2\alpha^2+1&2\alpha^2+4\alpha-1\end{bmatrix}+\zeta_2^b\begin{bmatrix}-1&1\\-1&1\end{bmatrix}\textcolor{brown}{\bigg\}}\\
	&+\frac{s^4}{128\sqrt{2}}\sqrt{\frac{\sqrt{2}-\varkappa}{\zeta_1^b}}\,Q_1^{12}\textcolor{orange}{\bigg\{}\frac{1}{\sqrt{2}-\varkappa}\begin{bmatrix}-2\alpha^2+4\alpha-5&2\alpha^2-3\\-2\alpha^2+3&2\alpha^2+4\alpha+5\end{bmatrix}+\zeta_2^b\begin{bmatrix}1&-1\\ 1&-1\end{bmatrix}\textcolor{orange}{\bigg\}}\\
	&-\frac{s^2}{8\sqrt{\zeta_1^b\sqrt{2}}}\sqrt{\frac{\sqrt{2}-\varkappa}{\zeta_1^b}}\big(Q_1^{12}\big)^2\frac{\alpha}{\sqrt{2}-\varkappa}\begin{bmatrix}1-\alpha&1+\alpha\\ 1-\alpha&1+\alpha\end{bmatrix}+\frac{s^2}{8\zeta_1^b}\sqrt{1-\frac{\varkappa}{\sqrt{2}}}\,Q_1^{11}\textcolor{blue}{\bigg\{}\frac{1}{\sqrt{2}-\varkappa}\begin{bmatrix}1-2\alpha&1+2\alpha\\ -1+2\alpha&-1-2\alpha\end{bmatrix}\\
	&+\bigg(\frac{1}{2(\sqrt{2}-\varkappa)}-\frac{1}{20\sqrt{2}}\bigg)\begin{bmatrix}-1&1\\-1&1\end{bmatrix}\textcolor{blue}{\bigg\}}+\frac{s^6}{256\sqrt{2}}\sqrt{1-\frac{\varkappa}{\sqrt{2}}}\textcolor{magenta}{\bigg\{}\frac{1}{\sqrt{2}-\varkappa}\begin{bmatrix}1&1\\-1&-1\end{bmatrix}+\bigg(\frac{1}{2(\sqrt{2}-\varkappa)}-\frac{1}{20\sqrt{2}}\bigg)\\
	&\hspace{0.5cm}\times\begin{bmatrix}1&-1\\1&-1\end{bmatrix}\textcolor{magenta}{\bigg\}}\Bigg\}\sigma^{-1}.
\end{align*}
Combining the above identities in \eqref{g:9} results in \eqref{g:8}. Our proof of the Proposition is complete.
\end{proof}
Equipped with \eqref{g:5},\eqref{g:6},\eqref{g:8} the asymptotics of $R_k$ are now within reach and thus the asymptotics for $\kappa_{n-1,n}^2,\kappa_{n,n}^{-2},\delta_{n,n}$ and $\gamma_{n,n}$ in \eqref{g:1},\eqref{g:2}. To obtain them fully, we recall $\varkappa=-\sqrt{2}-s_{n0}$ and ($\theta=0$ at present)
\begin{equation*}
	\sqrt{2}-\varkappa=2\sqrt{2}\bigg(1+\frac{s}{4n^{\frac{2}{3}}}\bigg),\ \ \ \ \zeta_1^b=\sqrt{2}\bigg(1+\frac{s}{10n^{\frac{2}{3}}}+\mathcal{O}\big(n^{-\frac{4}{3}}\big)\bigg),\ \ \ \ n^{\frac{2}{3}}\zeta^b(\sqrt{2})=s\bigg(1+\frac{s}{20n^{\frac{2}{3}}}+\mathcal{O}\big(n^{-\frac{4}{3}}\big)\bigg),
\end{equation*}
from Proposition \ref{conf1} and the expansions, compare Lemma \ref{ArnoLax}, using $\tr Q_1=0$,
\begin{align*}
	Q_1^{12}\big(n^{\frac{2}{3}}\zeta^b(\sqrt{2}),\alpha,\beta\big)&\,=\,a(s,\alpha,\beta)+\bigg[\frac{\d a}{\d s}(s,\alpha,\beta)\bigg]\frac{s^2}{20n^{\frac{2}{3}}}+\mathcal{O}\big(n^{-\frac{4}{3}}\big),\\
	-2Q_1^{11}\big(n^{\frac{2}{3}}\zeta^b(\sqrt{2}),\alpha,\beta\big)&\,=\,Q_1^{22}\big(n^{\frac{2}{3}}\zeta^b(\sqrt{2}),\alpha,\beta\big)-Q_1^{11}\big(n^{\frac{2}{3}}\zeta^b(\sqrt{2}),\alpha,\beta\big)=b(s,\alpha,\beta)+\mathcal{O}\big(n^{-\frac{2}{3}}\big).
\end{align*}
What results are the following expansions
\begin{cor}\label{dataexp} As $n\rightarrow\infty$, we have at $\theta=0$,
\begin{align*}
	&\kappa_{n-1,n}^2=\frac{\e^{n\ell_{V_0}}}{2\pi}2^{\frac{1}{2}(\alpha-1)}\Bigg\{1-\textcolor{orange}{\bigg[}a-\frac{s^2}{4}\textcolor{orange}{\bigg]}\frac{1}{\sqrt[3]{n}}+\frac{1}{2}\textcolor{blue}{\bigg[}(\alpha-1)b-\frac{as^2}{2}+\alpha a^2+\frac{s^4}{16}+(1-\alpha)\frac{s}{2}\textcolor{blue}{\bigg]}\frac{1}{\sqrt[3]{n^2}}\\
	&+\frac{1}{4}\textcolor{magenta}{\Bigg[}\frac{\alpha}{2}-\frac{\alpha^2}{4}-\frac{1}{6}+(\alpha-1)^2Q_1^{21}+\frac{s^6}{96}+\frac{5s^3}{24}-\frac{s^4}{8}a-Q_2^{12}\bigg(\alpha^2-2\alpha+\frac{1}{5}\bigg)+\frac{s^2}{2}b\bigg(\alpha-\frac{7}{10}\bigg)\\
	&-\alpha^2a^3+\frac{s^2}{2}a^2\bigg(\alpha+\frac{3}{10}\bigg)-\frac{1}{2}ab\bigg(3\alpha^2-2\alpha-\frac{3}{5}\bigg)+\alpha s\bigg(a-\frac{s^2}{4}\bigg)-\frac{s^2}{5}\frac{\d a}{\d s}-\frac{3s}{10}a\textcolor{magenta}{\Bigg]}\frac{1}{n}+\mathcal{O}\big(n^{-\frac{4}{3}}\big)\Bigg\},
\end{align*}
and
\begin{align*}
	&\kappa_{n,n}^{-2}=2\pi\e^{-n\ell_{V_0}}2^{-\frac{1}{2}(\alpha+1)}\Bigg\{1+\textcolor{orange}{\bigg[}a-\frac{s^2}{4}\textcolor{orange}{\bigg]}\frac{1}{\sqrt[3]{n}}+\frac{1}{2}\textcolor{blue}{\bigg[}-(\alpha+1)b-\frac{as^2}{2}-\alpha a^2+\frac{s^4}{16}+(1+\alpha)\frac{s}{2}\textcolor{blue}{\bigg]}\frac{1}{\sqrt[3]{n^2}}\\
	&+\frac{1}{4}\textcolor{magenta}{\bigg[}\frac{\alpha}{2}+\frac{\alpha^2}{4}+\frac{1}{6}-(\alpha+1)^2Q_1^{21}-\frac{s^6}{96}-\frac{5s^3}{24}+\frac{s^4}{8}\hat{a}+Q_2^{12}\bigg(\alpha^2+2\alpha+\frac{1}{5}\bigg)+\frac{s^2}{2}b\bigg(\alpha+\frac{7}{10}\bigg)\\
	&+\alpha^2a^3+\frac{s^2}{2}a^2\bigg(\alpha-\frac{3}{10}\bigg)+\frac{1}{2}ab\bigg(3\alpha^2+2\alpha-\frac{3}{5}\bigg)+\alpha s\bigg(a-\frac{s^2}{4}\bigg)+\frac{s^2}{5}\frac{\d a}{\d s}+\frac{3s}{10}a\textcolor{magenta}{\bigg]}\frac{1}{n}+\mathcal{O}\big(n^{-\frac{4}{3}}\big)\Bigg\}.
\end{align*}
Moreover,
\begin{align*}
	\delta_{n,n}=\frac{1}{\sqrt{2}}\Bigg\{&\alpha-\textcolor{orange}{\bigg[}a-\frac{s^2}{4}\textcolor{orange}{\bigg]}\frac{1}{\sqrt[3]{n}}+\frac{\alpha}{2}\textcolor{blue}{\bigg[}b+a^2+\frac{s}{2}\textcolor{blue}{\bigg]}\frac{1}{\sqrt[3]{n^2}}+\frac{1}{4}\textcolor{magenta}{\bigg[}\frac{\alpha^2}{4}-\frac{1}{8}+(\alpha^2-1)Q_1^{21}-\bigg(\alpha^2+\frac{1}{5}\bigg)Q_2^{12}\\
	&+\frac{3s^2}{20}\big(b+a^2\big)-\alpha^2a^3-\frac{3}{2}ab\bigg(\alpha^2-\frac{1}{5}\bigg)-\frac{s^2}{5}\frac{\d a}{\d s}-\frac{3s}{10}a+\frac{s^3}{8}\textcolor{magenta}{\bigg]}\frac{1}{n}+\mathcal{O}\big(n^{-\frac{4}{3}}\big)\Bigg\},
\end{align*}
and
\begin{align*}
	&\gamma_{n,n}=-\frac{n}{4}+\frac{1}{4}(1+\alpha+\alpha^2)-\frac{1}{2}(1+\alpha)\textcolor{orange}{\bigg[}a-\frac{s^2}{4}\textcolor{orange}{\bigg]}\frac{1}{\sqrt[3]{n}}+\frac{1}{2}\textcolor{blue}{\bigg[}\frac{s}{4}(1+3\alpha+\alpha^2)+\frac{\alpha}{2}(1+\alpha)\big(b+a^2\big)-\frac{b}{2}\\
	&-\frac{s^2}{4}a+\frac{s^4}{32}\textcolor{blue}{\bigg]}\frac{1}{\sqrt[3]{n^2}}+\frac{1}{8}\textcolor{magenta}{\bigg[}\bigg(\frac{9\alpha}{5}-\alpha^3-\alpha^2-\frac{21}{5}\bigg)Q_2^{12}+(\alpha^3+\alpha^2-3\alpha-3\big)Q_1^{21}+\frac{5}{8}s^3\\
	&+\frac{3}{8}\alpha-\frac{1}{4}\alpha^2+\frac{\alpha^3}{4}+\frac{3}{8}\alpha s^3+ab\bigg(\frac{3}{10}-\frac{3\alpha^2}{2}+\frac{13}{10}\alpha-\frac{3}{2}\alpha^3\bigg)+s^2\bigg(\frac{13}{20}\alpha b+\frac{13}{20}\alpha a^2-\frac{1}{5}(1+\alpha)\frac{\d a}{\d s}\\
	&+\frac{3}{20}b+\frac{3}{20}a^2\bigg)-\frac{13}{10}\alpha as-\frac{33}{10}sa-\alpha^2(1+\alpha)a^3\textcolor{magenta}{\bigg]}\frac{1}{n}+\mathcal{O}\big(n^{-\frac{4}{3}}\big)
\end{align*}
Here, $a=a(s,\alpha,\beta),b=b(s,\alpha,\beta)$ as in Lemma \ref{ArnoLax} and $Q_1^{21}=Q_1^{21}(s,\alpha,\beta),Q_2^{12}=Q_2^{12}(s,\alpha,\beta)$ as in RHP \ref{YattRHP}. All expansions are uniform in $(s,\alpha,\beta)\in\mathbb{R}\times(-1,\infty)\times(\mathbb{C}\setminus(-\infty,0))$ on compact sets.
\end{cor}
We are now well-prepared to proof Theorem \ref{BS:2}. First we establish \eqref{e:23} within our calculations, although the same is known from \cite[Theorem $2$]{BCI}.
\begin{proof}[Proof of \eqref{e:23}] Take $\alpha=0$ in Corollary \ref{dataexp}, use \eqref{a:9} and Lemma \ref{ArnoLax},
\begin{equation*}
	\frac{\partial}{\partial\lambda_{n0}}\ln D_n\big(\lambda_{n0},0,\beta;x^2\big)=\sqrt{2}n^{\frac{2}{3}}\sigma_0(s,\beta)+\mathcal{O}(1).
\end{equation*}
Consequently, integrating from $\lambda_n=\sqrt{2}+s/(\sqrt{2}n^{\frac{2}{3}})$ to $t_n=\sqrt{2}+t/(\sqrt{2}n^{\frac{2}{3}})$ with $s<t$, we obtain
\begin{equation}\label{g:10}
	\ln\bigg[\frac{D_n(\lambda_n,0,\beta;x^2)}{D_n(t_n,0,\beta;x^2)}\bigg]=-\int_{s}^{t}\sigma_0(x,\beta)\d x+\mathcal{O}\big(n^{-\frac{2}{3}}\big),\ \ \ n\rightarrow\infty,
\end{equation}
and utilizing Remark \ref{unif} we let $t\rightarrow+\infty$ in both sides of \eqref{g:10}. Evidently $D_n(t_n,0,\beta;x^2)\rightarrow D_n(\lambda_n,0,1;x^2)$ and the right hand side in \eqref{g:10} becomes $\int_{s}^{\infty}\sigma_0(x,\beta)\d x$ by \eqref{e:17}. Thus \eqref{e:23} follows.
\end{proof}
Second, we move to \eqref{e:21} where we split our calculations in two parts, as driven by \eqref{a:10}. First we focus on 
\begin{equation*}
	D_{1n}:=\frac{n}{2}\frac{\partial}{\partial\alpha}\ln\kappa_{n-1,n}^2+\frac{1}{2}(n+\alpha)\frac{\partial}{\partial\alpha}\ln\kappa_{n,n}^{-2}-n\bigg(\frac{\kappa_{n-1,n}}{\kappa_{n,n}}\bigg)^2\frac{\partial}{\partial\alpha}\ln\bigg(\frac{\kappa_{n-1,n}^2}{\kappa_{n,n}^2}\bigg)+2n\frac{\partial}{\partial\alpha}\bigg(\gamma_{n,n}-\frac{1}{2}\delta_{n,n}^2\bigg).
\end{equation*}
A straightforward, albeit tedious, calculation based on Corollary \ref{dataexp} leads to
\begin{cor} As $n\rightarrow\infty$, 
\begin{equation}\label{g:11}
	D_{1n}=\frac{n}{2}(1-\ln 2)+s\sqrt[3]{n}+\frac{\partial}{\partial\alpha}\bigg[\frac{1}{3}a^3-sa+\frac{1}{2}ab-Q_2^{12}-Q_1^{21}\bigg]
	-\frac{\alpha}{4}\ln 2+\mathcal{O}\big(n^{-\frac{1}{3}}\big),
\end{equation}
uniformly in $(s,\alpha,\beta)\in\mathbb{R}\times(-1,\infty)\times(\mathbb{C}\setminus(-\infty,0))$ on compact sets. Here, $a=a(s,\alpha,\beta),b=b(s,\alpha,\beta)$ as in Lemma \ref{ArnoLax} and $Q_2^{12}=Q_2^{12}(s,\alpha,\beta),Q_1^{21}=Q_1^{21}(s,\alpha,\beta)$ as in RHP \ref{YattRHP}.
\end{cor}
The second contribution in \eqref{a:10}, i.e.
\begin{equation}\label{g:12}
	D_{2n}=\alpha\bigg[\kappa_{n,n}^{-1}\frac{\partial}{\partial\alpha}\big(\kappa_{n,n}X^{11}(\lambda_{n0})\big)\mathring{X}^{22}(\lambda_{n0})-\kappa_{n-1,n}\frac{\partial}{\partial\alpha}\big(\kappa_{n-1,n}^{-1}X^{21}(\lambda_{n0})\big)\mathring{X}^{12}(\lambda_{n0})\bigg]
\end{equation}
requires more work as we still need to calculate the leading order $n$-asymptotics of $X^{11}(\lambda_{n0}),X^{21}(\lambda_{n0})$ and of $\mathring{X}^{22}(\lambda_{n0}),\mathring{X}^{12}(\lambda_{n0})$.
\begin{prop}\label{weird} As $n\rightarrow\infty$, 
\begin{align*}
	X^{11}(\lambda_{n0})=&\,\,\e^{\frac{n}{2}V_0(\sqrt{2}+s_{n0})-\frac{n}{2}\ell_{V_0}}n^{\frac{\alpha}{3}}(\zeta_1^b)^{\frac{\alpha}{2}}\frac{\chi}{\sqrt{2}}\e^{\im\frac{\pi}{4}}\big(n^{\frac{2}{3}}\zeta_1^b(2\sqrt{2}+s_{n0})\big)^{\frac{1}{4}}H_n^{11}\e^{\im\frac{\pi}{2}\alpha},\\
	X^{21}(\lambda_{n0})=&\,\,\e^{\frac{n}{2}V_0(\sqrt{2}+s_{n0})+\frac{n}{2}\ell_{V_0}}n^{\frac{\alpha}{3}}(\zeta_1^b)^{\frac{\alpha}{2}}\frac{\chi^{-1}}{\sqrt{2}}\e^{-\im\frac{\pi}{4}}\big(n^{\frac{2}{3}}\zeta_1^b(2\sqrt{2}+s_{n0})\big)^{\frac{1}{4}}H_n^{21}\e^{\im\frac{\pi}{2}\alpha},
\end{align*}
followed by
\begin{align*}
	\mathring{X}^{22}(\lambda_{n0})=&\,\,\e^{-\frac{n}{2}V_0(\sqrt{2}+s_{n0})+\frac{n}{2}\ell_{V_0}}n^{-\frac{\alpha}{3}}(\zeta_1^b)^{-\frac{\alpha}{2}}\frac{\chi^{-1}}{\sqrt{2}}\e^{-\im\frac{\pi}{4}}\big(n^{\frac{2}{3}}\zeta_1^b(2\sqrt{2}+s_{n0})\big)^{\frac{1}{4}}H_n^{22}\e^{-\im\frac{\pi}{2}\alpha}\\
	\mathring{X}^{12}(\lambda_{n0})=&\,\,\e^{-\frac{n}{2}V_0(\sqrt{2}+s_{n0})-\frac{n}{2}\ell_{V_0}}n^{-\frac{\alpha}{3}}(\zeta_1^b)^{-\frac{\alpha}{2}}\frac{\chi}{\sqrt{2}}\e^{\im\frac{\pi}{4}}\big(n^{\frac{2}{3}}\zeta_1^b(2\sqrt{2}+s_{n0})\big)^{\frac{1}{4}}H_n^{12}\e^{-\im\frac{\pi}{2}\alpha},
\end{align*}
uniformly in $(s,\alpha,\beta)\in\mathbb{R}\times((-1,\infty)\setminus\mathbb{Z})\times(\mathbb{C}\setminus(-\infty,0))$ on compact sets. Here, $H_n^{jk}=H_n^{jk}(s,\alpha,\beta)$ have the large $n$-asymptotics
\begin{align*}
	H_n^{11}=&\,\,\widehat{Q}^{11}(0)-\frac{1}{2}\bigg\{(1+\alpha)\widehat{Q}^{21}(0)+\bigg(\frac{s^2}{4}+\alpha a\bigg)\widehat{Q}^{11}(0)\bigg\}\frac{1}{\sqrt[3]{n}}+\mathcal{O}\big(n^{-\frac{2}{3}}\big),\\
	H_n^{22}=&\,\,\widehat{Q}^{12}(0)+\frac{1}{2}\bigg\{(1-\alpha)\widehat{Q}^{22}(0)+\bigg(\frac{s^2}{4}-\alpha a\bigg)\widehat{Q}^{12}(0)\bigg\}\frac{1}{\sqrt[3]{n}}+\mathcal{O}\big(n^{-\frac{2}{3}}\big),\\
	H_n^{21}=&\,\,\widehat{Q}^{11}(0)+\frac{1}{2}\bigg\{(1-\alpha)\widehat{Q}^{21}(0)+\bigg(\frac{s^2}{4}-\alpha a\bigg)\widehat{Q}^{11}(0)\bigg\}\frac{1}{\sqrt[3]{n}}+\mathcal{O}\big(n^{-\frac{2}{3}}\big),\\
	H_n^{12}=&\,\,\widehat{Q}^{12}(0)-\frac{1}{2}\bigg\{(1+\alpha)\widehat{Q}^{22}(0)+\bigg(\frac{s^2}{4}+\alpha a\bigg)\widehat{Q}^{12}(0)\bigg\}\frac{1}{\sqrt[3]{n}}+\mathcal{O}\big(n^{-\frac{2}{3}}\big),
\end{align*}
with $a=a(s,\alpha,\beta)$ as in Lemma \ref{ArnoLax} and $\widehat{Q}^{jk}(0)=\widehat{Q}^{jk}(0;s,\alpha,\beta)$ as in RHP \ref{YattRHP}.
\end{prop}
\begin{proof} Utilizing \eqref{f:2} with $\theta=0$ throughout, we start with the exact identities
\begin{align*}
	X^{11}(\lambda_{n0})=\e^{\frac{n}{2}V_0(\sqrt{2}+s_{n0})-\frac{n}{2}\ell_{V_0}}n^{\frac{\alpha}{3}}(\zeta_1^b)^{\frac{\alpha}{2}}\big(R(\sqrt{2})E^b(\sqrt{2})\widehat{Q}(0)\big)^{11}\e^{\im\frac{\pi}{2}\alpha}
\end{align*}
\begin{align*}
	X^{21}(\lambda_{n0})=&\,\,\e^{\frac{n}{2}V_0(\sqrt{2}+s_{n0})+\frac{n}{2}\ell_{V_0}}n^{\frac{\alpha}{3}}(\zeta_1^b)^{\frac{\alpha}{2}}\big(R(\sqrt{2})E^b(\sqrt{2})\widehat{Q}(0)\big)^{21}\e^{\im\frac{\pi}{2}\alpha},\\
	\mathring{X}^{22}(\lambda_{n0})=&\,\,\e^{-\frac{n}{2}V_0(\sqrt{2}+s_{n0})+\frac{n}{2}\ell_{V_0}}n^{-\frac{\alpha}{3}}(\zeta_1^b)^{-\frac{\alpha}{2}}\big(R(\sqrt{2})E^b(\sqrt{2})\widehat{Q}(0)\big)^{22}\e^{-\im\frac{\pi}{2}\alpha},\\
	\mathring{X}^{12}(\lambda_{n0})=&\,\,\e^{-\frac{n}{2}V_0(\sqrt{2}+s_{n0})-\frac{n}{2}\ell_{V_0}}n^{-\frac{\alpha}{3}}(\zeta_1^b)^{-\frac{\alpha}{2}}\big(R(\sqrt{2})E^b(\sqrt{2})\widehat{W}(0)\big)^{12}\e^{-\im\frac{\pi}{2}\alpha}.
\end{align*}
After that we are in need of
\begin{equation*}
	R(\sqrt{2})\stackrel{\eqref{a:34}}{=}I+\frac{1}{2\pi\im}\int_{\Sigma_R}R_-(w)\big(G_R(w)-I\big)\frac{\d w}{w-\sqrt{2}}\stackrel{\eqref{a:35}}{=}I+\frac{1}{2\pi\im}\oint_{\Sigma_b}\big(G_R(w)-I\big)\frac{\d w}{w-\sqrt{2}}+\mathcal{O}\big(n^{-\frac{2}{3}}\big),
\end{equation*}
with $\Sigma_b=\partial\mathbb{D}_{\epsilon}(\sqrt{2})$ oriented clockwise. By residue theorem, compare \eqref{a:25},
\begin{align*}
	\frac{1}{2\pi\im}&\oint_{\Sigma_b}\big(G_R(w)-I\big)\frac{\d w}{w-\sqrt{2}}=\frac{1}{2\pi\im}\oint_{\Sigma_b}M_1(w)n^{-\frac{1}{3}}\frac{\d w}{w-\sqrt{2}}+\mathcal{O}\big(n^{-\frac{2}{3}}\big)=-\frac{M_0^1}{\sqrt[3]{n}}+\mathcal{O}\big(n^{-\frac{2}{3}}\big),
\end{align*}
with $M_k^1$ as Laurent coefficient in the expansion $M_1(z)=\sum_{k=-1}^{\infty}M_k^1(z-\sqrt{2})^k, 0<|z-\sqrt{2}|<\epsilon$, and so
\begin{align*}
	R(\sqrt{2})&E^b(\sqrt{2})\widehat{Q}(0)=\chi^{\sigma_3}\e^{\im\frac{\pi}{4}\sigma_3}\frac{1}{\sqrt{2}}\big(n^{\frac{2}{3}}\zeta_1^b(2\sqrt{2}+s_{n0})\big)^{\frac{1}{4}}\Bigg\{\begin{bmatrix}\widehat{Q}^{11}(0) & \widehat{Q}^{12}(0)\smallskip\\ \widehat{Q}^{11}(0) & \widehat{Q}^{12}(0)\end{bmatrix}+\big(n^{\frac{2}{3}}\zeta_1^b(2\sqrt{2}+s_{n0})\big)^{-\frac{1}{2}}\\
	&\hspace{1.5cm}\times\begin{bmatrix}-(1+\alpha)\widehat{Q}^{21}(0)&-(1+\alpha)\widehat{Q}^{22}(0)\smallskip\\ (1-\alpha)\widehat{Q}^{21}(0) & (1-\alpha)\widehat{Q}^{22}(0)\end{bmatrix}-\frac{Q_1^{12}}{\sqrt{\zeta_1^b(2\sqrt{2}+s_{n0})}}\frac{\alpha}{\sqrt[3]{n}}\begin{bmatrix}\widehat{Q}^{11}(0)&\widehat{Q}^{12}(0)\smallskip\\ \widehat{Q}^{11}(0) &\widehat{Q}^{12}(0)\end{bmatrix}\\
	&\hspace{2.5cm}-\frac{s^2}{4\sqrt{\sqrt{2}(2\sqrt{2}+s_{n0})}}\frac{1}{\sqrt[3]{n}}\begin{bmatrix}\widehat{Q}^{11}(0) & \widehat{Q}^{12}(0)\smallskip\\ -\widehat{Q}^{11}(0) & -\widehat{Q}^{12}(0)\end{bmatrix}+\mathcal{O}\big(n^{-\frac{2}{3}}\big)\Bigg\}.%\ \ \ \ \ \widehat{Q}^{jk}(0)=\widehat{Q}^{jk}\big(0;n^{\frac{2}{3}},\zeta^b(b),\alpha\big),
\end{align*}
The four listed expansions for $X^{11}(\lambda_{n0}),X^{21}(\lambda_{n0}),\mathring{X}^{22}(\lambda_{n0})$ and $\mathring{X}^{12}(\lambda_{n0})$ follow at once.
\end{proof}
\begin{cor} Recall \eqref{g:12}. Then, as $n\rightarrow\infty$,
\begin{align}\label{g:13}
	D_{2n}=\frac{\alpha}{3}\ln n+\frac{\alpha}{4}\ln 2+q\frac{\partial a}{\partial\alpha}+\frac{\im\pi}{2}\alpha+\alpha\bigg[\frac{\partial}{\partial\alpha}\widehat{Q}^{11}(0)\bigg]\widehat{Q}^{22}(0)-\alpha\bigg[\frac{\partial}{\partial\alpha}\widehat{Q}^{21}(0)\bigg]\widehat{Q}^{12}(0)+\mathcal{O}\bigg(\frac{\ln n}{\sqrt[3]{n}}\bigg),
\end{align}
uniformly in $s\in\mathbb{R}$ and $\alpha\in(-1,\infty)\setminus\mathbb{Z}$ and $\beta\in\mathbb{C}\setminus(-\infty,0)$, all chosen on compact sets. Here, $\widehat{Q}^{jk}(0)=\widehat{Q}^{jk}(0;s,\alpha,\beta)$ as in RHP \ref{YattRHP} and $q=q(s,\alpha,\beta),a=a(s,\alpha,\beta)$ as in Lemma \ref{ArnoLax}.
\end{cor}
\begin{proof} One applies Proposition \ref{weird} together with $\det\widehat{Q}(0;s,\alpha,\beta)=1$, leading to
\begin{equation*}
	H_n^{11}H_n^{22}-H_n^{21}H_n^{12}=\frac{1}{\sqrt[3]{n}}+\mathcal{O}\big(n^{-\frac{2}{3}}\big),\ \ \ H_n^{11}H_n^{22}+H_n^{21}H_n^{12}=-\frac{2q}{\alpha}+\mathcal{O}\big(n^{-\frac{1}{3}}\big),
\end{equation*}
and
\begin{equation*}
	\frac{\partial H_n^{11}}{\partial\alpha}H_n^{22}-\frac{\partial H_n^{21}}{\partial\alpha}H_n^{12}=\Bigg\{\bigg[\frac{\partial\widehat{Q}^{11}}{\partial\alpha}(0;s,\alpha)\bigg]\widehat{Q}^{22}(0;s,\alpha)-\bigg[\frac{\partial\widehat{Q}^{21}}{\partial\alpha}(0;s,\alpha)\bigg]\widehat{Q}^{12}(0;s,\alpha)\Bigg\}\frac{1}{\sqrt[3]{n}}+\mathcal{O}\big(n^{-\frac{2}{3}}\big).
\end{equation*}
This completes the proof of \eqref{g:13}.
\end{proof}
Combining \eqref{g:11} and \eqref{g:13} we have thus for the right hand side of \eqref{a:10},
\begin{cor} As $n\rightarrow\infty$,
\begin{align}
	\frac{\partial}{\partial\alpha}\ln D_n\big(\lambda_{n0},\alpha,\beta;&\,x^2\big)=\frac{n}{2}(1-\ln 2)+s\sqrt[3]{n}+\frac{\alpha}{3}\ln n+\frac{\partial}{\partial\alpha}\bigg[\frac{1}{3}a^3-sa+\frac{1}{2}ab-Q_2^{12}-Q_1^{21}\bigg]\nonumber\\
	&+q\frac{\partial a}{\partial\alpha}+\frac{\im\pi}{2}\alpha+\alpha\bigg[\frac{\partial}{\partial\alpha}\widehat{Q}^{11}(0)\bigg]\widehat{Q}^{22}(0)-\alpha\bigg[\frac{\partial}{\partial\alpha}\widehat{Q}^{21}(0)\bigg]\widehat{Q}^{12}(0)+\mathcal{O}\bigg(\frac{\ln n}{\sqrt[3]{n}}\bigg),\label{g:14}
\end{align}
uniformly in $(s,\alpha,\beta)\in\mathbb{R}\times((-1,\infty)\setminus\mathbb{Z})\times(\mathbb{C}\setminus(-\infty,0))$.
\end{cor}
Likewise, generalizing the workings that underlie \eqref{g:10},
\begin{cor} As $n\rightarrow\infty$,
\begin{equation}\label{g:15}
	\frac{\partial}{\partial s}\ln D_n\big(\lambda_{n0},\alpha,\beta;x^2\big)=\alpha\sqrt[3]{n}+\sigma_{\alpha}(s,\beta)+\mathcal{O}\big(n^{-\frac{1}{3}}\big),
\end{equation}
uniformly in $(s,\alpha,\beta)\in\mathbb{R}\times(-1,\infty)\times(\mathbb{C}\setminus(-\infty,0))$.
\end{cor}
At this point everything is in place to derive \eqref{e:21}.
\begin{proof}[Proof of \eqref{e:21}] Definite $\alpha$-integration in \eqref{g:14} yields
\begin{equation}\label{g:16}
	\ln\bigg[\frac{D_n(\lambda_{n0},\alpha,\beta;x^2)}{D_n(\lambda_{n0},0,\beta;x^2)}\bigg]=\frac{n\alpha}{2}(1-\ln 2)+\alpha s\sqrt[3]{n}+\frac{\alpha^2}{6}\ln n+\eta_{\alpha}(s,\beta)+\mathcal{O}\bigg(\frac{\ln n}{\sqrt[3]{n}}\bigg)
\end{equation}
in terms of the function
\begin{align}\label{g:17}
	\eta_{\alpha}(s,\beta):=\int_0^{\alpha}\Bigg(&\frac{\partial}{\partial\alpha'}\bigg[\frac{1}{3}a^3-sa+\frac{1}{2}ab-Q_2^{12}-Q_1^{21}\bigg]+q\frac{\partial a}{\partial\alpha'}+\frac{\im\pi}{2}\alpha'\nonumber\\
	&+\alpha'\bigg[\frac{\partial}{\partial\alpha'}\widehat{Q}^{11}(0)\bigg]\widehat{Q}^{22}(0)-\alpha'\bigg[\frac{\partial}{\partial\alpha'}\widehat{Q}^{21}(0)\bigg]\widehat{Q}^{12}(0)\Bigg)\d\alpha'
\end{align}
with $a=a(s,\alpha',\beta),b=b(s,\alpha',\beta)$ as in Lemma \ref{ArnoLax} and $Q_n^{jk}=Q_n^{jk}(s,\alpha',\beta),\widehat{Q}^{jk}(0)=\widehat{Q}^{jk}(0;s,\alpha',\beta)$ as in RHP \ref{YattRHP} in the integrand. Note that the integrand in \eqref{g:16} is integrable at $\alpha'=0$ as can be seen from RHP \ref{YattRHP}, and so $\eta_0(s,\beta)\equiv 0$. Also, the integral \eqref{g:17} can be made more explicit once $\beta=1$, see Remark \ref{speccon}. Moving on, indefinite $s$-integration in \eqref{g:15} yields
\begin{equation}\label{g:18}
	\ln D_n\big(\lambda_{n0},\alpha,\beta;x^2\big)=\alpha s\sqrt[3]{n}+\int^s\sigma_{\alpha}(x,\beta)\d x+c_n(\alpha,\beta)+\mathcal{O}\big(n^{-\frac{1}{3}}\big)
\end{equation}
with $c_n(\alpha,\beta)$ as integration constant. Comparing \eqref{g:16} with \eqref{g:18} and \eqref{e:23} we deduce 
\begin{equation*}
	c_n(\alpha,\beta)=\frac{n\alpha}{2}(1-\ln 2)+\frac{\alpha^2}{6}\ln n+\ln D_n\big(\lambda_{n0},0,1;x^2\big)+\mathcal{O}\bigg(\frac{\ln n}{\sqrt[3]{n}}\bigg),\ \ n\rightarrow\infty,
\end{equation*}
and the $s$-differentiable asymptotic
\begin{equation*}
	\int^s\sigma_{\alpha}(x,\beta)\d x=\eta_{\alpha}(s,\beta)-\int_s^{\infty}\sigma_0(x,\beta)\d x+\mathcal{O}\bigg(\frac{\ln n}{\sqrt[3]{n}}\bigg).
\end{equation*}
Utilizing once more \eqref{g:15}, noting $\lambda_{n0}\equiv\lambda_n$, we arrive at \eqref{e:22} and so \eqref{g:16} completes the proof of \eqref{e:21}.
\end{proof}

\section{The ratio of Hankel determinants - proof of Theorem \ref{BS:3}}\label{sec6}

For arbitrary one-cut regular $V$ we apply \eqref{a:11} and derive large $n$-asymptotics for its right hand side. Fix a bounded contour $\gamma_+$, resp. $\gamma_-$, in the unbounded component of $\mathbb{C}\setminus\Sigma_R$, such that it lies in the domain of analyticity of $V$ in the upper, resp. lower, half-plane, and terminates at $\tau_{\pm}\in\mathbb{R}$ so  $-\infty<\tau_-<\varkappa<\sqrt{2}<\tau_+<\infty$. Here, $\varkappa=-\sqrt{2}-s_{n\theta}$ and $\theta\in[0,1]$ is kept arbitrary throughout. Then
\begin{lem} As $n\rightarrow\infty$, using $E=[-\sqrt{2},\sqrt{2}]\subset\mathbb{R}$,
\begin{align}
	\textnormal{RHS of}\ \eqref{a:11}=-&n^2\int_E\rho_{V_{\theta}}(x)\frac{\partial V_{\theta}}{\partial\theta}(x)\,\d x-\frac{n}{2\pi\im}\bigg[\int_{\gamma_+}-\int_{\gamma_-}\bigg]\frac{\mathcal{D}'(z)}{\mathcal{D}(z)}\frac{\partial V_{\theta}}{\partial\theta}(z+s_{n\theta})\,\d z\nonumber\\
	&+\frac{n}{2\pi\im}\bigg[\int_{\gamma_+}-\int_{\gamma_-}\bigg]\bigg\{\big(R(z)P(z)\big)^{-1}\Big(\frac{\d}{\d z}R(z)\Big)P(z)\bigg\}^{11}\,\frac{\partial V_{\theta}}{\partial\theta}(z+s_{n\theta})\,\d z+\mathcal{O}\big(n^{-\infty}\big),\label{h:1}
\end{align}
uniformly in $\theta\in[0,1]$ and $(s,\alpha,\beta)\in\mathbb{R}\times(-1,\infty)\times(\mathbb{C}\setminus(-\infty,0))$ on compact sets. Here, $s_{n\theta}$ as in \eqref{a:12}.
\end{lem}
\begin{proof} By Cauchy's theorem and RHP \ref{FIKbeast},
\begin{align}
	\textnormal{RHS of}\ \eqref{a:11}=\frac{n}{2\pi\im}\bigg[\int_{\gamma_+}-\int_{\gamma_-}&\bigg]\bigg\{X(z)^{-1}\frac{\d}{\d z}X(z)\bigg\}^{11}\,\frac{\partial V_{\theta}}{\partial\theta}(z)\,\d z\nonumber\\
	&-\frac{n}{2\pi\im}\bigg[\int_{-\infty}^{\tau_-}+\int_{\tau_+}^{\infty}\bigg]\bigg\{X(x)^{-1}\frac{\d}{\d x}X(x)\bigg\}^{21}\,\frac{\partial V_{\theta}}{\partial\theta}(x)\,\d\nu(x),\label{h:2}
\end{align}
where $\tau_-<\varkappa<\sqrt{2}<\tau_+$ are the terminal points of $\gamma_{\pm}$. By the chain \eqref{f:2}, with general $\theta\in[0,1]$ in place, we find at once, for $z\in(-\infty,\tau_-)\cup\gamma_{\pm}\cup(\tau_+,\infty)$,
\begin{align}
	X(z)^{-1}\frac{\d}{\d z}X(z)=n\frac{\d}{\d z}g_{\theta}(z)&\sigma_3+\e^{-n(g_{\theta}(z)+\frac{1}{2}\ell_{V_{\theta}})\sigma_3}P(z-s_{n\theta})^{-1}R(z-s_{n\theta})^{-1}\nonumber\\
	&\hspace{2cm}\times\frac{\d}{\d z}\big(R(z-s_{n\theta})P(z-s_{n\theta})\big)\e^{n(g_{\theta}(z)+\frac{1}{2}\ell_{V_{\theta}})\sigma_3},\label{h:3}
\end{align}
and consequently, by \eqref{a:34},\eqref{a:35},
\begin{align}
	\bigg[\int_{-\infty}^{\tau_-}+&\int_{\tau_+}^{\infty}\bigg]\bigg\{X(x)^{-1}\frac{\d}{\d x}X(x)\bigg\}^{21}\,\frac{\partial V_{\theta}}{\partial\theta}(x)\,\d\nu(x)\label{h:4}\\
	&\,=\bigg[\int_{-\infty}^{\tau_-}+\int_{\tau_+}^{\infty}\bigg]\bigg\{P(x-s_{n\theta})^{-1}\frac{\d}{\d x}P(x-s_{n\theta})+\mathcal{O}\big(n^{-\frac{1}{3}}\big)\bigg\}^{21}\e^{n\eta(x-s_{n\theta})}|x-\lambda_{n\theta}|^{\alpha}\frac{\partial V_{\theta}}{\partial\theta}(x)\,\d x\nonumber
\end{align}
with $\eta(z)$ as in \eqref{a:15}. But $x\mapsto\eta(x-s_{n\theta})$ is strictly negative on $(-\infty,\tau_-)\cup(\tau_+,\infty)$, compare Lemma \ref{glemma}. Hence, \eqref{h:4} is exponentially decaying as $n\rightarrow\infty$ and so
\begin{align*}
	\textnormal{RHS of}\ \eqref{a:11}=\frac{n}{2\pi\im}\bigg[\int_{\gamma_+}-\int_{\gamma_-}\bigg]\bigg\{X(z)^{-1}\frac{\d}{\d z}X(z)\bigg\}^{11}\,\frac{\partial V_{\theta}}{\partial\theta}(z)\,\d z+\mathcal{O}\big(n^{-\infty}\big),
\end{align*}
where by \eqref{h:2} and \eqref{a:19},
\begin{equation*}
	\bigg\{X(z)^{-1}\frac{\d}{\d z}X(z)\bigg\}^{11}=n\frac{\d}{\d z}g_{\theta}(z)-\frac{\mathcal{D}'(z-s_{n\theta})}{\mathcal{D}(z-s_{n\theta})}+\bigg\{\big(R(z-s_{n\theta})P(z-s_{n\theta})\big)^{-1}R'(z-s_{n\theta})P(z-s_{n\theta})\bigg\}^{11}.
\end{equation*}
Inserting the last identity into \eqref{h:2} results in
\begin{align*}
	\frac{n}{2\pi\im}\bigg[\int_{\gamma_+}-\int_{\gamma_-}&\bigg]\bigg\{X(z)^{-1}\frac{\d}{\d z}X(z)\bigg\}^{11}\,\frac{\partial V_{\theta}}{\partial\theta}(z)\,\d z=\frac{n^2}{2\pi\im}\bigg[\int_{\gamma_+}-\int_{\gamma_-}\bigg]\Big(\frac{\d}{\d z}g_{\theta}(z)\Big)\,\frac{\partial V_{\theta}}{\partial\theta}(z)\,\d z\\
	&-\frac{n}{2\pi\im}\bigg[\int_{\gamma_+}-\int_{\gamma_-}\bigg]\frac{\mathcal{D}'(z-s_{n\theta})}{\mathcal{D}(z-s_{n\theta})}\frac{\partial V_{\theta}}{\partial\theta}(z)\,\d z+\frac{n}{2\pi\im}\bigg[\int_{\gamma_+}-\int_{\gamma_-}\bigg]\bigg\{\big(R(z-s_{n\theta})P(z-s_{n\theta})\big)^{-1}\\
	&\hspace{2cm}\times\Big(\frac{\d}{\d z}R(z-s_{n\theta})\Big)P(z-s_{n\theta})\bigg\}^{11}\,\frac{\partial V_{\theta}}{\partial\theta}(z)\,\d z,
\end{align*}
and which establishes \eqref{h:1} after another application of Cauchy's theorem while using the properties of $\pi(z)$ as written below \eqref{a:15}. The proof of the Lemma is complete.
\end{proof}
Moving forward, we now evaluate the second integral term in the right hand side of \eqref{h:1} for large $n$.
\begin{lem} As $n\rightarrow\infty$,
\begin{align}
	\frac{n}{2\pi\im}\bigg[\int_{\gamma_+}-\int_{\gamma_-}\bigg]\frac{\mathcal{D}'(z)}{\mathcal{D}(z)}\frac{\partial V_{\theta}}{\partial\theta}(z+s_{n\theta})\,\d z=&\,-\frac{n\alpha}{2}\frac{\partial V_{\theta}}{\partial\theta}(b+s_{n\theta})+\frac{\partial}{\partial\theta}\Bigg[\frac{n\alpha}{2\pi}\int_E\frac{V_{\theta}(x)}{\sqrt{2-x^2}}\,\d x\bigg]+\mathcal{O}\big(n^{-\frac{1}{3}}\big)\label{h:5}
\end{align}
uniformly in $\theta\in[0,1]$ and $(s,\alpha,\beta)\in\mathbb{R}\times(-1,\infty)\in(\mathbb{C}\setminus(-\infty,0))$ on compact sets.
\end{lem}
\begin{proof} By \eqref{a:18}, for $z\in\mathbb{C}\setminus[-\sqrt{2}-s_{n\theta},\sqrt{2}]$,
\begin{equation*}
	\frac{\mathcal{D}'(z)}{\mathcal{D}(z)}=\frac{\alpha}{2}\frac{1}{z-\sqrt{2}}-\frac{\alpha}{2\sqrt{2}+s_{n\theta}}\frac{1}{(w^2-1)^{\frac{1}{2}}},\ \ \ \ w=1+2\,\frac{z-\sqrt{2}}{2\sqrt{2}+s_{n\theta}},
\end{equation*}
and so with Cauchy's theorem,
\begin{align*}
	\textnormal{LHS of}\ \eqref{h:5}=&\,-\frac{n\alpha}{2}\frac{\partial V_{\theta}}{\partial\theta}(\sqrt{2}+s_{n\theta})-\frac{n\alpha}{4\pi\im}\int_{-1}^1\Bigg[\frac{1}{(w^2-1)^{\frac{1}{2}}_+}-\frac{1}{(w^2-1)^{\frac{1}{2}}_-}\Bigg]\frac{\partial V_{\theta}}{\partial\theta}\bigg(\frac{2\sqrt{2}+s_{n\theta}}{2}w+\frac{s_{n\theta}}{2}\bigg)\,\d w\\
	=&\,-\frac{n\alpha}{2}\frac{\partial V_{\theta}}{\partial\theta}(\sqrt{2}+s_{n\theta})+\frac{n\alpha}{2\pi}\int_E\frac{1}{\sqrt{2-x^2}}\frac{\partial V_{\theta}}{\partial\theta}\bigg(x+\frac{s_{n\theta}}{2\sqrt{2}}(x+\sqrt{2})\bigg)\d x.
\end{align*}
This yields \eqref{h:5} by Taylor expansion, recalling that $s_{n\theta}$ in \eqref{a:12} and $\tau_{\theta}$ in \eqref{a:3} and using the integral identity, see \eqref{h:11} below,
\begin{equation*}
	\int_E\sqrt{\frac{\sqrt{2}+x}{\sqrt{2}-x}}V'_{\theta}(x)\,\d x=2\pi,\ \ \ \ \theta\in[0,1].
\end{equation*}
Our derivation of \eqref{h:5} is complete.
\end{proof}
The third integral in the right hand side of \eqref{h:1} involves $R(z)$ for $z\in\gamma_{\pm}$. So, we are once more in need of 
\begin{equation*}
	R(z)=I+\frac{1}{2\pi\im}\bigg[\oint_{\Sigma_a}+\oint_{\Sigma_b}\bigg]\big(G_R(w)-I\big)\frac{\d w}{w-z}+\frac{1}{2\pi\im}\oint_{\Sigma_b}\big(R_-(w)-I\big)\big(G_R(w)-I\big)\frac{\d w}{w-z}+\mathcal{O}\big(n^{-\frac{4}{3}}\big),
\end{equation*}
where $z\in\gamma_{\pm}$, and its large $n$-asymptotic that can be computed from residue theorem, and \eqref{a:25},\eqref{a:31},\eqref{g:3}.
\begin{prop} As $n\rightarrow\infty$, with $\sigma=\chi^{\sigma_3}\e^{\im\frac{\pi}{4}\sigma_3}$ and $\varkappa=-\sqrt{2}-s_{n\theta}$,
\begin{align}
	&\frac{1}{2\pi\im}\oint_{\Sigma_b}\big(G_R(w)-I\big)\frac{\d w}{w-z}=\sigma\begin{bmatrix}-1&1\\ -1&1\end{bmatrix}\sigma^{-1}\textcolor{orange}{\Bigg\{}\frac{1}{2}\sqrt{\frac{\sqrt{2}-\varkappa}{\zeta_1^b}}\,Q_1^{12}-\frac{s^2}{8}\sqrt{1-\frac{\varkappa}{\sqrt{2}}}\frac{1}{(\pi h_{V_{\theta}}(\sqrt{2}))^{\frac{1}{3}}}\textcolor{orange}{\Bigg\}}\frac{n^{-\frac{1}{3}}}{z-\sqrt{2}}\nonumber\\
	&+\sigma\textcolor{blue}{\Bigg\{}\frac{Q_1^{11}}{\zeta_1^b}\begin{bmatrix}-\alpha&1+\alpha\\ 1-\alpha&\alpha\end{bmatrix}+\frac{s^4}{32\sqrt{2}}\frac{1}{(\pi h_{V_{\theta}}(\sqrt{2}))^{\frac{2}{3}}}\begin{bmatrix}1&0\\ 0&1\end{bmatrix}-\frac{s^2}{8\sqrt{\zeta_1^b\sqrt{2}}}\frac{Q_1^{12}}{(\pi h_{V_{\theta}}(\sqrt{2}))^{\frac{1}{3}}}\begin{bmatrix}1-\alpha&1+\alpha\\ 1-\alpha&1+\alpha\end{bmatrix}\textcolor{blue}{\Bigg\}}\sigma^{-1}\nonumber\\
	&\times\frac{n^{-\frac{2}{3}}}{z-\sqrt{2}}+\sigma\textcolor{red}{\Bigg\{}\begin{bmatrix}-1&1\\ -1&1\end{bmatrix}\Bigg(\frac{1}{2\zeta_1^b}\sqrt{\frac{\sqrt{2}-\varkappa}{\zeta_1^b}}Q_2^{12}+\frac{s^3}{48\sqrt{2}}\sqrt{1-\frac{\varkappa}{\sqrt{2}}}\bigg(1-\frac{s^3}{16}\bigg)\frac{1}{\pi h_{V_{\theta}}(\sqrt{2})}+\frac{s^4}{64\sqrt{2}}\sqrt{\frac{\sqrt{2}-\varkappa}{\zeta_1^b}}\nonumber\\
	&\times\frac{Q_1^{12}}{(\pi h_{V_{\theta}}(\sqrt{2}))^{\frac{2}{3}}}-\frac{s^2}{8\zeta_1^b}\sqrt{1-\frac{\varkappa}{\sqrt{2}}}\frac{Q_1^{11}}{(\pi h_{V_{\theta}}(\sqrt{2}))^{\frac{1}{3}}}\Bigg)\frac{1}{z-\sqrt{2}}+\frac{1}{2(\zeta_1^b)^2}\sqrt{\frac{\zeta_1^b}{\sqrt{2}-\varkappa}}Q_1^{21}\begin{bmatrix}\alpha^2-1 & -(\alpha+1)^2\\ (\alpha-1)^2&1-\alpha^2\end{bmatrix}\nonumber\\
	&-\frac{Q_2^{12}}{4\zeta_1^b}\sqrt{\frac{\sqrt{2}-\varkappa}{\zeta_1^b}}\textcolor{olive}{\Bigg(}\frac{1}{\sqrt{2}-\varkappa}\begin{bmatrix}2\alpha^2+1 & -2\alpha^2-4\alpha-1\\ 2\alpha^2-4\alpha+1 & -2\alpha^2-1\end{bmatrix}+3\zeta_2^b\begin{bmatrix}-1&1\\ -1&1\end{bmatrix}\textcolor{olive}{\Bigg)}-\frac{s^3}{96\sqrt{2}}\sqrt{1-\frac{\varkappa}{\sqrt{2}}}\frac{1}{\pi h_{V_{\theta}}(\sqrt{2})}\nonumber\\
	&\times\textcolor{magenta}{\Bigg(}\bigg(1-\frac{s^3}{16}\bigg)\frac{1}{\sqrt{2}-\varkappa}\begin{bmatrix}3&1\\ -1&-3\end{bmatrix}-6\bigg(1-\frac{s^3}{80}\bigg)\bigg\{\frac{h_{V_{\theta}}'(\sqrt{2})}{h_{V_{\theta}}(\sqrt{2})}+\frac{1}{4\sqrt{2}}\bigg\}\begin{bmatrix}1&-1\\ 1&-1\end{bmatrix}\textcolor{magenta}{\Bigg)}-\frac{s^4}{128\sqrt{2}}\sqrt{\frac{\sqrt{2}-\varkappa}{\zeta_1^b}}\nonumber\\
	&\times\frac{Q_1^{12}}{(\pi h_{V_{\theta}}(\sqrt{2}))^{\frac{2}{3}}}\textcolor{brown}{\Bigg(}\frac{1}{\sqrt{2}-\varkappa}\begin{bmatrix}2\alpha^2+1 & -2\alpha^2-4\alpha-1\\ 2\alpha^2-4\alpha+1 & -2\alpha^2-1\end{bmatrix}+\bigg(\zeta_2^b+\frac{4}{5}\bigg\{\frac{h_{V_{\theta}}'(\sqrt{2})}{h_{V_{\theta}}(\sqrt{2})}+\frac{1}{4\sqrt{2}}\bigg\}\bigg)\begin{bmatrix}-1&1\\ -1&1\end{bmatrix}\textcolor{brown}{\Bigg)}\nonumber\\
	&+\frac{s^2}{8\zeta_1^b}\sqrt{1-\frac{\varkappa}{\sqrt{2}}}\frac{Q_1^{11}}{(\pi h_{V_{\theta}}(\sqrt{2}))^{\frac{1}{3}}}\textcolor{cyan}{\Bigg(}\frac{1}{\sqrt{2}-\varkappa}\begin{bmatrix}2\alpha^2-1 & -2\alpha^2-4\alpha-1\\ 2\alpha^2-4\alpha+1 & -2\alpha^2+1\end{bmatrix}-\bigg(\frac{1}{2(\sqrt{2}-\varkappa)}\nonumber\\
	&-\frac{1}{5}\bigg\{\frac{h_{V_{\theta}}'(\sqrt{2})}{h_{V_{\theta}}(\sqrt{2})}+\frac{1}{4\sqrt{2}}\bigg\}-\zeta_2^b\bigg)\begin{bmatrix}-1&1\\ -1&1\end{bmatrix}\textcolor{cyan}{\Bigg)}\textcolor{red}{\Bigg\}}\sigma^{-1}\frac{n^{-1}}{z-\sqrt{2}}+\mathcal{O}\big(n^{-\frac{4}{3}}\big)\label{h:6}
\end{align}
uniformly in $z\in\gamma_{\pm}$ and in $(s,\alpha,\beta)\in\mathbb{R}\times(-1,\infty)\times(\mathbb{C}\setminus(-\infty,0))$ on compact sets. Expansion \eqref{h:6} is differentiable with respect to $z\in\gamma_{\pm}$. Throughout $Q_k^{ij}=Q_k^{ij}(n^{\frac{2}{3}}\zeta^b(\sqrt{2}),\alpha,\beta)$ as in RHP \ref{YattRHP}.
\end{prop}
\begin{prop} As $n\rightarrow\infty$, with $\sigma=\chi^{\sigma_3}\e^{\im\frac{\pi}{4}\sigma_3}$ and $\varkappa=-\sqrt{2}-s_{n\theta}$,
\begin{align}
	\frac{1}{2\pi\im}&\oint_{\Sigma_b}\big(R_-(w)-I\big)\big(G_R(w)-I\big)\frac{\d w}{w-z}=\sigma\Bigg\{\frac{s^4}{32\sqrt{2}}\frac{1}{(\pi h_{V_{\theta}}(\sqrt{2}))^{\frac{2}{3}}}\begin{bmatrix}-1&1\\ 1&-1\end{bmatrix}+\frac{\alpha}{2\zeta_1^b}\big(Q_1^{12}\big)^2\nonumber\\
	&\times\begin{bmatrix}1&-1\\ 1&-1\end{bmatrix}+\frac{s^2}{8\sqrt{\zeta_1^b\sqrt{2}}}\frac{Q_1^{12}}{(\pi h_{V_{\theta}}(\sqrt{2}))^{\frac{1}{3}}}\begin{bmatrix}1-\alpha&-1+\alpha\\ -1-\alpha&1+\alpha\end{bmatrix}\Bigg\}\sigma^{-1}\frac{n^{-\frac{2}{3}}}{z-\sqrt{2}}+\sigma\textcolor{red}{\Bigg\{}-\frac{s^4}{128\sqrt{2}}\sqrt{\frac{\sqrt{2}-\varkappa}{\zeta_1^b}}\nonumber\\
	&\times\frac{Q_1^{12}}{(\pi h_{V_{\theta}}(\sqrt{2}))^{\frac{2}{3}}}\textcolor{brown}{\bigg\{}\frac{1}{\sqrt{2}-\varkappa}\begin{bmatrix}-2\alpha^2-1&2\alpha^2+4\alpha-7\\ -2\alpha^2+4\alpha+7&2\alpha^2+1\end{bmatrix}+\bigg(\zeta_2^b+\frac{4}{5}\bigg\{\frac{h_{V_{\theta}}'(\sqrt{2})}{h_{V_{\theta}}(\sqrt{2})}+\frac{1}{4\sqrt{2}}\bigg\}\bigg)\begin{bmatrix}1&-1\\ 1&-1\end{bmatrix}\textcolor{brown}{\bigg\}}\nonumber\\
	&+\frac{s^2}{16\sqrt{\zeta_1^b\sqrt{2}}}\sqrt{\frac{\sqrt{2}-\varkappa}{\zeta_1^b}}\frac{(Q_1^{12})^2}{(\pi h_{V_{\theta}}(\sqrt{2}))^{\frac{1}{3}}}\textcolor{orange}{\bigg\{}\frac{4\alpha}{\sqrt{2}-\varkappa}\begin{bmatrix}0&1\\ 1&0\end{bmatrix}+\bigg(\frac{1}{\sqrt{2}-\varkappa}+\frac{2}{5}\bigg\{\frac{h_{V_{\theta}}'(\sqrt{2})}{h_{V_{\theta}}(\sqrt{2})}+\frac{1}{4\sqrt{2}}\bigg\}\bigg)\begin{bmatrix}1&-1\\ 1&-1\end{bmatrix}\textcolor{orange}{\bigg\}}\nonumber\\
	&-\frac{\alpha^2}{2\zeta_1^b}\sqrt{\frac{\sqrt{2}-\varkappa}{\zeta_1^b}}\frac{(Q_1^{12})^3}{\sqrt{2}-\varkappa}\begin{bmatrix}1&-1\\ 1&-1\end{bmatrix}+\frac{s^2}{8\zeta_1^b}\sqrt{1-\frac{\varkappa}{\sqrt{2}}}\frac{Q_1^{11}}{(\pi h_{V_{\theta}}(\sqrt{2}))^{\frac{1}{3}}}\textcolor{magenta}{\bigg\{}\frac{1}{\sqrt{2}-\varkappa}\begin{bmatrix}-2\alpha^2-1 & 2\alpha^2-1\\ -2\alpha^2+1&2\alpha^2+1\end{bmatrix}\nonumber\\
	&+\bigg(\zeta_2^b+\frac{1}{2(\sqrt{2}-\varkappa)}-\frac{1}{5}\bigg\{\frac{h_{V_{\theta}}'(\sqrt{2})}{h_{V_{\theta}}(\sqrt{2})}+\frac{1}{4\sqrt{2}}\bigg\}\bigg)\begin{bmatrix}1&-1\\ 1&-1\end{bmatrix}\textcolor{magenta}{\bigg\}}-\frac{1}{4\zeta_1^b}\sqrt{\frac{\sqrt{2}-\varkappa}{\zeta_1^b}}Q_1^{11}Q_1^{12}\textcolor{olive}{\bigg\{}\frac{1}{\sqrt{2}-\varkappa}\nonumber
\end{align}
\begin{align}
	&\times\begin{bmatrix}-6\alpha^2+1 & 6\alpha^2+4\alpha-1\\ -6\alpha^2+4\alpha+1&6\alpha^2-1\end{bmatrix}+\zeta_2^b\begin{bmatrix}1&-1\\ 1&-1\end{bmatrix}\textcolor{olive}{\bigg\}}+\frac{s^6}{256\sqrt{2}}\sqrt{1-\frac{\varkappa}{\sqrt{2}}}\frac{1}{\pi h_{V_{\theta}}(\sqrt{2})}\textcolor{blue}{\bigg\{}\frac{1}{2}\bigg(\frac{1}{\sqrt{2}-\varkappa}\nonumber\\
	&+\frac{2}{5}\bigg\{\frac{h_{V_{\theta}}'(\sqrt{2})}{h_{V_{\theta}}(\sqrt{2})}+\frac{1}{4\sqrt{2}}\bigg\}\bigg)\begin{bmatrix}1&-1\\ 1&-1\end{bmatrix}-\frac{1}{\sqrt{2}-\varkappa}\begin{bmatrix}1&1\\ -1&-1\end{bmatrix}\textcolor{blue}{\bigg\}}\textcolor{red}{\Bigg\}}\sigma^{-1}\frac{n^{-1}}{z-\sqrt{2}}+\mathcal{O}\big(n^{-\frac{4}{3}}\big),\label{h:7}
\end{align}
uniformly in $z\in\gamma_{\pm}$ and in $(s,\alpha,\beta)\in\mathbb{R}\times(-1,\infty)\times(\mathbb{C}\setminus(-\infty,0))$ on compact sets. Expansion \eqref{h:7} is differentiable with respect to $z\in\gamma_{\pm}$. Throughout, $Q_k^{ij}=Q_k^{ij}(n^{\frac{2}{3}}\zeta^b(\sqrt{2}),\alpha,\beta)$ as in RHP \ref{YattRHP}.
\end{prop}
\begin{prop} As $n\rightarrow\infty$, with $\sigma=\chi^{\sigma_3}\e^{\im\frac{\pi}{4}\sigma_3}$ and $\varkappa=-\sqrt{2}-s_{n\theta}$,
\begin{align}
	\frac{1}{2\pi\im}&\oint_{\Sigma_a}\big(G_R(w)-I\big)\frac{\d w}{w-z}=\sigma\textcolor{red}{\Bigg\{}\frac{5}{96\sqrt{2}}\frac{1}{\pi h_{V_{\theta}}(-\sqrt{2})}\sqrt{1-\frac{\varkappa}{\sqrt{2}}}\begin{bmatrix}-1&-1\\ 1&1\end{bmatrix}\frac{1}{z-\varkappa}+\frac{5}{96\sqrt{2}}\sqrt{1-\frac{\varkappa}{\sqrt{2}}}\nonumber\\
	&\times\frac{1}{\pi h_{V_{\theta}}(-\sqrt{2})}\textcolor{orange}{\bigg\{}\frac{1}{2(\varkappa-\sqrt{2})}\begin{bmatrix}-(2\alpha^2+1)&-2\alpha^2-4\alpha-1\\ 2\alpha^2-4\alpha+1 & 2\alpha^2+1\end{bmatrix}-\frac{3}{5}\bigg\{\frac{h_{V_{\theta}}'(-\sqrt{2})}{h_{V_{\theta}}(-\sqrt{2})}-\frac{1}{4\sqrt{2}}\bigg\}\begin{bmatrix}-1&-1\\ 1&1\end{bmatrix}\textcolor{orange}{\bigg\}}\nonumber\\
	&+\frac{7}{192}\bigg(1-\frac{\varkappa}{\sqrt{2}}\bigg)^{-\frac{1}{2}}\frac{1}{\pi h_{V_{\theta}}(-\sqrt{2})}\begin{bmatrix}\alpha^2-1&(\alpha+1)^2\\ -(\alpha-1)^2&1-\alpha^2\end{bmatrix}\textcolor{red}{\Bigg\}}\sigma^{-1}\frac{n^{-1}}{z-\varkappa}+\mathcal{O}\big(n^{-\frac{5}{3}}\big)\label{h:8}
\end{align}
uniformly in $z\in\gamma_{\pm}$ and in $(s,\alpha,\beta)\in\mathbb{R}\times(-1,\infty)\times(\mathbb{C}\setminus(-\infty,0))$ on compact sets. Expansion \eqref{h:8} is differentiable with respect to $z\in\gamma_{\pm}$.
\end{prop}
We now proceed with the evaluation of the third integral in \eqref{h:1}, using that by Neumann series and the structure displayed in \eqref{h:6}, as $n\rightarrow\infty$, $R(z)^{-1} R'(z)=\sum_{j=1}^3R_j'(z)/\sqrt[3]{n^j}+\mathcal{O}(n^{-\frac{4}{3}})$,
%\begin{align*}
%	R(z)^{-1}\frac{\d R}{\d z}(z)=\Bigg\{\sum_{j=1}^2\frac{\d R_j}{\d z}(z)\frac{1}{\sqrt[3]{n^j}}+\bigg(\frac{\d R_3}{\d z}(z)-R_1(z)\frac{\d R_2}{\d z}(z)-R_2(z)\frac{\d R_1}{\d z}(z)\bigg)\frac{1}{n}\Bigg\}+\mathcal{O}\big(n^{-\frac{4}{3}}\big)%\\
%	%&+\big(P(z)\big)^{-1}\bigg[\frac{\d R_3}{\d z}(z)-R_1(z)\frac{\d R_2}{\d z}(z)-R_2(z)\frac{\d R_1}{\d z}(z)\bigg]P(z)\frac{1}{n}+\mathcal{O}\big(n^{-\frac{4}{3}}\big),
%\end{align*}
uniformly in $z\in\gamma_{\pm}$, where with the shorthand $Q_k^{ij}=Q_k^{ij}(n^{\frac{2}{3}}\zeta^b(\sqrt{2}),\alpha,\beta)$ throughout,
\begin{align*}
	R_1(z):=&\,\sigma\begin{bmatrix}-1&1\\ -1&1\end{bmatrix}\sigma^{-1}\textcolor{orange}{\Bigg\{}\frac{1}{2}\sqrt{\frac{\sqrt{2}-\varkappa}{\zeta_1^b}}\,Q_1^{12}-\frac{s^2}{8}\sqrt{1-\frac{\varkappa}{\sqrt{2}}}\frac{1}{(\pi h_{V_{\theta}}(\sqrt{2}))^{\frac{1}{3}}}\textcolor{orange}{\Bigg\}}\frac{1}{z-\sqrt{2}},\\
	R_2(z):=&\,\sigma\textcolor{blue}{\Bigg\{}\frac{Q_1^{11}}{\zeta_1^b}\begin{bmatrix}-\alpha&1+\alpha\\ 1-\alpha&\alpha\end{bmatrix}-\frac{s^2}{4\sqrt{\zeta_1^b\sqrt{2}}}\frac{Q_1^{12}}{(\pi h_{V_{\theta}}(\sqrt{2}))^{\frac{1}{3}}}\begin{bmatrix}0&1\\ 1&0\end{bmatrix}+\frac{\alpha}{2\zeta_1^b}\big(Q_1^{12}\big)^2\begin{bmatrix}1&-1\\ 1&-1\end{bmatrix}\\
	&\hspace{2cm}+\frac{s^4}{32\sqrt{2}}\frac{1}{(\pi h_{V_{\theta}}(\sqrt{2}))^{\frac{2}{3}}}\begin{bmatrix}0&1\\ 1&0\end{bmatrix}\textcolor{blue}{\Bigg\}}\sigma^{-1}\frac{1}{z-\sqrt{2}}.
\end{align*}
and the lengthy
\begin{align*}
	&R_3(z):=\sigma\textcolor{red}{\Bigg\{}\begin{bmatrix}-1&1\\ -1&1\end{bmatrix}\Bigg(\frac{1}{2\zeta_1^b}\sqrt{\frac{\sqrt{2}-\varkappa}{\zeta_1^b}}Q_2^{12}+\frac{s^3}{48b}\sqrt{1-\frac{\varkappa}{\sqrt{2}}}\bigg(1-\frac{s^3}{16}\bigg)\frac{1}{\pi h_{V_{\theta}}(\sqrt{2})}+\frac{s^4}{64\sqrt{2}}\sqrt{\frac{\sqrt{2}-\varkappa}{\zeta_1^b}}\\
	&\times\frac{Q_1^{12}}{(\pi h_{V_{\theta}}(\sqrt{2}))^{\frac{2}{3}}}-\frac{s^2}{8\zeta_1^b}\sqrt{1-\frac{\varkappa}{\sqrt{2}}}\frac{Q_1^{11}}{(\pi h_{V_{\theta}}(\sqrt{2}))^{\frac{1}{3}}}\Bigg)\frac{1}{z-\sqrt{2}}+\frac{1}{2(\zeta_1^b)^2}\sqrt{\frac{\zeta_1^b}{\sqrt{2}-\varkappa}}Q_1^{21}\begin{bmatrix}\alpha^2-1&-(\alpha+1)^2\\ (\alpha-1)^2&1-\alpha^2\end{bmatrix}\\
	&-\frac{Q_2^{12}}{4\zeta_1^b}\sqrt{\frac{\sqrt{2}-\varkappa}{\zeta_1^b}}\textcolor{olive}{\bigg(}\frac{1}{\sqrt{2}-\varkappa}\begin{bmatrix}2\alpha^2+1 & -2\alpha^2-4\alpha-1\\ 2\alpha^2-4\alpha+1 & -2\alpha^2-1\end{bmatrix}+3\zeta_2^b\begin{bmatrix}-1&1\\ -1&1\end{bmatrix}\textcolor{olive}{\bigg)}-\frac{s^4}{16\sqrt{2}}\sqrt{\frac{\sqrt{2}-\varkappa}{\zeta_1^b}}\\
	&\times\frac{Q_1^{12}}{(\pi h_{V_{\theta}}(\sqrt{2}))^{\frac{2}{3}}}\frac{1}{\sqrt{2}-\varkappa}\begin{bmatrix}0&-1\\ 1&0\end{bmatrix}+\frac{s^2}{8\zeta_1^b}\sqrt{1-\frac{\varkappa}{\sqrt{2}}}\frac{Q_1^{11}}{(\pi h_{V_{\theta}}(\sqrt{2}))^{\frac{1}{3}}}\textcolor{magenta}{\bigg(}\frac{1}{\sqrt{2}-\varkappa}\begin{bmatrix}-2&-4\alpha-2\\ -4\alpha+2&2\end{bmatrix}\\
	&+\bigg(\frac{1}{\sqrt{2}-\varkappa}-\frac{2}{5}\bigg\{\frac{h_{V_{\theta}}'(\sqrt{2})}{h_{V_{\theta}}(\sqrt{2})}+\frac{1}{4\sqrt{2}}\bigg\}\bigg)\begin{bmatrix}1&-1\\ 1&-1\end{bmatrix}\textcolor{magenta}{\bigg)}-\frac{\alpha^2}{2\zeta_1^b}\sqrt{\frac{\sqrt{2}-\varkappa}{\zeta_1^b}}\frac{(Q_1^{12})^3}{\sqrt{2}-\varkappa}\begin{bmatrix}1&-1\\ 1&-1\end{bmatrix}+\frac{s^2}{16\sqrt{\zeta_1^b\sqrt{2}}}
\end{align*}
\begin{align*}
	&\times\sqrt{\frac{\sqrt{2}-\varkappa}{\zeta_1^b}}\frac{(Q_1^{12})^2}{(\pi h_{V_{\theta}}(\sqrt{2}))^{\frac{1}{3}}}\textcolor{brown}{\bigg(}\frac{4\alpha}{\sqrt{2}-\varkappa}\begin{bmatrix}0&1\\ 1&0\end{bmatrix}+\bigg(\frac{1}{\sqrt{2}-\varkappa}+\frac{2}{5}\bigg\{\frac{h_{V_{\theta}}'(\sqrt{2})}{h_{V_{\theta}}(\sqrt{2})}+\frac{1}{4\sqrt{2}}\bigg\}\bigg)\begin{bmatrix}1&-1\\ 1&-1\end{bmatrix}\textcolor{brown}{\bigg)}\\
	&-\frac{s^3}{96\sqrt{2}}\sqrt{1-\frac{\varkappa}{\sqrt{2}}}\frac{1}{\pi h_{V_{\theta}}(\sqrt{2})}\textcolor{pink}{\bigg(}\bigg(1-\frac{s^3}{16}\bigg)\frac{1}{\sqrt{2}-\varkappa}\begin{bmatrix}3&1\\ -1&-3\end{bmatrix}-6\bigg\{\frac{h_{V_{\theta}}'(\sqrt{2})}{h_{V_{\theta}}(\sqrt{2})}+\frac{1}{4\sqrt{2}}\bigg\}\begin{bmatrix}1&-1\\ 1&-1\end{bmatrix}\textcolor{pink}{\bigg)}\\
	&+\frac{s^6}{512\sqrt{2}}\sqrt{1-\frac{\varkappa}{\sqrt{2}}}\frac{1}{\pi h_{V_{\theta}}(\sqrt{2})}\frac{1}{\sqrt{2}-\varkappa}\begin{bmatrix}-1&-3\\ 3&1\end{bmatrix}-\frac{1}{4\zeta_1^b}\sqrt{\frac{\sqrt{2}-\varkappa}{\zeta_1^b}}\\
	&\times Q_1^{11}Q_1^{12}\textcolor{brown}{\bigg(}\frac{1}{\sqrt{2}-\varkappa}\begin{bmatrix}-6\alpha^2+1&6\alpha^2+4\alpha-1\\ -6\alpha^2+4\alpha+1 & 6\alpha^2-1\end{bmatrix}+\zeta_2^b\begin{bmatrix}1&-1\\ 1&-1\end{bmatrix}\textcolor{brown}{\bigg)}
	\textcolor{red}{\Bigg\}}\sigma^{-1}\frac{1}{z-\sqrt{2}}\\
	&\hspace{1cm}+\sigma\textcolor{red}{\Bigg\{}\frac{5}{96\sqrt{2}}\frac{1}{\pi h_{V_{\theta}}(-\sqrt{2})}\sqrt{1-\frac{\varkappa}{\sqrt{2}}}\begin{bmatrix}-1&-1\\ 1&1\end{bmatrix}\frac{1}{z-\varkappa}+\frac{5}{96\sqrt{2}}\sqrt{1-\frac{\varkappa}{\sqrt{2}}}\frac{1}{\pi h_{V_{\theta}}(-\sqrt{2})}\textcolor{orange}{\bigg(}\frac{1}{2(\varkappa-\sqrt{2})}\\
	&\times\begin{bmatrix}-(2\alpha^2+1)&-2\alpha^2-4\alpha-1\\ 2\alpha^2-4\alpha+1 & 2\alpha^2+1\end{bmatrix}-\frac{3}{5}\bigg\{\frac{h_{V_{\theta}}'(-\sqrt{2})}{h_{V_{\theta}}(-\sqrt{2})}-\frac{1}{4\sqrt{2}}\bigg\}\begin{bmatrix}-1&-1\\ 1&1\end{bmatrix}\textcolor{orange}{\bigg)}+\frac{7}{192}\bigg(1-\frac{\varkappa}{\sqrt{2}}\bigg)^{-\frac{1}{2}}\\
	&\times\frac{1}{\pi h_{V_{\theta}}(-\sqrt{2})}\begin{bmatrix}\alpha^2-1&(\alpha+1)^2\\ -(\alpha-1)^2&1-\alpha^2\end{bmatrix}\textcolor{red}{\Bigg\}}\sigma^{-1}\frac{1}{z-\varkappa}.
\end{align*}
The aforementioned structure yields $0=R_1(z)R_1(z)=R_1(z)\frac{\d R_1}{\d z}(z)=R_1(z)\frac{\d R_2}{\d z}(z)+R_2(z)\frac{\d R_1}{\d z}(z)$ and thus the compact identity for $R(z)^{-1}\frac{\d R}{\d z}(z)$. Consequently, by \eqref{a:19}, uniformly in $z\in\gamma_{\pm}$ as $n\rightarrow\infty$,
\begin{align}
	\bigg\{&\big(R(z)P(z)\big)^{-1}\frac{\d R}{\d z}(z)P(z)\bigg\}^{11}=\frac{\nu^2(z)}{(z-\sqrt{2})^2}\textcolor{orange}{\Bigg\{}\frac{1}{2}\sqrt{\frac{\sqrt{2}-\varkappa}{\zeta_1^b}}\,Q_1^{12}-\frac{s^2}{8}\sqrt{1-\frac{\varkappa}{\sqrt{2}}}\frac{1}{(\pi h_{V_{\theta}}(\sqrt{2}))^{\frac{1}{3}}}\textcolor{orange}{\Bigg\}}\frac{1}{\sqrt[3]{n}}\nonumber\\
	&+\frac{\alpha\nu^2(z)}{\zeta_1^b(z-\sqrt{2})^2}\textcolor{blue}{\bigg\{}Q_1^{11}-\frac{1}{2}\big(Q_1^{12}\big)^2\textcolor{blue}{\bigg\}}\frac{1}{\sqrt[3]{n^2}}+\textcolor{red}{\Bigg\{}\Bigg(\frac{1}{2\zeta_1^b}\sqrt{\frac{\sqrt{2}-\varkappa}{\zeta_1^b}}Q_2^{12}+\frac{s^3}{48\sqrt{2}}\sqrt{1-\frac{\varkappa}{\sqrt{2}}}\bigg(1-\frac{s^3}{16}\bigg)\frac{1}{\pi h_{V_{\theta}}(\sqrt{2})}\nonumber\\
	&+\frac{s^4}{64\sqrt{2}}\sqrt{\frac{\sqrt{2}-\varkappa}{\zeta_1^b}}\frac{Q_1^{12}}{(\pi h_{V_{\theta}}(\sqrt{2}))^{\frac{2}{3}}}-\frac{s^2}{8\zeta_1^b}\sqrt{1-\frac{\varkappa}{\sqrt{2}}}\frac{Q_1^{11}}{(\pi h_{V_{\theta}}(\sqrt{2}))^{\frac{1}{3}}}\Bigg)\frac{2\nu^2(z)}{(z-\sqrt{2})^3}+\frac{1}{2(\zeta_1^b)^2}\sqrt{\frac{\zeta_1^b}{\sqrt{2}-\varkappa}}Q_1^{21}\nonumber\\
	&\times\frac{\nu^{-2}(z)-\alpha^2\nu^2(z)}{(z-\sqrt{2})^2}-\frac{Q_2^{12}}{4\zeta_1^b}\sqrt{\frac{\sqrt{2}-\varkappa}{\zeta_1^b}}\bigg(3\zeta_2^b-\frac{2\alpha^2+1}{\sqrt{2}-\varkappa}\bigg)\frac{\nu^2(z)}{(z-\sqrt{2})^2}+\frac{s^4}{32\sqrt{2}}\sqrt{\frac{\sqrt{2}-\varkappa}{\zeta_1^b}}\frac{Q_1^{12}}{(\pi h_{V_{\theta}}(\sqrt{2}))^{\frac{2}{3}}}\nonumber\\
	&\times\frac{\nu^2(z)-\nu^{-2}(z)}{(\sqrt{2}-\varkappa)(z-\sqrt{2})^2}+\frac{s^2}{8\zeta_1^b}\sqrt{1-\frac{\varkappa}{\sqrt{2}}}\frac{Q_1^{11}}{(\pi h_{V_{\theta}}(\sqrt{2}))^{\frac{1}{3}}}\bigg(\frac{2\nu^{-2}(z)}{(\sqrt{2}-\varkappa)(z-\sqrt{2})^2}-\bigg(\frac{1}{\sqrt{2}-\varkappa}-\frac{2}{5}\bigg\{\frac{h_{V_{\theta}}'(\sqrt{2})}{h_{V_{\theta}}(\sqrt{2})}\nonumber\\
	&+\frac{1}{4\sqrt{2}}\bigg\}\bigg)\frac{\nu^2(z)}{(z-\sqrt{2})^2}\bigg)+\frac{\alpha^2}{2\zeta_1^b}\sqrt{\frac{b-a}{\zeta_1^b}}\frac{(Q_1^{12})^3}{\sqrt{2}-\varkappa}\frac{\nu^2(z)}{(z-\sqrt{2})^2}-\frac{s^2}{16\sqrt{\zeta_1^b\sqrt{2}}}\sqrt{\frac{\sqrt{2}-\varkappa}{\zeta_1^b}}\frac{(Q_1^{12})^2}{(\pi h_{V_{\theta}}(\sqrt{2}))^{\frac{1}{3}}}\nonumber\\
	&\times\bigg(\frac{1}{\sqrt{2}-\varkappa}+\frac{2}{5}\bigg\{\frac{h_{V_{\theta}}'(\sqrt{2})}{h_{V_{\theta}}(\sqrt{2})}+\frac{1}{4\sqrt{2}}\bigg\}\bigg)\frac{\nu^2(z)}{(z-\sqrt{2})^2}+\frac{s^3}{96\sqrt{2}}\sqrt{1-\frac{\varkappa}{\sqrt{2}}}\frac{1}{\pi h_{V_{\theta}}(\sqrt{2})}\Bigg(\bigg(1-\frac{s^3}{16}\bigg)\nonumber\\
	&\times\frac{2\nu^{-2}(z)+\nu^2(z)}{(\sqrt{2}-\varkappa)(z-\sqrt{2})^2}-6\bigg\{\frac{h_{V_{\theta}}'(\sqrt{2})}{h_{V_{\theta}}(\sqrt{2})}+\frac{1}{4\sqrt{2}}\bigg\}\frac{\nu^2(z)}{(z-\sqrt{2})^2}\Bigg)+\frac{s^6}{512\sqrt{2}}\sqrt{1-\frac{\varkappa}{\sqrt{2}}}\frac{1}{\pi h_{V_{\theta}}(\sqrt{2})}\nonumber\\
	&\times\frac{2\nu^{-2}(z)-\nu^2(z)}{(\sqrt{2}-\varkappa)(z-\sqrt{2})^2}+\frac{5}{48\sqrt{2}}\frac{1}{\pi h_{V_{\theta}}(-\sqrt{2})}\sqrt{1-\frac{\varkappa}{\sqrt{2}}}\frac{\nu^{-2}(z)}{(z-\varkappa)^3}+\frac{5}{96\sqrt{2}}\sqrt{1-\frac{\varkappa}{\sqrt{2}}}\frac{1}{\pi h_{V_{\theta}}(-\sqrt{2})}\bigg(\frac{2\alpha^2+1}{2(\varkappa-\sqrt{2})}\nonumber\\
	&-\frac{3}{5}\bigg\{\frac{h_{V_{\theta}}'(-\sqrt{2})}{h_{V_{\theta}}(-\sqrt{2})}-\frac{1}{4\sqrt{2}}\bigg\}\bigg)\frac{\nu^{-2}(z)}{(z-\varkappa)^2}+\frac{7}{192}\bigg(1-\frac{\varkappa}{\sqrt{2}}\bigg)^{-\frac{1}{2}}\frac{1}{\pi h_{V_{\theta}}(-\sqrt{2})}\frac{\nu^2(z)-\alpha^2\nu^{-2}(z)}{(z-\varkappa)^2}\nonumber
\end{align}
\begin{align}
	&+\frac{1}{4\zeta_1^b}\sqrt{\frac{\sqrt{2}-\varkappa}{\zeta_1^b}}Q_1^{11}Q_1^{12}\bigg(\frac{1-6\alpha^2}{\sqrt{2}-\varkappa}+\zeta_2^b\bigg)\frac{\nu^2(z)}{(z-\sqrt{2})^2}\textcolor{red}{\Bigg\}}\frac{1}{n}+\mathcal{O}\big(n^{-\frac{4}{3}}\big),\label{h:9}
\end{align}
and so we are naturally led to the evaluation of the model integrals ($\gamma=\gamma_+-\gamma_-$)
\begin{equation}\label{h:10}
	\int_{\gamma}f(z)V(z+s_{n\theta})\d z,\ \ \ f\in\bigg\{\frac{\nu^2(z)}{(z-\sqrt{2})^3},\frac{\nu^2(z)}{(z-\sqrt{2})^2},\frac{\nu^{-2}(z)}{(z-\sqrt{2})^2},\frac{\nu^{-2}(z)}{(z-\varkappa)^2},\frac{\nu^{-2}(z)}{(z-\varkappa)^3},\frac{\nu^2(z)}{(z-\varkappa)^2}\bigg\},
\end{equation}
with one-cut regular $V$ and $\nu(z)=((z-\sqrt{2})/(z-\varkappa))^{1/4}$. We draw inspiration from \cite[page $150-156$]{BWW}:
\begin{lem} Suppose $V$ is as in Assumption \ref{1-cut}, with $h_V$ as in \eqref{e:11}, then for any $x\in(-\sqrt{2},\sqrt{2})$,
\begin{equation}\label{h:11}
	\textnormal{pv}\int_E\frac{\sqrt{2-\lambda^2}}{\lambda-x}V'(\lambda)\,\d\lambda=-2\pi+2\pi^2h_V(x)(2-x^2).
\end{equation}
\end{lem}
\begin{proof} Exactly as in \cite[Lemma $5.8$]{BWW}: one considers the function
\begin{equation*}
	\mathbb{C}\setminus E\ni z\mapsto 2\pi(z^2-2)^{\frac{1}{2}}\int_E\frac{\rho_V(x)}{x-z}\d x+\int_E\frac{\sqrt{2-x^2}}{x-z}V'(x)\,\d x
\end{equation*}
that is defined with principal branches. What results from the Plemelj-Sokhostki identity and equality above \eqref{e:10} is analyticity everywhere of the same and thus it can be computed by Liouville's theorem. Identity \eqref{h:11} is simply the sum of the limiting values of the same function in the interior of $E$.
\end{proof}
\begin{prop} Suppose $V$ is as in Assumption \ref{1-cut}. Then as $n\rightarrow\infty$,
\begin{equation}\label{h:12}
	\bigg[\int_{\gamma_+}-\int_{\gamma_-}\bigg]\frac{\nu^2(z)}{(z-\sqrt{2})^3}V(z+s_{n\theta})\,\d z=\frac{2\pi\im}{3}\Big(1-4\pi h_V(\sqrt{2})+\mathcal{O}\big(n^{-\frac{2}{3}}\big)\Big)
\end{equation}
with $s_{n\theta}=s/(\pi h_{V_{\theta}}(\sqrt{2})2^{\frac{3}{4}}n)^{\frac{2}{3}}$, see \eqref{a:3},\eqref{a:12}.
\end{prop}
\begin{proof} We note that
\begin{equation*}
	\frac{\nu^2(z)}{(z-\sqrt{2})^3}=\frac{4}{3(\sqrt{2}-\varkappa)^2}\frac{\d}{\d z}\bigg[\frac{(z-\frac{3\sqrt{2}-\varkappa}{2})((z-\varkappa)(z-\sqrt{2}))^{\frac{1}{2}}}{(z-\sqrt{2})^2}\bigg],\ \ \ z\in\mathbb{C}\setminus[\varkappa,\sqrt{2}],
\end{equation*}
and so integration by parts and Cauchy's theorem yield, with the shorthand $\gamma:=\gamma_+-\gamma_-$,
\begin{align*}
	&-\frac{3}{4}(\sqrt{2}-\varkappa)^2\,\times\textnormal{LHS in}\ \eqref{h:12}=\int_{\gamma}\frac{(z-\frac{3\sqrt{2}-\varkappa}{2})((z-\varkappa)(z-\sqrt{2}))^{\frac{1}{2}}}{(z-\sqrt{2})^2}V'(z+s_{n\theta})\,\d z\\
	=&\int_{\gamma}\bigg(\frac{z-\varkappa}{z-\sqrt{2}}\bigg)^{\frac{1}{2}}V'(z+s_{n\theta})\,\d z-\frac{\sqrt{2}-\varkappa}{2}\int_{\gamma}\frac{((z-\varkappa)(z-\sqrt{2}))^{\frac{1}{2}}}{(z-\sqrt{2})^2}V'(z+s_{n\theta})\,\d z\\
	=&-2\im\int_{\varkappa}^{\sqrt{2}}\sqrt{\frac{x-\varkappa}{\sqrt{2}-x}}V'(x+s_{n\theta})\,\d x-\frac{1}{2}(\sqrt{2}-\varkappa)\frac{\d}{\d x}\int_{\gamma}\frac{((z-\varkappa)(z-\sqrt{2}))^{\frac{1}{2}}}{z-x}V'(z+s_{n\theta})\,\d z\bigg|_{x=\sqrt{2}}\\
	=&\,-2\im\int_E\sqrt{\frac{x+\sqrt{2}}{\sqrt{2}-x}}V'(x)\,\d x-2\im\sqrt{2}\frac{\d}{\d x}\textnormal{pv}\int_E\frac{\sqrt{2-\lambda^2}}{\lambda-x}V'(\lambda)\d\lambda\,\bigg|_{x=\sqrt{2}}+\mathcal{O}\big(n^{-\frac{2}{3}}\big)\\
	=&\,-4\pi\im\Big(1-4\pi h_V(\sqrt{2})+\mathcal{O}\big(n^{-\frac{2}{3}}\big)\Big),
\end{align*}
where we have used \eqref{h:11} in the last equality. The proof is complete.
\end{proof}
\begin{prop} Suppose $V$ is as in Assumption \ref{1-cut}. Then as $n\rightarrow\infty$,
\begin{equation}\label{h:13}
	\bigg[\int_{\gamma_+}-\int_{\gamma_-}\bigg]\frac{\nu^2(z)}{(z-\sqrt{2})^2}V(z+s_{n\theta})\,\d z=-\frac{4\pi\im}{\sqrt{2}}\Big(1+s_{n\theta}\bigg(\pi\sqrt{2}h_V(\sqrt{2})-\frac{1}{2\sqrt{2}}\bigg)+\mathcal{O}\big(n^{-\frac{4}{3}}\big)\Big),
\end{equation}
with $s_{n\theta}=s/(\pi h_{V_{\theta}}(\sqrt{2})2^{\frac{3}{4}}n)^{\frac{2}{3}}$, see \eqref{a:3},\eqref{a:12}.
\end{prop}
\begin{proof} Use
\begin{equation*}
	\frac{\nu^2(z)}{(z-\sqrt{2})^2}=-\frac{2}{\sqrt{2}-\varkappa}\frac{\d}{\d z}\bigg(\frac{z-\varkappa}{z-\sqrt{2}}\bigg)^{\frac{1}{2}},\ \ \ z\in\mathbb{C}\setminus [\varkappa,\sqrt{2}],
\end{equation*}
and integrate by parts, afterwards expand for large $n$. What results is, with the previous $\gamma=\gamma_+-\gamma_-$,
\begin{align*}
	&\textnormal{LHS in}\ \eqref{h:13}=\frac{2}{\sqrt{2}-\varkappa}\int_{\gamma}\bigg(\frac{z-\varkappa}{z-\sqrt{2}}\bigg)^{\frac{1}{2}}V'(z+s_{n\theta})\,\d z=\frac{1}{\sqrt{2}}\Bigg[\int_{\gamma}\bigg(\frac{z+\sqrt{2}}{z-\sqrt{2}}\bigg)^{\frac{1}{2}}V'(z)\,\d z+\mathcal{O}\big(n^{-\frac{4}{3}}\big)\\
	&+s_{n\theta}\bigg\{\int_{\gamma}\bigg(\frac{z+\sqrt{2}}{z-\sqrt{2}}\bigg)^{\frac{1}{2}}V''(z)\,\d z-\frac{1}{2\sqrt{2}}\int_{\gamma}\bigg(\frac{z+\sqrt{2}}{z-\sqrt{2}}\bigg)^{\frac{1}{2}}V'(z)\,\d z+\frac{1}{2}\int_{\gamma}\frac{V'(z)}{((z-\sqrt{2})(z+\sqrt{2}))^{\frac{1}{2}}}\,\d z\bigg\}\Bigg],
\end{align*}
where by \eqref{h:12} and Cauchy's theorem,
\begin{equation}\label{h:14}
	\int_{\gamma}\bigg(\frac{z+\sqrt{2}}{z-\sqrt{2}}\bigg)^{\frac{1}{2}}V'(z)\,\d z=-2\im\int_E\sqrt{\frac{\sqrt{2}+x}{\sqrt{2}-x}}V'(x)\,\d x=-4\pi\im.
\end{equation}
On the other hand, by Cauchy's theorem and the properties of the $g$-function, compare also \cite[$(6.143)$]{D},
\begin{equation*}
	\int_{\gamma}\frac{V'(z)}{((z-\sqrt{2})(z+\sqrt{2}))^{\frac{1}{2}}}\,\d z=-2\im\int_E\frac{V'(x)}{\sqrt{2-x^2}}\,\d x=0;\hspace{1.5cm}
	\frac{1}{2\pi}\int_E\sqrt{2-x^2}\,V''(x)\,\d x=1.
\end{equation*}
Consequently
\begin{align}
	\int_{\gamma}\bigg(\frac{z+\sqrt{2}}{z-\sqrt{2}}\bigg)^{\frac{1}{2}}V''(z)\,\d z=&\,-2\im\int_E\sqrt{\frac{\sqrt{2}+x}{\sqrt{2}-x}}V''(x)\,\d x=-\im\sqrt{2}\Bigg[2\pi+\int_E\sqrt{\frac{\sqrt{2}+x}{\sqrt{2}-x}}\,xV''(x)\,\d x\Bigg]\nonumber\\
	=&\,-\im\sqrt{2}\Bigg[2\pi+\frac{\im}{2}\int_{\gamma}\bigg(\frac{z+\sqrt{2}}{z-\sqrt{2}}\bigg)^{\frac{1}{2}}\,zV''(z)\,\d z\Bigg].\label{h:15}
\end{align}
At this point we integrate by parts,
\begin{align}\label{h:16}
	\frac{\im}{2}\int_{\gamma}\bigg(\frac{z+\sqrt{2}}{z-\sqrt{2}}\bigg)^{\frac{1}{2}}\,zV''(z)\,\d z=-\frac{\im}{2}\int_{\gamma}\bigg(\frac{z+\sqrt{2}}{z-\sqrt{2}}\bigg)^{\frac{1}{2}}\,V'(z)\,\d z+\frac{\im}{\sqrt{2}}\int_{\gamma}\bigg(\frac{z-\sqrt{2}}{z+\sqrt{2}}\bigg)^{\frac{1}{2}}\frac{zV'(z)}{(z-\sqrt{2})^2}\,\d z
\end{align}
and realize that the first integral in \eqref{h:16} is known from \eqref{h:14}. For the second integral we use partial fractions,
\begin{align*}
	\int_{\gamma}&\bigg(\frac{z-\sqrt{2}}{z+\sqrt{2}}\bigg)^{\frac{1}{2}}\frac{zV'(z)}{(z-\sqrt{2})^2}\,\d z=\frac{\d}{\d x}\int_{\gamma}\frac{((z-\sqrt{2})(z+\sqrt{2}))^{\frac{1}{2}}}{z+\sqrt{2}}\frac{zV'(z)}{z-x}\,\d z\,\bigg|_{x=\sqrt{2}}\\
	&=\frac{\d}{\d x}\Bigg[\int_{\gamma}\big((z-\sqrt{2})(z+\sqrt{2})\big)^{\frac{1}{2}}\frac{V'(z)}{z-x}\,\d z-\sqrt{2}\int_{\gamma}\big((z-\sqrt{2})(z+\sqrt{2})\big)^{\frac{1}{2}}\frac{V'(z)}{(z+\sqrt{2})(z-x)}\Bigg]\Bigg|_{x=\sqrt{2}}\\
	&=\frac{\d}{\d x}\Bigg[\frac{x}{x+\sqrt{2}}\int_{\gamma}\big((z-\sqrt{2})(z+\sqrt{2})\big)^{\frac{1}{2}}\frac{V'(z)}{z-x}\,\d z+\frac{\sqrt{2}}{x+\sqrt{2}}\int_{\gamma}\bigg(\frac{z-\sqrt{2}}{z+\sqrt{2}}\bigg)^{\frac{1}{2}}V'(z)\,\d z\Bigg]\Bigg|_{x=\sqrt{2}}\\
	&=\frac{\d}{\d x}\Bigg[\frac{2\im x}{x+\sqrt{2}}\,\textnormal{pv}\int_E\frac{\sqrt{2-\lambda^2}}{\lambda-x}V'(\lambda)\,\d\lambda-\frac{4\sqrt{2}\,\pi\im}{x+\sqrt{2}}\Bigg]\Bigg|_{x=\sqrt{2}}\stackrel{\eqref{h:11}}{=}-4\pi^2\im\sqrt{2}\,h_V(\sqrt{2}).
\end{align*}
Consequently back in \eqref{h:15},
\begin{equation*}
	\int_{\gamma}\bigg(\frac{z+\sqrt{2}}{z-\sqrt{2}}\bigg)^{\frac{1}{2}}V''(z)\,\d z=-4\pi^2\im\sqrt{2}\,h_V(\sqrt{2}),
\end{equation*}
and this completes the proof of \eqref{h:13}.
\end{proof}
\begin{prop} Suppose $V$ is as in Assumption \ref{1-cut}. Then as $n\rightarrow\infty$,
\begin{equation}\label{h:17}
	\bigg[\int_{\gamma_+}-\int_{\gamma_-}\bigg]\frac{\nu^{-2}(z)}{(z-\sqrt{2})^2}V(z+s_{n\theta})\,\d z=-\frac{4\pi\im}{3\sqrt{2}}\Big(1+8\pi h_V(\sqrt{2})+\mathcal{O}\big(n^{-\frac{2}{3}}\big)\Big).
\end{equation}
with $s_{n\theta}=s/(\pi h_{V_{\theta}}(\sqrt{2})2^{\frac{3}{4}}n)^{\frac{2}{3}}$, see \eqref{a:3},\eqref{a:12}.
\end{prop}
\begin{proof} We use
\begin{equation*}
	\frac{\nu^{-2}(z)}{(z-\sqrt{2})^2}=-\frac{2}{3}\frac{\d}{\d z}\bigg[\frac{((z-\varkappa)(z-\sqrt{2}))^{\frac{1}{2}}}{(z-\sqrt{2})^2}\bigg]+\frac{1}{3}\frac{\nu^2(z)}{(z-\sqrt{2})^2},\ \ \ z\in\mathbb{C}\setminus [\varkappa,\sqrt{2}],
\end{equation*}
and obtain from integration by parts and \eqref{h:13},
\begin{align*}
	\textnormal{LHS in}\ \eqref{h:17}=&\,\frac{2}{3}\frac{\d}{\d x}\int_{\gamma}\frac{((z+\sqrt{2})(z-\sqrt{2}))^{\frac{1}{2}}}{z-x}V'(z)\,\d z\,\Bigg|_{x=\sqrt{2}}-\frac{4\pi\im}{3\sqrt{2}}+\mathcal{O}\big(n^{-\frac{2}{3}}\big)\\
	=&\,\frac{4\im}{3}\frac{\d}{\d x}\,\textnormal{pv}\int_E\frac{\sqrt{2-\lambda^2}}{\lambda-x}V'(\lambda)\,\d\lambda\,\Bigg|_{x=\sqrt{2}}-\frac{4\pi\im}{3\sqrt{2}}+\mathcal{O}\big(n^{-\frac{2}{3}}\big).
\end{align*}
This implies \eqref{h:17} by another application of \eqref{h:11} and completes the proof.
\end{proof}
\begin{prop} Suppose $V$ is as in Assumption \ref{1-cut}. Then as $n\rightarrow\infty$,
\begin{align}
	\bigg[\int_{\gamma_+}-\int_{\gamma_-}\bigg]&\frac{\nu^{-2}(z)}{(z-\varkappa)^2}V(z+s_{n\theta})\,\d z=\frac{4\pi\im}{\sqrt{2}}+\mathcal{O}\big(n^{-\frac{2}{3}}\big),\label{h:18}\\
	\bigg[\int_{\gamma_+}-\int_{\gamma_-}\bigg]&\frac{\nu^{2}(z)}{(z-\varkappa)^2}V(z+s_{n\theta})\,\d z=\frac{4\pi\im}{3\sqrt{2}}\Big(1+8\pi h_V(-\sqrt{2})+\mathcal{O}\big(n^{-\frac{2}{3}}\big)\Big),\label{h:19}\\
	\bigg[\int_{\gamma_+}-\int_{\gamma_-}\bigg]&\frac{\nu^{-2}(z)}{(z-\varkappa)^3}V(z+s_{n\theta})\,\d z=\frac{2\pi\im}{3}\Big(1-4\pi h_V(-\sqrt{2})+\mathcal{O}\big(n^{-\frac{2}{3}}\big)\Big),\label{h:20}
\end{align}
with $s_{n\theta}=s/(\pi h_{V_{\theta}}(\sqrt{2})2^{\frac{3}{4}}n)^{\frac{2}{3}}$, see \eqref{a:3},\eqref{a:12}.
\end{prop}
\begin{proof} Noting that 
\begin{equation*}
	\frac{\nu^{-2}(z)}{(z-\varkappa)^2}\ \  \textnormal{and}\ \ \frac{\nu^2(z)}{(z-\sqrt{2})^2};\ \ \ \ \frac{\nu^{2}(z)}{(z-\varkappa)^2}\ \ \textnormal{and}\ \ \frac{\nu^{-2}(z)}{(z-\sqrt{2})^2};\ \ \ \ \frac{\nu^{-2}(z)}{(z-\varkappa)^3}\ \ \textnormal{and}\ \ \frac{\nu^{2}(z)}{(z-\sqrt{2})^3},
\end{equation*}
are the same up to the exchange $\varkappa\leftrightarrow\sqrt{2}$, \eqref{h:18},\eqref{h:19} and \eqref{h:20} follows from the workings leading to \eqref{h:13},\eqref{h:17} and \eqref{h:12}.
\end{proof}
%\begin{prop} Suppose $V$ is as in Assumption \ref{1-cut}. Then as $n\rightarrow\infty$,
%\begin{equation}\label{h:19}
%	\bigg[\int_{\gamma_+}-\int_{\gamma_-}\bigg]\frac{\nu^{2}(z)}{(z-\varkappa)^2}V(z+s_{n\theta})\,\d z=\frac{4\pi\im}{3\sqrt{2}}\Big(1+8\pi h_V(-\sqrt{2})+\mathcal{O}\big(n^{-\frac{2}{3}}\big)\Big),
%\end{equation}
%with $s_{n\theta}=s/(\pi h_{V_{\theta}}(\sqrt{2})2^{\frac{3}{4}}n)^{\frac{2}{3}}$, see \eqref{a:3},\eqref{a:12}.
%\end{prop}
%\begin{proof} Since
%\begin{equation*}
%	\frac{\nu^{2}(z)}{(z-\varkappa)^2}\ \ \ \ \ \ \textnormal{and}\ \ \ \ \ \ \frac{\nu^{-2}(z)}{(z-\sqrt{2})^2}
%\end{equation*}
%are the same up to the exchange $\varkappa\leftrightarrow\sqrt{2}$, \eqref{h:19} follows from the workings leading to \eqref{h:17}.
%\end{proof}
%\begin{prop} Suppose $V$ is as in Assumption \ref{1-cut}. Then as $n\rightarrow\infty$,
%\begin{equation}\label{h:20}
%	\bigg[\int_{\gamma_+}-\int_{\gamma_-}\bigg]\frac{\nu^{-2}(z)}{(z-\varkappa)^3}V(z+s_{n\theta})\,\d z=\frac{2\pi\im}{3}\Big(1-4\pi h_V(-\sqrt{2})+\mathcal{O}\big(n^{-\frac{2}{3}}\big)\Big),
%\end{equation}
%with $s_{n\theta}=s/(\pi h_{V_{\theta}}(\sqrt{2})2^{\frac{3}{4}}n)^{\frac{2}{3}}$, see \eqref{a:3},\eqref{a:12}.
%\end{prop}
%\begin{proof} Since
%\begin{equation*}
%	\frac{\nu^{-2}(z)}{(z-\varkappa)^3}\ \ \ \ \ \ \textnormal{and}\ \ \ \ \ \ \frac{\nu^{2}(z)}{(z-\sqrt{2})^3}
%\end{equation*}
%are the same up to the exchange $\varkappa\leftrightarrow\sqrt{2}$, \eqref{h:20} follows from the workings leading to \eqref{h:12}.
%\end{proof}

At this point of our calculation we return to expansion \eqref{h:9} and carry out the relevant $z$-integration along $\gamma$ with the help of \eqref{h:12},\eqref{h:13},\eqref{h:17},\eqref{h:18},\eqref{h:19},\eqref{h:20}. Note that our potential is $V_{\theta}(x)=(1-\theta)x^2+\theta V(x)$, as convex combination of two one-cut regular potentials, see \eqref{a:1}. Consequently, $\frac{\partial}{\partial\theta}V_{\theta}(x)=V(x)-x^2$ in the upcoming integral \eqref{h:21} is the \textit{difference} of two one-cut regular potentials.

\begin{prop} As $n\rightarrow\infty$, with $\gamma=\gamma_+-\gamma_-$ and $a=a(s,\alpha,\beta),Q_i^{jk}=Q_i^{jk}(s,\alpha,\beta)$ as in Lemma \ref{ArnoLax},
\begin{align}\label{h:21}
	\frac{n}{2\pi\im}\int_{\gamma}&\,\bigg\{\big(R(z)P(z)\big)^{-1}\Big(\frac{\d}{\d z}R(z)\Big)P(z)\bigg\}^{11}\frac{\partial V_{\theta}}{\partial\theta}(z+s_{n\theta})\,\d z\nonumber\\
	&\,=\bigg[\frac{s^3}{6}-sa-\frac{2}{3}\big(Q_1^{21}+2Q_2^{12}\big)\bigg]\frac{\partial}{\partial\theta}\ln\big(h_{V_{\theta}}(\sqrt{2})\big)-\frac{1}{24}\frac{\partial}{\partial\theta}\ln\big(h_{V_{\theta}}(-\sqrt{2})\big)+\mathcal{O}\big(n^{-\frac{1}{3}}\big)
\end{align}
uniformly in $\theta\in[0,1]$ and in $(s,\alpha,\beta)\in\mathbb{R}\times(-1,\infty)\times(\mathbb{C}\setminus(-\infty,0))$ when chosen on compact sets. 
\end{prop}
\begin{proof} We proceed according to \eqref{h:9}, using the same colour coding, and thus begin with
\begin{align*}
	\frac{n}{2\pi\im}\int_{\gamma}\frac{\nu^2(z)}{(z-\sqrt{2})^2}\big(V(z+s_{n\theta})-(z+s_{n\theta})^2\big)\,\frac{\d z}{\sqrt[3]{n}}\textcolor{orange}{\bigg\{}\ldots\textcolor{orange}{\bigg\}}\stackrel{\eqref{h:13}}{=}-s\bigg(a-\frac{s^2}{4}\bigg)\bigg\{\frac{\partial}{\partial\theta}\ln\big(h_{V_{\theta}}(\sqrt{2})\big)+\mathcal{O}\big(n^{-\frac{2}{3}}\big)\bigg\},
\end{align*}
followed by
\begin{equation*}
	\frac{n}{2\pi\im}\int_{\gamma}\frac{\alpha\nu^2(z)}{\zeta_1^b(z-\sqrt{2})^2}\big(V(z+s_{n\theta})-(z+s_{n\theta})^2\big)\frac{\d z}{\sqrt[3]{n^2}}\textcolor{blue}{\bigg\{}\ldots\textcolor{blue}{\bigg\}}\stackrel{\eqref{h:13}}{=}\mathcal{O}\big(n^{-\frac{1}{3}}\big).
\end{equation*}
The remaining term, indicated in $\textcolor{red}{\{}\ldots\textcolor{red}{\}}$ in \eqref{h:9}, is more tedious and requires all six of \eqref{h:12},\eqref{h:13},\eqref{h:17},\eqref{h:18}, \eqref{h:19} and \eqref{h:20}. Still, since $\textcolor{red}{\{}\ldots\textcolor{red}{\}}$ is already of order $\frac{1}{n}$ we only need to focus on those terms that use \eqref{h:12},\eqref{h:17},\eqref{h:19} and \eqref{h:20}, because of the difference structure in $\frac{\partial}{\partial\theta}V_{\theta}(x)$. What results is
\begin{align*}
	\frac{n}{2\pi\im}\int_{\gamma}\textcolor{red}{\bigg\{}\ldots\textcolor{red}{\bigg\}}\big(V(z+s_{n\theta})-(z+s_{n\theta})^2\big)\frac{\d z}{n}=-\frac{2}{3}\bigg(2Q_2^{12}&+Q_1^{21}+\frac{s^3}{8}\bigg)\frac{\partial}{\partial\theta}\ln\big(h_{V_{\theta}}(\sqrt{2})\big)\\
	&-\frac{1}{24}\frac{\partial}{\partial\theta}\ln\big(h_{V_{\theta}}(-\sqrt{2})\big)+\mathcal{O}\big(n^{-\frac{2}{3}}\big),
\end{align*}
and combining the above we obtain \eqref{h:20}.
\end{proof}
Lastly, the combination of \eqref{h:1},\eqref{h:5} and \eqref{h:20} yields
\begin{cor} As $n\rightarrow\infty$,
\begin{align}
	&\frac{\partial}{\partial\theta}\ln D_n\big(s,\alpha;V_{\theta}(x)\big)=\frac{\partial}{\partial\theta}\Bigg[-n^2\int_E\frac{1}{2}\bigg(\theta^2 h_V(x)-(1-\theta)^2\frac{1}{\pi}\bigg)\sqrt{2-x^2}\,\big(V(x)-x^2\big)\,\d x+\frac{n\alpha}{2}V_{\theta}(\sqrt{2})\nonumber\\
	&-\frac{n\alpha}{2\pi}\int_E\frac{V_{\theta}(x)}{\sqrt{2-x^2}}\,\d x+\frac{3\alpha s}{2\sqrt{2}}\sqrt[3]{n}\frac{(\pi h_{V_{\theta}}(\sqrt{2}))^{\frac{1}{3}}}{\pi h_{V}(\sqrt{2})-1}+\bigg(\frac{s^3}{6}-sa-\frac{2}{3}\big(Q_1^{21}+2Q_2^{12}\big)\bigg)\ln h_{V_{\theta}}(\sqrt{2})\nonumber\\
	&-\frac{1}{24}\ln h_{V_{\theta}}(-\sqrt{2})\Bigg]+\mathcal{O}\big(n^{-\frac{1}{3}}\big),\label{h:22}
\end{align}
uniformly in $\theta\in[0,1]$ and $(s,\alpha)\in\mathbb{R}\times(-1,\infty)$ when chosen on compact sets.
\end{cor}
\begin{proof}[Proof of Theorem \ref{BS:3}] All that's left to do is integrate \eqref{h:22} from $\theta=0$ to $\theta=1$, exploiting uniformity of the error term on the way. What results is \eqref{e:26} and the special value of $\Theta_0(s,1)$ follows at once from Remark \ref{speccon}.
\end{proof}
\begin{appendix}
\section{The matching condition \eqref{a:25}}\label{tedious} Using \eqref{a:21}, condition $(4)$ in RHP \ref{YattRHP} and \eqref{a:22} one finds, as $n\rightarrow\infty$ with $s\in\mathbb{R},\alpha>-1,\beta\notin(-\infty,0)$,
\begin{align}
	M(z)=E^b(z)&\,\left\{I+\sum_{k=1}^2Q_k\Big(n^{\frac{2}{3}}\zeta^b(\sqrt{2}),\alpha\Big)\Big(n^{\frac{2}{3}}\big(\zeta^b(z)-\zeta^b(\sqrt{2})\big)\Big)^{-k}+\mathcal{O}\big(n^{-2}\big)\right\}E^b(z)^{-1}\nonumber\\
	&\hspace{1cm}\times P(z)\exp\left[-\varpi\Big(n^{\frac{2}{3}}\big(\zeta^b(z)-\zeta^b(\sqrt{2})\big),n^{\frac{2}{3}}\zeta^b(\sqrt{2})\Big)\sigma_3+\frac{2n}{3}\big(\zeta^b(z)\big)^{\frac{3}{2}}\sigma_3\right],\label{nasty1}
\end{align}
uniformly in $0<r_1\leq|z-\sqrt{2}|\leq r_2<\frac{\epsilon}{2}$ and $\theta\in[0,1]$, with fixed $r_1,r_2$. Observing the exponential factor in \eqref{nasty1} we begin with the following observation.
\begin{lem} Let $s\in\mathbb{R},\theta\in[0,1]$ and $n\geq n_0$ so $\sqrt{2}-s_{n\theta}\in\mathbb{D}_{\epsilon}(\sqrt{2})$ for $\epsilon>0$ small. Set
\begin{equation*}
	\vartheta(z)=\vartheta(z;s_{n\theta}):=\frac{2}{3}\Big\{\big(\zeta^b(z)-\zeta^b(\sqrt{2})\big)^{\frac{3}{2}}-\big(\zeta^b(z)\big)^{\frac{3}{2}}\Big\}+\zeta^b(\sqrt{2})\big(\zeta^b(z)-\zeta^b(\sqrt{2})\big)^{\frac{1}{2}},
\end{equation*}
then $z\mapsto\vartheta(z)$ is analytic for $z\in\mathbb{D}_{\epsilon}(\sqrt{2})\setminus(-\infty,\max\{\sqrt{2},\sqrt{2}-s_{n\theta}\}]$ and satisfies, as $n\rightarrow\infty$,
\begin{equation}\label{nasty2}
	\vartheta(z)=\sum_{k=2}^3\vartheta_k(z;\theta)s^kn^{-\frac{2k}{3}}+\mathcal{O}\big(n^{-\frac{8}{3}}\big)
\end{equation}
uniformly in $0<r_1\leq|z-\sqrt{2}|\leq r_2<\frac{\epsilon}{2},\theta\in[0,1]$ and pointwise in $s\in\mathbb{R}$. The coefficients $\vartheta_k(z;\theta)$ in \eqref{nasty2} are $n$-independent and equal
%\begin{align*}
%	\nu_1(z;\theta)=&\,-\frac{1}{2}\big(\zeta^b(z)\big)^{-\frac{1}{2}}\bigg[\frac{\partial}{\partial s}\zeta^b(b)\bigg]^2\Bigg|_{s=0}\hspace{1cm}\textcolor{red}{\tau=\pi h_{V_{\theta}}(b)2^{\frac{3}{4}}},\\
%	\nu_2(z;\theta)=&\,\frac{1}{4}\big(\zeta^b(z)\big)^{-\frac{3}{2}}\Bigg\{3\bigg[\frac{\partial}{\partial s}\zeta^b(z)\bigg]\bigg[\frac{\partial}{\partial s}\zeta^b(b)\bigg]^2-2\bigg[\frac{\partial}{\partial s}\zeta^b(b)\bigg]^3\Bigg\}-\frac{3}{2}\big(\zeta^b(z)\big)^{-\frac{1}{2}}\bigg[\frac{\partial}{\partial s}\zeta^b(b)\bigg]\bigg[\frac{\partial^2}{\partial s^2}\zeta^b(b)\bigg]\Bigg|_{s=0}.
%\end{align*}
\begin{align*}
	\vartheta_2(z;\theta)=&\,-\frac{1}{4}\big(\zeta^b(z)\big)^{-\frac{1}{2}}\bigg|_{s=0},\\
	\vartheta_3(z;\theta)=&\,-\frac{1}{12}\big(\zeta^b(z)\big)^{-\frac{3}{2}}\bigg|_{s=0}\Bigg\{1-\frac{3}{2}\frac{\pi h_{V_{\theta}}(z)}{\tau_{\theta}^{2/3}}\bigg(\frac{z^2-2}{\zeta^b(z)}\bigg)^{\frac{1}{2}}\Bigg|_{s=0}\Bigg\}
	-\frac{1}{5}\big(\zeta^b(z)\big)^{-\frac{1}{2}}\bigg|_{s=0}\bigg\{\frac{h_{V_{\theta}}'(b)}{h_{V_{\theta}}(b)}+\frac{1}{4\sqrt{2}}\bigg\}\tau_{\theta}^{-\frac{2}{3}}.%\ \ \ \ \hspace{1cm}\textcolor{red}{\tau=\pi h_{V_{\theta}}(b)2^{\frac{3}{4}}}
\end{align*}
\end{lem}
\begin{proof} We note that $\mathbb{R}\ni s\mapsto\vartheta(z;s_{n\theta})$ is real analytic so \eqref{nasty2} follows from straightforward, albeit somewhat tedious, Taylor expansion at $s=0$, using $\zeta^b(\sqrt{2})|_{s=0}=0$.
\end{proof}
The exponential factor in \eqref{nasty1} is to be conjugated with $P(z)$, so we are in need of the following result.
\begin{lem} Let $s\in\mathbb{R},\theta\in[0,1]$ and $n\geq n_0$ so $\sqrt{2}-s_{n\theta}\in\mathbb{D}_{\epsilon}(\sqrt{2})$ for $\epsilon>0$ small. Then, as $n\rightarrow\infty$,
\begin{equation}\label{nasty3}
	P(z)\e^{-n\vartheta(z)\sigma_3}P(z)^{-1}=I+\sum_{k=1}^3D_k(z;s,n,\theta)n^{-\frac{k}{3}}+\mathcal{O}\big(n^{-\frac{4}{3}}\big)
\end{equation}
uniformly in $0<r_1\leq|z-\sqrt{2}|\leq r_2<\frac{\epsilon}{2},\theta\in[0,1]$ and pointwise in $(s,\alpha)\in\mathbb{R}\times(-1,\infty)$.  The coefficients $D_k(z)=D_k(z;s,n,\theta),k\in\{1,2,3\}$ in \eqref{nasty3} are bounded in $n$ and they equal
\begin{align*}
	D_1(z)=&\,-\frac{s^2}{2}\vartheta_2(z;\theta)\chi^{\sigma_3}\begin{bmatrix}\nu(z)^2+\nu(z)^{-2} & \im\big(\nu(z)^2-\nu(z)^{-2}\big)\smallskip\\
	\im\big(\nu(z)^2-\nu(z)^{-2}\big) & -\nu(z)^2-\nu(z)^{-2}\end{bmatrix}\chi^{-\sigma_3},\\
	D_2(z)=&\,\frac{s^4}{2}\vartheta_2^2(z;\theta)\begin{bmatrix}1&0\\ 0&1\end{bmatrix},\\
	D_3(z)=&\,-\frac{s^3}{12}\Big(6\vartheta_3(z;\theta)+\vartheta_2^3(z;\theta)s^3\Big)\chi^{\sigma_3}\begin{bmatrix}\nu(z)^2+\nu(z)^{-2} & \im\big(\nu(z)^2-\nu(z)^{-2}\big)\smallskip\\
	\im\big(\nu(z)^2-\nu(z)^{-2}\big) & -\nu(z)^2-\nu(z)^{-2}\end{bmatrix}\chi^{-\sigma_3}.
\end{align*}
\end{lem}
\begin{proof} One substitutes \eqref{nasty2} into the left hand side of \eqref{nasty3} and keeps carefully track of the terms.
\end{proof}
Additionally, the below Laurent series expansion of $D_k(z)$ in \eqref{nasty3} turn out to be useful.
\begin{lem}\label{ugly} The coefficients $D_k(z)$ are meromorphic near $\omega=z-\sqrt{2}=0$ with Laurent expansions
\begin{align*}
	D_1&(z)=\frac{s^2}{4}\sqrt{1-\frac{\varkappa}{\sqrt{2}}}\frac{\pi h_{V_{\theta}}(\sqrt{2})}{\omega\tau_{\theta}^{4/3}}\sigma\Bigg\{\begin{bmatrix}1 & -1\\ 1 & -1\end{bmatrix}+\omega\textcolor{orange}{\Bigg(}\Bigg(\frac{1}{2(\sqrt{2}-\varkappa)}-\frac{1}{5}\bigg\{\frac{h_{V_{\theta}}'(\sqrt{2})}{h_{V_{\theta}}(\sqrt{2})}+\frac{1}{4\sqrt{2}}\bigg\}\Bigg)\begin{bmatrix}1 & -1\\ 1 & -1\end{bmatrix}\\
	&\hspace{1.4cm}+\frac{1}{\sqrt{2}-\varkappa}\,\begin{bmatrix}1 & 1\\ -1 & -1\end{bmatrix}\textcolor{orange}{\Bigg)}+\omega^2\textcolor{blue}{\Bigg(}-\frac{1}{8(\sqrt{2}-\varkappa)^2}\begin{bmatrix}1&-1\\ 1&-1\end{bmatrix}-\frac{1}{2(\sqrt{2}-\varkappa)^2}\begin{bmatrix}1&1\\ -1&-1\end{bmatrix}\\
	&\hspace{1.4cm}+\Bigg(-\frac{1}{7}\bigg\{\frac{h_{V_{\theta}}''(\sqrt{2})}{2h_{V_{\theta}}(\sqrt{2})}+\frac{1}{4\sqrt{2}}\frac{h_{V_{\theta}}'(\sqrt{2})}{h_{V_{\theta}}(\sqrt{2})}-\frac{1}{64}\bigg\}+\frac{2}{25}\bigg\{\frac{h_{V_{\theta}}'(\sqrt{2})}{h_{V_{\theta}}(\sqrt{2})}+\frac{1}{4\sqrt{2}}\bigg\}^2\Bigg)\begin{bmatrix}1&-1\\ 1&-1\end{bmatrix}\\
	&\hspace{1.4cm}-\frac{1}{5}\bigg\{\frac{h_{V_{\theta}}'(\sqrt{2})}{h_{V_{\theta}}(\sqrt{2})}+\frac{1}{4\sqrt{2}}\bigg\}\Bigg(\frac{1}{2(\sqrt{2}-\varkappa)}\begin{bmatrix}1&-1\\ 1&-1\end{bmatrix}+\frac{1}{\sqrt{2}-\varkappa}\begin{bmatrix}1&1\\ -1&-1\end{bmatrix}\Bigg)\textcolor{blue}{\Bigg)}+\mathcal{O}\big(\omega^3\big)\Bigg\}\sigma^{-1};\\
	D_2&(z)=\frac{s^4}{8\sqrt{2}}\frac{(\pi h_{V_{\theta}}(\sqrt{2}))^2}{\omega\tau_{\theta}^{8/3}}\Bigg\{\begin{bmatrix}1 & 0\\ 0 & 1\end{bmatrix}-\frac{2\omega}{5}\bigg\{\frac{h_{V_{\theta}}'(\sqrt{2})}{h_{V_{\theta}}(\sqrt{2})}+\frac{1}{4\sqrt{2}}\bigg\}\begin{bmatrix}1 & 0\\ 0 & 1\end{bmatrix}-\omega^2\Bigg(\frac{2}{7}\bigg\{\frac{h_{V_{\theta}}''(\sqrt{2})}{2h_{V_{\theta}}(\sqrt{2})}\\
	&\hspace{1.4cm}+\frac{1}{4\sqrt{2}}\frac{h_{V_{\theta}}'(\sqrt{2})}{h_{V_{\theta}}(\sqrt{2})}-\frac{1}{64}\bigg\}-\frac{1}{5}\bigg\{\frac{h_{V_{\theta}}'(\sqrt{2})}{h_{V_{\theta}}(\sqrt{2})}+\frac{1}{4\sqrt{2}}\bigg\}^2\Bigg)\begin{bmatrix}1&0\\0&1\end{bmatrix}+\mathcal{O}\big(\omega^3\big)\Bigg\};\\
	D_3&(z)=-\frac{s^3}{24}\sqrt{1-\frac{\varkappa}{\sqrt{2}}}\frac{\pi h_{V_{\theta}}(\sqrt{2})}{(\omega\tau_{\theta})^2}\textcolor{magenta}{\Bigg\{}1-3\bigg\{\frac{h_{V_{\theta}}'(\sqrt{2})}{h_{V_{\theta}}(\sqrt{2})}+\frac{1}{4\sqrt{2}}\bigg\}\omega-\frac{s^3}{16}\textcolor{blue}{\Bigg(}1-\frac{3}{5}\bigg\{\frac{h_{V_{\theta}}'(\sqrt{2})}{h_{V_{\theta}}(\sqrt{2})}+\frac{1}{4\sqrt{2}}\bigg\}\omega\textcolor{blue}{\Bigg)}\\
	&+\mathcal{O}\big(\omega^2\big)\textcolor{magenta}{\Bigg\}}\sigma\Bigg\{\begin{bmatrix}1&-1\\ 1&-1\end{bmatrix}+\omega\textcolor{orange}{\Bigg(}\frac{1}{2(\sqrt{2}-\varkappa)}\begin{bmatrix}1&-1\\ 1&-1\end{bmatrix}+\frac{1}{\sqrt{2}-\varkappa}\begin{bmatrix}1&1\\ -1&-1\end{bmatrix}\textcolor{orange}{\Bigg)}+\mathcal{O}\big(\omega^2\big)\Bigg\}\sigma^{-1},
\end{align*}
that hold for $0<|\omega|<\epsilon$ with $\epsilon>0$ sufficiently small. Here, $\sigma=\chi^{\sigma_3}\e^{\im\frac{\pi}{4}\sigma_3}$ and $\varkappa=-\sqrt{2}-s_{n\theta}$.

\end{lem}

\section{The Painlev\'e-XXXIV model problem}\label{appPain}
The following RHP appeared seemingly first in \cite[Section $1.2$]{IKO}, up to a trivial scaling transformation. Further occurrences of similar model problems are in \cite[page $77$]{XY}, in \cite[Section $2$]{WXZ} and in \cite[Section $4.3$]{Y}. 

\begin{problem}\label{YattRHP} Let $x\in\mathbb{R},\alpha>-1$ and $\beta\in\mathbb{C}\setminus(-\infty,0)$. Find $Q(\zeta)=Q(\zeta;x,\alpha,\beta)\in\mathbb{C}^{2\times 2}$ with the following properties:
\begin{enumerate}
	\item[(1)] $\zeta\mapsto Q(\zeta)$ is analytic for $\zeta\in\mathbb{C}\setminus\Sigma_Q$ where $\Sigma_Q=\bigcup_{j=1}^4\Sigma_j$ with
	\begin{equation*}
		\Sigma_1:=(0,\infty)\subset\mathbb{R},\ \ \ \ \ \Sigma_2:=\e^{-\im\frac{\pi}{3}}(-\infty,0),\ \ \ \ \ \Sigma_3:=(-\infty,0)\subset\mathbb{R},\ \ \ \ \ \Sigma_4:=\e^{\im\frac{\pi}{3}}(-\infty,0),
	\end{equation*}
	is the oriented contour shown in Figure \ref{fig2}. On $\Sigma_Q\setminus\{0\}$, $Q(\zeta)$ admits continuous limiting values as we approach $\Sigma_Q\setminus\{0\}$ from either side of $\mathbb{C}\setminus\Sigma_Q$.
		\begin{figure}[tbh]
	\begin{tikzpicture}[xscale=0.9,yscale=0.9]
	\draw [thick, color=red, decoration={markings, mark=at position 0.25 with {\arrow{>}}}, decoration={markings, mark=at position 0.75 with {\arrow{>}}}, postaction={decorate}] (-3,0) -- (3,0);
	\draw [thick, color=red, decoration={markings, mark=at position 0.5 with {\arrow{>}}}, postaction={decorate}] (-1.5,2.598076212) -- (0,0);
	\draw [thick, color=red, decoration={markings, mark=at position 0.5 with {\arrow{>}}}, postaction={decorate}] (-1.5,-2.598076212) -- (0,0);
	\node [below] at (2.8,-0.15) {{\footnotesize $\Sigma_1$}};
	\node [below] at (-2.8,-0.15) {{\footnotesize $\Sigma_3$}};
	\node [right] at (-1.2,2.4) {{\footnotesize $\Sigma_2$}};
	\node [right] at (-1.2,-2.4) {{\footnotesize $\Sigma_4$}};
	\node [right] at (0.8,1.4) {{\footnotesize $\Omega_1$}};
	\node [right] at (0.8,-1.4) {{\footnotesize $\Omega_4$}};
	\node [right] at (-2.2,1.1) {{\footnotesize $\Omega_2$}};
	\node [right] at (-2.2,-1.1) {{\footnotesize $\Omega_3$}};
	\draw [fill, color=black] (-0.02,0) circle [radius=0.06];
\end{tikzpicture}
\caption{The oriented jump contour $\Sigma_Q$, shown in red, for the model function $Q(\zeta)$ in the complex $\zeta$-plane.}
\label{fig2}
\end{figure}
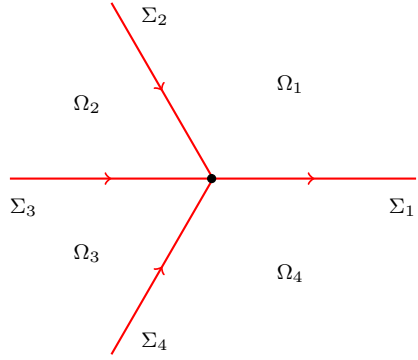
	\item[(2)] The limiting values $Q_{\pm}(\zeta)$ on $\Sigma_Q\setminus\{0\}$ satisfy $Q_+(\zeta)=Q_-(\zeta)G_Q(\zeta)$ with $G_Q(\zeta)=G_Q(\zeta;\alpha,\beta)$ given by
	\begin{equation*}
		G_Q(\zeta)=\begin{bmatrix}1&\beta\\ 0&1\end{bmatrix},\ \ \zeta\in\Sigma_1;\ \ \ \ G_Q(\zeta)=\begin{bmatrix}0&1\\ -1&0\end{bmatrix},\ \ \zeta\in\Sigma_3
	\end{equation*}
	\begin{equation*}
		G_Q(\zeta)=\begin{bmatrix}1&0\\ \e^{\im\pi\alpha}&1\end{bmatrix},\ \ \zeta\in\Sigma_2;\ \ \ \ G_Q(\zeta)=\begin{bmatrix}1&0\\ \e^{-\im\pi\alpha}&1\end{bmatrix},\ \ \zeta\in\Sigma_4
	\end{equation*}
	\item[(3)] Near $\zeta=0$, $\zeta\mapsto Q(\zeta)$ is weakly singular in that, with some $\zeta\mapsto\widehat{Q}(\zeta)=\widehat{Q}(\zeta;x,\alpha,\beta)$ analytic at $\zeta=0$ and non-vanishing, $Q(\zeta)$ is of the form
	\begin{equation*}
		Q(\zeta)=\widehat{Q}(\zeta)S(\zeta)\mathcal{M}_j,\ \ \ \ \zeta\in\big(\Omega_j\cap\mathbb{D}_{\epsilon}(0)\big)\setminus\{0\}.
	\end{equation*}
	Here, in case $\alpha\notin\mathbb{Z}_{>-1}$, with principal branches cut on $(-\infty,0]\subset\mathbb{R}$ throughout,
	\begin{equation}\label{B:0}
		S(\zeta)=\zeta^{\frac{\alpha}{2}\sigma_3},\ \ \ \ \ \mathcal{M}_2=\begin{bmatrix}\displaystyle\frac{1}{2\cos(\frac{\pi}{2}\alpha)}\frac{1-\beta\e^{\im\pi\alpha}}{1-\e^{\im\pi\alpha}}& \displaystyle\frac{1}{2\cos(\frac{\pi\alpha}{2})}\frac{\beta-\e^{\im\pi\alpha}}{1-\e^{\im\pi\alpha}} \smallskip\\ \displaystyle -\e^{\im\frac{\pi}{2}\alpha} & \e^{-\im\frac{\pi}{2}\alpha}\end{bmatrix},
	\end{equation}
	followed by
	\begin{equation}\label{B:1}
		\mathcal{M}_1=\mathcal{M}_2\begin{bmatrix}1&0\\ \e^{\im\pi\alpha}&1\end{bmatrix},\hspace{1cm} \mathcal{M}_4=\mathcal{M}_1\begin{bmatrix}1&-\beta\\ 0&1\end{bmatrix},\hspace{1cm} \mathcal{M}_3=\mathcal{M}_4\begin{bmatrix}1&0\\ \e^{-\im\pi\alpha}&1\end{bmatrix}.
	\end{equation}
	If $\alpha\in\mathbb{Z}_{>-1}$, then \eqref{B:1} remain formally unchanged however \eqref{B:0} changes to
	\begin{equation*}
		S(\zeta)=\zeta^{\frac{\alpha}{2}\sigma_3}\begin{bmatrix}1&\frac{1-\beta}{2\pi\im}\ln\zeta\\ 0 & 1\end{bmatrix},\ \ \ \ \ \mathcal{M}_2=\begin{bmatrix}1/2&1/2\smallskip\\ -1&1\end{bmatrix}\e^{\im\frac{\pi}{2}\alpha\sigma_3}\ \ \ \textnormal{if}\ \alpha\equiv 0\mod 2;
	\end{equation*}
	\begin{equation*}
		S(\zeta)=\zeta^{\frac{\alpha}{2}\sigma_3}\begin{bmatrix}1&\frac{1+\beta}{2\pi\im}\ln\zeta\\ 0&1\end{bmatrix},\ \ \ \ \ \mathcal{M}_2=\begin{bmatrix}0&1\smallskip\\  -1& 1\end{bmatrix}\e^{\im\frac{\pi}{2}\alpha\sigma_3}\ \ \ \textnormal{if}\ \alpha\equiv 1\mod 2.
	\end{equation*}
	\item[(4)] As $\zeta\rightarrow\infty$ and $\zeta\notin\Sigma_Q$, $Q(\zeta)$ is normalized as follows,
	\begin{equation*}
		Q(\zeta)=\Bigg\{I+\sum_{k=1}^2Q_k(x,\alpha,\beta)\zeta^{-k}+\mathcal{O}\big(\zeta^{-3}\big)\Bigg\}\zeta^{-\frac{1}{4}\sigma_3}\frac{1}{\sqrt{2}}\begin{bmatrix}1&1\\ -1&1\end{bmatrix}\e^{-\im\frac{\pi}{4}\sigma_3}\e^{-\varpi(\zeta;x)\sigma_3},
	\end{equation*}
	with $Q_k(x,\alpha,\beta)$ independent of $\zeta$, the function $\varpi(\zeta;x):=\frac{2}{3}\zeta^{\frac{3}{2}}+x\zeta^{\frac{1}{2}}$ defined on $\mathbb{C}\setminus(-\infty,0]$, and principal branches used throughout.
\end{enumerate}
\end{problem}
The model problem \ref{YattRHP}, for the chosen parameters $(x,\alpha,\beta)$, is uniquely solvable, cf. \cite[Theorem $4.1$]{Y} and we will be interested in the asymptotic behavior of its solution $Q(\zeta;x,\alpha,\beta)$ as $x\rightarrow+\infty$ and in its connection to Painlev\'e special function theory. Both topics have been studied before and we will be brief in our presentation below:
\subsection{Connection to Painlev\'e special function theory} Seeing that $G_Q(\zeta)$ in condition $(2)$ of RHP \ref{YattRHP} is piecewise constant and independent of $x$, an application of Liouville's theorem yields the below Lax system, see also \cite[Proposition $2$]{WXZ}.

\begin{lem}\label{ArnoLax} Suppose $Q(\zeta)=Q(\zeta;x,\alpha,\beta)\in\mathbb{C}^{2\times 2}$ solves RHP \ref{YattRHP} for $x\in\mathbb{R},\alpha>-1,\beta\notin(-\infty,0)$. Then
\begin{align*}
	\frac{\partial Q}{\partial\zeta}(\zeta)=&\,\,\left\{\zeta\begin{bmatrix}0 & 0\\ 1 & 0\end{bmatrix}+\begin{bmatrix}a & 1\\ \frac{x}{2}+b & -a\end{bmatrix}+\frac{1}{\zeta}\begin{bmatrix}p & q\\ r & -p\end{bmatrix}\right\}Q(\zeta)\equiv A(\zeta)Q(\zeta),\\
	\frac{\partial Q}{\partial x}(\zeta)=&\,\,\left\{\zeta\begin{bmatrix}0 & 0\\ 1 & 0\end{bmatrix}+\begin{bmatrix}a & 1\\ b & -a\end{bmatrix}\right\}Q(\zeta)\equiv B(\zeta)Q(\zeta),
\end{align*}
where $a,b,p,q,r$ are functions of $(x,\alpha,\beta)$ given in terms of the data in RHP \ref{YattRHP},
\begin{align*}
	a:=Q_1^{12},\ \ \ b:=&\,\,Q_1^{22}-Q_1^{11},\ \ \ p:=\frac{\alpha}{2}\Big(\widehat{Q}^{11}(0)\widehat{Q}^{22}(0)+\widehat{Q}^{12}(0)\widehat{Q}^{21}(0)\Big),\\
	q:=&\,-\alpha\widehat{Q}^{11}(0)\widehat{Q}^{12}(0),\ \ \ \ \ \ \ \,r:=\alpha\widehat{Q}^{21}(0)\widehat{Q}^{22}(0).
\end{align*}
Those functions satisfy the constraints
\begin{equation*}
	q=\frac{x}{2}-b-a^2,\ \ \ \ \ \frac{\d a}{\d x}=a^2+b,\ \ \ \ \ 4\big(p^2+rq\big)=\alpha^2,
\end{equation*}
and the compatibility of $\frac{\partial Q}{\partial\zeta}=AQ,\frac{\partial Q}{\partial x}=BQ$ yields the coupled ODE system
\begin{equation*}
	\frac{\d a}{\d x}=\frac{x}{2}-q,\ \ \ \frac{\d b}{\d x}=\frac{1}{2}-xa+2p,\ \ \ \frac{\d p}{\d x}=r-qb,\ \ \ \frac{\d q}{\d x}=2(aq-p),\ \ \ \frac{\d r}{\d x}=2(bp-ar).
\end{equation*}
In particular, $q$ solves the Painlev\'e-XXXIV equation
\begin{equation}\label{B:2}
	\frac{\d^2q}{\d x^2}=\frac{1}{2q}\bigg(\frac{\d q}{\d x}\bigg)^2+2xq-4q^2-\frac{\alpha^2}{2q},\ \ \ \ x\in\mathbb{R},
\end{equation}
and $\sigma=\sigma(x,\alpha,\beta)$ defined as $\sigma:=\frac{1}{4}x^2-a$ solves the Jimbo-Miwa-Okamoto sigma-Painlev\'e-II equation
\begin{equation}\label{B:3}
	\bigg(\frac{\d^2\sigma}{\d x^2}\bigg)^2+4\frac{\d\sigma}{\d x}\Bigg(\bigg(\frac{\d\sigma}{\d x}\bigg)^2-x\frac{\d\sigma}{\d x}+\sigma\Bigg)=\alpha^2,\ \ \ \ x\in\mathbb{R}.
\end{equation}
\end{lem}	
\subsection{Large $x$-asymptotics} The same asymptotics have been partially worked out in \cite[Theorem $1.2$]{IKO2}, with further improvements found in \cite[Theorem $1$]{WXZ} and in \cite[$(4.2),(4.3)$]{Y}. None of those three references precisely match our needs for RHP \ref{YattRHP}, still the overall nonlinear steepest descent methodology for deriving large $x$-asymptotics in the same problem is well-known and we only summarize results:
\begin{lem}\label{ArnoAsy} Suppose $Q(\zeta)=Q(\zeta;x,\alpha,\beta)\in\mathbb{C}^{2\times 2}$ solves RHP \ref{YattRHP} for $x\in\mathbb{R},\alpha>-1,\beta\notin(-\infty,0)$. Then, as $x\rightarrow+\infty$, uniformly in $(a,\beta)\in(-1,\infty)\times(\mathbb{C}\setminus(-\infty,0))$ on compact sets,
\begin{align}
	Q_1^{12}(x,\alpha,\beta)=\frac{x^2}{4}+\alpha\sqrt{x}\,\,+&\,\sum_{k=1}^ma_kx^{-\frac{1}{2}(3k-1)}+\mathcal{O}\big(x^{-\frac{1}{2}(3m+2)}\big)\nonumber\\
	-&\,\,\big(\e^{\im\pi\alpha}-\beta\big)\frac{\Gamma(1+\alpha)}{2^{3(1+\alpha)}\pi}x^{-\frac{1}{2}(2+3\alpha)}\e^{-\frac{4}{3}x^{\frac{3}{2}}}\Big(1+\mathcal{O}\big(x^{-\frac{1}{4}}\big)\Big)\label{B:4}
\end{align}
with $m\in\mathbb{N}$ and coefficients 
\begin{equation*}
	a_1=\frac{\alpha^2}{4},\ \ \ a_2=-\frac{\alpha}{32}(1+4\alpha^2),\ \ \ a_3=\frac{\alpha^2}{64}(7+8\alpha^2).
\end{equation*}
The asymptotic \eqref{B:4} is differentiable with respect to $x$.
\end{lem}
The combination of Lemma \ref{ArnoLax} and \ref{ArnoAsy} yields the leading order large $x$-asymptotics for $a,q$ and $\sigma$.
\begin{cor} As $x\rightarrow+\infty$, uniformly in $(a,\beta)\in(-1,\infty)\times(\mathbb{C}\setminus(-\infty,0))$ on compact sets,
\begin{equation*}
	a(x,\alpha,\beta)=\frac{x^2}{4}+\alpha\sqrt{x}+\frac{\alpha^2}{4x}+\mathcal{O}\big(x^{-\frac{5}{2}}\big)-\big(\e^{\im\pi\alpha}-\beta\big)\frac{\Gamma(1+\alpha)}{2^{3(1+\alpha)}\pi}x^{-1-\frac{3}{2}\alpha}\e^{-\frac{4}{3}x^{\frac{3}{2}}}\Big(1+\mathcal{O}\big(x^{-\frac{1}{4}}\big)\Big),
\end{equation*}
as well as
\begin{equation*}
	q(x,\alpha,\beta)=-\frac{\alpha}{2\sqrt{x}}+\frac{\alpha^2}{4x^2}+\mathcal{O}\big(x^{-\frac{7}{2}}\big)-\big(\e^{\im\pi\alpha}-\beta\big)\frac{\Gamma(1+\alpha)}{2^{2+3\alpha}\pi}x^{-\frac{1}{2}(1+3\alpha)}\e^{-\frac{4}{3}x^{\frac{3}{2}}}\Big(1+\mathcal{O}\big(x^{-\frac{3}{2}}\big)\Big),
\end{equation*}
and lastly
\begin{equation*}
	\sigma(x,\alpha,\beta)=-\alpha\sqrt{x}-\frac{\alpha^2}{4x}+\mathcal{O}\big(x^{-\frac{5}{2}}\big)+\big(\e^{\im\pi\alpha}-\beta\big)\frac{\Gamma(1+\alpha)}{2^{3(1+\alpha)}\pi}x^{-1-\frac{3}{2}\alpha}\e^{-\frac{4}{3}x^{\frac{3}{2}}}\Big(1+\mathcal{O}\big(x^{-\frac{1}{4}}\big)\Big).
\end{equation*}
\end{cor}

\section{The factorized form \eqref{e:8}}\label{facderiv}
The formula for the joint density of the eigenvalues in the invariant model \eqref{e:1}, cf. \cite[$(4.28)$]{PS}, yields
\begin{equation}\label{rev1}
	E_n[\phi;\lambda,\alpha,\beta]=\frac{1}{Q_n}\int_{\mathbb{R}^n}|\Delta(x_1,\ldots,x_n)|^2\prod_{j=1}^n\big(1-\phi(x_j)\big)\omega_{\alpha\beta}(x_j-\lambda)\e^{-nV(x_j)}\d x_j,
\end{equation}
in terms of the Vandermonde determinant $\Delta(x_1,\ldots,x_n)$, and in terms of the normalization constant (the denominator in \eqref{e:8})
\begin{equation*}
	Q_n:=\int_{\mathbb{R}^n}|\Delta(x_1,\ldots,x_n)|^2\prod_{j=1}^n\e^{-nV(x_j)}\d x_j=n!\,D_n\big(\lambda,0,1;V(x)\big).
\end{equation*}
Using the orthonormal polynomials $\{\pi_{j,n}\}_{j=0}^{\infty}\subset\mathbb{C}[x]$ in \eqref{e:9}, the Vandermonde factor in \eqref{rev1} equals
\begin{equation*}
	\Delta(x_1,\ldots,x_n)=\frac{1}{\prod_{\ell=0}^{n-1}\varsigma_{\ell,n}}\det\big[\pi_{j-1,n}(x_k)\big]_{j,k=1}^n,\ \ \ \ \pi_{k,n}(x)=\varsigma_{k,n}\Big(x^n+\mathcal{O}\big(x^{n-1}\big)\Big),
\end{equation*}
and we thus obtain
\begin{equation}\label{rev2}
	|\Delta(x_1,\ldots,x_n)|^2\prod_{j=1}^n\omega_{\alpha\beta}(x_j-\lambda)\e^{-nV(x_j)}=\frac{1}{\prod_{\ell=0}^{n-1}\varsigma_{\ell,n}^2}\Big(\det\big[\psi_{j-1,n}(x_k)\big]_{j,k=1}^n\Big)^2,
\end{equation}
with $\psi_{j,n}(x):=\omega_{\alpha\beta}^{\frac{1}{2}}(x-\lambda)\e^{-\frac{n}{2}V(x)}\pi_{j,n}(x)$. Consequently, by Andr\'eief's identity,
\begin{align*}
	\eqref{rev1}=&\,\,\frac{1}{Q_n\prod_{\ell=0}^{n-1}\varsigma_{\ell,n}^2}\int_{\mathbb{R}^n}\bigg[\prod_{j=1}^n\big(1-\phi(x_j)\big)\bigg]\Big(\det\big[\psi_{j-1,n}(x_k)\big]_{j,k=1}^n\Big)^2\prod_{\ell=1}^n\d x_{\ell}\\
	=&\,\,\frac{n!}{Q_n\prod_{\ell=0}^{n-1}\varsigma_{\ell,n}^2}\det\bigg[\delta_{jk}-\int_J\phi(x)\psi_{j-1,n}(x)\psi_{k-1,n}(x)\d x\bigg]_{j,k=1}^n,
\end{align*}
where the remaining $n\times n$ determinant is precisely the Fredholm determinant $F_n[\phi;\lambda,\alpha,\beta;V(x)]$ in \eqref{e:8}, cf. \cite[page $111$]{PS}. On the other hand, 
\begin{equation*}
	n!\,D_n\big(\lambda,\alpha,\beta;V(x)\big)\stackrel{\eqref{rev2}}{=}\frac{1}{\prod_{\ell=0}^{n-1}\varsigma_{\ell,n}^2}\int_{\mathbb{R}^n}\Big(\det\big[\psi_{j-1,n}(x_k)\big]_{j,k=1}^n\Big)^2\prod_{\ell=1}^n\d x_{\ell}=\frac{n!}{\prod_{\ell=0}^{n-1}\varsigma_{\ell,n}^2}
\end{equation*}
by orthonormality and Andr\'eief again. The derivation of \eqref{e:8} is complete.
\end{appendix}
%%%%%%%%%%%%%%%%%%%%%%%%%%%%%%%%%%%%%%%%%%%%%%%%%%%%%%%%%%%%%%

% % % % % % % % % % % % % % % % % 
% % % % % % % % % % % % % % % % % 
%%%%%%%%%%%%%%%%%%%%%%%%%%%%%%%%%%%%%%%%%%%%%%%%%%%%%%%%%%%%%%
\section*{Acknowledgements}
T.B. is grateful to Peter Forrester for drawing his attention to the problem at hand. This work is dedicated, with admiration, to Arno Kuijlaars on the occasion of his $60$th birthday. We are deeply indebted to Arno for his trail-blazing contributions to the modern theory of Riemann-Hilbert problems and their asymptotic analysis over the past $25$ years.
%%%%%%%%%%%%%%%%%%%%%%%%%%%%%%%%%%%%%%%%%%%%%%%%%%%%%%%%%%%%%%

\end{document}